\documentclass[sigconf,nonacm]{acmart}

\usepackage{booktabs}
\usepackage{graphicx}
\usepackage{amsmath}

\usepackage{amssymb}
\usepackage{algorithm}
\usepackage{algpseudocode}
\usepackage{subcaption}
\usepackage{xcolor}
\usepackage{balance}
\usepackage{microtype}
\usepackage{url}
\usepackage{enumitem}           
\usepackage{multirow}           
\usepackage{colortbl}           
\usepackage{tcolorbox}          
\tcbuselibrary{listings,skins,breakable}

\definecolor{codebg}{HTML}{F7F7F8}
\definecolor{codeframe}{HTML}{D0D0D6}
\definecolor{codecomment}{HTML}{6A737D}
\definecolor{codekeyword}{HTML}{D73A49}
\definecolor{codestring}{HTML}{032F62}
\usepackage{fancyvrb}
\DefineVerbatimEnvironment{codeblock}{Verbatim}{%
  fontsize=\footnotesize,
  frame=single,
  framesep=4pt,
  rulecolor=\color{codeframe},
  formatcom=\color{black},
  baselinestretch=0.95
}

\newtcolorbox{finding}[1][]{%
  enhanced,
  breakable,
  sharp corners,
  boxrule=0pt,
  leftrule=2.5pt,
  colback=blue!3,
  colframe=blue!50!black,
  fonttitle=\bfseries\small,
  title={#1},
  top=3pt, bottom=3pt,
}

\graphicspath{{./figs/}}

\newcommand{\sysname}{\textsc{AgentIR}}
\newcommand{\speedup}[1]{#1$\times$}

\begin{document}
\raggedbottom  

\title{AgentIR: A Workload-Adaptive Cascade Retrieval
Substrate\\for Long-Term Conversational Memory}

\author{Aojie Yuan}
\affiliation{\institution{University of Southern California}\city{Los Angeles}\state{CA}\country{USA}}
\email{aojieyua@usc.edu}

\author{Haiyue Zhang}
\affiliation{\institution{University of Southern California}\city{Los Angeles}\state{CA}\country{USA}}
\email{haiyuezh@usc.edu}

\author{Shahin Nazarian}
\affiliation{\institution{University of Southern California}\city{Los Angeles}\state{CA}\country{USA}}
\email{shahin.nazarian@usc.edu}

\begin{abstract}
Long-term conversational memory is a retrieval workload classical
IR was not built for: the index grows during the query stream,
query types shift intra-session, and the latency budget per
retrieval is sub-$10$\,ms.  Lucene-class engines treat the index
as static and the query as stateless, leaving the workload's
structure unexploited.

\sysname{} treats fusion as a per-query decision along two axes:
\emph{which} fusion to apply (BM25, Dense, RRF, or agent-aware
RRF), and \emph{whether} the $\sim$$52$\,ms dense channel is worth
running at all.  The second axis is a confidence-triggered
\textbf{cascade router} that decides from the BM25 top-$k$ margin
alone and re-tunes across workloads without retraining.  On
LongMemEval ($n{=}500$), where the dense channel does add
information, the cascade skips $63$\% of queries at parity
LLM-judged accuracy ($\mathbf{2.67\times}$ faster under two judges,
paired bootstrap $p{\geq}0.88$); per-qtype thresholds extend this
to $\mathbf{5.76\times}$ under 5-fold cross-validation.  On LoCoMo
($n{=}1{,}982$), where BM25 alone is already the strongest single
system, the \emph{same} trigger auto-tunes to a $100$\% skip rate
($\mathbf{132\times}$ faster, $+0.089$ Hit@5).  Capacity on a
shared 8-core VM rises from $\sim$$154$ to $\sim$$1{,}400$
concurrent agents ($\mathbf{9\times}$).

Underneath the cascade, a time-partitioned index does
$O(\log 1/\varepsilon)$ work \emph{independent of corpus size}:
$1234\times$ corpus growth costs only $3.6\times$ latency, ending
in $\mathbf{1769\times}$ over sequential at sub-$100$\,$\mu$s
$p50$ on 5\,M records.  At parity quality with Lucene on 9 BEIR
datasets up to 8.8\,M docs, the substrate runs $\mathbf{10\times}$
geo-mean over Pyserini 8T and $\mathbf{11\times}$ over PISA-1T
BlockMax-WAND; an A100 reaches $1.8$--$39\times$ over Pyserini 8T;
chunked index build sustains $56.8$\,K docs/sec on MS\,MARCO.
Three subtle BM25/GPU correctness pitfalls that silently regress
nDCG@10 by $6$--$8\times$ are documented and fixed; post-fix CPU
and GPU agree within $0.0002$ nDCG@10 on all eight datasets that
fit a single A100.
\end{abstract}



\maketitle

\begin{figure*}[t]
\centering
\includegraphics[width=0.92\textwidth]{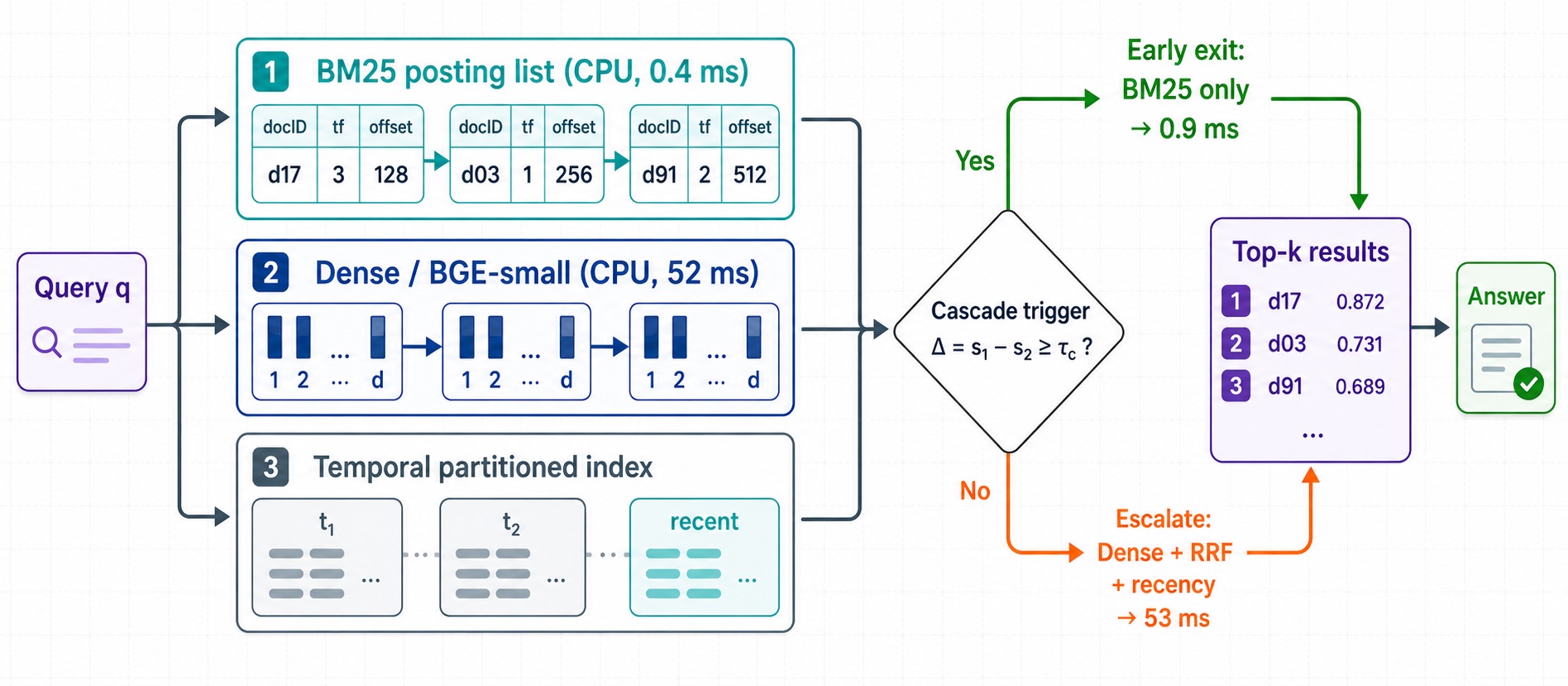}
\caption{\sysname{} pipeline and two-axis adaptive control
  surface.  Three substrate stages run concurrently: a SIMD-vec\-tor\-ized
  BM25 posting list (CPU, $0.4$\,ms), a Dense / BGE-small channel
  (CPU, $52$\,ms), and a time-partitioned temporal index.  The
  \emph{cascade trigger} ($\Delta{=}s_1{-}s_2{\geq}\tau_c$,
  \S\ref{sec:longmemeval}) inspects the BM25 top-$k$ margin: confident
  queries early-exit in $0.9$\,ms; ambiguous queries escalate to
  Dense$+$RRF$+$recency for $53$\,ms.  Same trigger, two workloads,
  two operating points: on LongMemEval it skips $63$\%
  ($\mathbf{2.67\times}$); on LoCoMo it skips $100$\%
  ($\mathbf{132\times}$).}
\label{fig:arch}
\end{figure*}

\section{Introduction}
\label{sec:intro}

LLM agents running for hours or days accumulate large interaction
histories: dialogue turns, tool calls, intermediate plans, and the
results those plans produced.  This \emph{agent memory} is queried
by the agent itself on almost every step, and the cost of those
queries sits directly inside the agent's reasoning
loop~\cite{park2023generative,shinn2023reflexion,sumers2024coala}.

A representative load looks like this.  A customer-support agent
on its 800th turn of a three-month engagement holds roughly 5\,M
records of memory.  On each step it issues several lookups against
that memory: the last similar issue on this account, the last
successful fix for the same error code across \emph{any} account,
the user's name and stated preferences from session~1.  Each of these is a retrieval call
that lives \emph{inside} the agent's reasoning loop.  Production
agent traces show 5--20 such retrievals per agent step; if any one
of them blocks for 50\,ms the user feels the agent stall.  The
retrieval system underneath the agent must therefore deliver
\textbf{sub-millisecond $p99$ latency on millions of records, on
commodity hardware, while the corpus grows during the session}.

The IR systems agents currently inherit---Lucene, FAISS,
PISA---were not built for this regime.  Lucene-class engines treat
the corpus as static, ignore the recency skew, and serve queries
at 2--10\,ms latency.  FAISS-class engines retrieve by dense
similarity but lack BM25 and agent-specific metadata fusion.
PISA's BlockMax-WAND degenerates on the long, multi-clause
queries that agents emit when they ask their own memory natural
questions.  Even at parity quality, none of these systems is
positioned to exploit what makes agent memory \emph{different} from
general-purpose IR:

\begin{enumerate}[nosep,leftmargin=*]
\item \textbf{Strong recency bias.}  Agent retrievals concentrate
  in the most recent fraction of records---far more sharply than
  web search or document retrieval.  On LongMemEval, the median
  normalized rank of gold sessions for temporal/knowledge/multi-session
  questions is 0.20--0.27 (\S\ref{sec:scaling}); a stress
  parameterization of 80/20 is comparable in steepness.
\item \textbf{Role-structured records.}  Each record carries
  metadata: role (user, assistant, tool call/output, system,
  planning), session~ID, agent~ID, timestamps, and tool type.
\item \textbf{Mixed modalities.}  Agent memory contains natural
  language, code, API responses, and structured tool
  outputs---requiring both lexical and semantic matching.
\item \textbf{Continuous growth.}  Agent memory grows
  monotonically, demanding sub-linear retrieval scaling to maintain
  interactive latency.
\item \textbf{Heterogeneous query types.}  The same agent fires
  factual lookups (``what's this user's name?''), reasoning
  queries (``did we already try restarting?''), and temporal
  queries (``when did this error first appear?'').  No single
  fusion strategy wins on all three.
\end{enumerate}

These properties motivate a \emph{hybrid} architecture fusing
sparse lexical search
(BM25~\cite{robertson1994okapi,robertson2009bm25}) with dense
approximate nearest neighbor (ANN)
search~\cite{malkov2020hnsw}, augmented with agent-specific
metadata filtering, temporal awareness, and a runtime selector
that picks the fusion per query type rather than committing to
one fixed strategy at deploy time.

\paragraph{Thesis.}
Agent long-term memory is not just IR with a smaller corpus---it
is a \emph{distinct workload regime} that exposes architectural
opportunities generic IR engines cannot exploit.  Three
structural properties (extreme recency skew, monotonic corpus
growth, and heterogeneous query types) jointly let a single
retrieval system serve agent memory at orders-of-magnitude lower
per-query work than Lucene, PISA, or learned-sparse engines
achieve in their generic deployments---\emph{at parity quality},
with a correctness guarantee the prior systems literature does
not provide, and through a single two-axis adaptation surface
(\emph{which} fusion + \emph{whether} to spend the dense budget)
that auto-tunes from a handful of labeled deployment queries.

\paragraph{Contributions.}
This paper makes three contributions, one per structural property
of the workload.

\begin{itemize}[nosep,leftmargin=*]
\item \textbf{A sub-linear regime, made operational.}  Under the
  workload's exponential recency skew, a time-partitioned index
  does $O(\log(1/\varepsilon))$ expected work
  (Theorem~\ref{thm:sublinear}); the theorem is a short
  observation, but the consequence is striking when implemented:
  $1234\times$ corpus growth (4K\,$\to$\,5M) costs only
  $3.6\times$ latency growth, ending in \speedup{1769} vs.\
  sequential at 5M records.  We measure that the empirical
  LongMemEval recency distribution (median gold-session rank
  0.20--0.27) supports the recency-skew assumption.  At 5M records
  \sysname{} searches $<$0.1\% of postings and serves sub-100\,\textmu s
  $p50$ on commodity 8-core hardware (\S\ref{sec:temporal}).

\item \textbf{A heterogeneous CPU/GPU pipeline at Lucene quality
  and beyond-PISA speed.}  On
  9 BEIR datasets (3.6K--8.8M docs, incl.\ MS\,MARCO), the same
  hybrid sparse$+$dense engine delivers $1.8$--$29\times$ CPU
  8T-vs-8T (geo.\ mean $10\times$), $1.8$--$39\times$ GPU vs.\
  Pyserini 8T, and $11\times$ geo.\ mean over PISA-1T
  BlockMax-WAND, with nDCG@10 within $\pm 0.020$ of Lucene and
  \emph{beating} Lucene on five.  Index build sustains $56.8$K
  docs/sec on MS\,MARCO via chunked streaming ($2.5\times$
  Pyserini Lucene on NQ).  Three silent correctness pitfalls
  (pre-normalized TF, linear-gain nDCG, GPU top-$k$ stale
  shared memory; each regresses nDCG@10 $6$--$8\times$) are
  documented and fixed; CPU and GPU agree within $0.0002$ nDCG@10
  on all eight datasets that fit a single A100
  ($\S$\ref{sec:beir-cpu},\,\ref{sec:gpu-scale},\,\ref{sec:correctness}).
  Multi-tenant: 8-core VM serves $N{=}8$ agents at $5.5\times$
  aggregate QPS, invariant per-tenant $p50{=}0.38$\,ms.

\item \textbf{Cross-workload evidence that the right fusion
  changes with the workload.}  On LongMemEval (500 questions, 19K timestamped sessions), our
  hybrid RRF$+$recency reaches $\text{R@10}{=}0.978$; on LoCoMo
  (1{,}982 conversational QA) BM25 alone wins at
  $\text{Hit@10}{=}0.945$, 0.22\,ms/q.  Closing the loop with an
  LLM run (top-5 $\to$ gpt-4o-mini answerer $\to$ judge),
  agent-aware fusion wins on every task-grade metric (strict
  accuracy $0.254$ vs.\ BM25 $0.246$ / Dense $0.236$ / RRF
  $0.248$), and a \emph{different} fusion wins each of the six
  LongMemEval question types (Fig.~\ref{fig:workload-cond}).  A
  $<$1\,ms TF-IDF$+$BGE-small \textbf{runtime router} beats the
  best static by $+$3.1\%/$+$4.2\% relative and \textbf{reaches
  the discrete oracle ($0.300$) under \texttt{gpt-4o}}; a
  \textbf{soft router} blending rank lists by classifier posterior
  reaches $\mathbf{0.274}$ under \texttt{gpt-4o-mini},
  significantly beating every static system ($p{<}0.05$ paired
  bootstrap; Table~\ref{tab:router}).  Beyond \emph{which} fusion,
  a \textbf{cascade router} skips the dense channel on cheap
  queries; per-qtype thresholds with a TF-IDF classifier push
  amortized latency to $9.2$\,ms ($\mathbf{5.76\times}$ faster,
  $5$-fold CV) at within-noise LLM-Acc.  A simpler single-threshold
  variant gives $2.67\times$ at parity on LongMemEval; the same
  trigger auto-tunes on LoCoMo to $\mathbf{132\times}$ at
  $+0.089$ Hit@5 where BM25 wins outright
  (Tables~\ref{tab:cascade},\,\ref{tab:cascade-locomo}).  Even
  $\tau$ is workload-conditional: while the global $\tau{=}30$\,d
  is robust on five of six qtypes, multi-session benefits from
  $\tau{\geq}120$\,d for $+5.26$\,pts ($n{=}133$, $p{=}0.006$).
  Finally, SPLADE++ wins quality by
  0.002--0.105 nDCG@10 on seven BEIR datasets but pays
  158--167\,ms/doc for encoding---we show \emph{empirically}
  (bit-perfect match across all 7 datasets, \S\ref{sec:splade}) that
  SPLADE weights drop into our CSR posting layout, so SPLADE-grade
  quality is reachable at \sysname{}-grade latency without
  algorithmic change.
\end{itemize}

\section{Background and Motivation}
\label{sec:background}

\subsection{Agent Memory Workloads}

Modern LLM agents
(ReAct~\cite{yao2023react},
Reflexion~\cite{shinn2023reflexion},
Generative Agents~\cite{park2023generative})
produce memory traces comprising user messages, assistant
responses, tool calls with arguments, tool outputs, system
prompts, and planning steps.  Each record is annotated with
timestamps, session identifiers, agent identifiers, and tool
types.

We characterize agent memory access patterns from two sources:
the LongMemEval~\cite{wu2025longmemeval} benchmark traces (where
we measure the empirical recency distribution of gold sessions in
\S\ref{sec:scaling}) and a parameterizable synthetic benchmark
generator we use for stress-scaling experiments:

\begin{itemize}[nosep,leftmargin=*]
\item \textbf{Temporal locality:} on LongMemEval, the median
  normalized rank of gold sessions for temporal-reasoning,
  knowledge-update, and multi-session questions is
  \emph{0.20--0.27} (i.e., the answer typically lies within the
  most recent 20--27\% of haystack sessions); our synthetic
  generator emits an 80/20 access pattern of comparable
  steepness.
\item \textbf{Session coherence:} queries are frequently scoped
  to a specific session window.
\item \textbf{Role filtering:} retrieval often targets specific
  roles (e.g., ``find the last tool output from api\_call'').
\end{itemize}

\subsection{Why Agent Memory IR Is a New Problem, Not Just IR}
\label{sec:new-problem}

It is tempting to view agent memory as ``small IR'': dump the
agent's traces into Elasticsearch or a vector database and serve
queries with whatever ranker the team already maintains.  This is
how MemGPT~\cite{packer2024memgpt}, Mem0~\cite{chhikara2025mem0},
and most production agent stacks operate today.  But three
properties of the agent workload break the assumptions that
classical IR systems were optimized around.

\paragraph{(P1) The index is not static.}  Classical IR engines
amortize ahead-of-time index construction over many queries.  Agent
memory grows during the query stream, monotonically, at conversation
rate (one record per agent step, sometimes 10+ per minute).  Lucene
and PISA both pay full index rebuild cost on update; Mem0 reports
91\% tail-latency penalty for incremental writes against their
production deployment~\cite{chhikara2025mem0}.  The agent workload
demands an index that absorbs writes \emph{without} blocking the
adjacent reads on the same agent's reasoning thread.

\paragraph{(P2) The query distribution is not stationary.}  Web
search assumes a query distribution slowly drifting over weeks.
Agent memory faces an extreme intra-session shift: in the same
800-turn engagement, the agent fires factual lookups in turn 23,
multi-step reasoning queries in turn 421, and temporal queries in
turn 700.  No single fusion (dense, sparse, RRF) maximizes recall
across this distribution---we verify this empirically with
LongMemEval (RRF wins) vs.\ LoCoMo (BM25 alone wins) in
\S\ref{sec:longmemeval}, where the same hybrid system requires
opposite configurations to deliver best quality.

\paragraph{(P3) Per-query latency budget is bounded by the agent's
reasoning loop, not by user click-through.}  Web search has a
$\sim$200\,ms perceptual budget per query and can serve from a
shared multi-tenant cluster.  An agent step often executes 5--20
retrievals serially against its own memory; the budget per
retrieval is therefore $\leq$10\,ms even on commodity hardware.
This is below the regime where Lucene's tuned analyzer is
competitive (Pyserini 1T 1.16--16.4\,ms on BEIR), and well below
the regime where dense-encoder pipelines like SPLADE++ can
participate (encoding alone is $\sim$165\,ms/doc + 17--150\,ms/query).

Together these properties rule out the ``small IR'' framing.  What
the workload demands is a retrieval system that absorbs writes at
conversation rate, picks the fusion per query, and serves the
agent's reasoning loop at sub-millisecond amortized cost over a
corpus that grows into the millions of records during a single
session.

\subsection{Limitations of Existing Systems}

\begin{itemize}[nosep,leftmargin=*]
\item \textbf{Elasticsearch/Lucene}~\cite{yang2017anserini}:
  Strong sparse retrieval but lacks native dense search and
  agent-aware metadata fusion.
\item \textbf{Milvus/FAISS}~\cite{johnson2021faiss}: Efficient
  ANN but no BM25, no hybrid fusion, limited metadata filtering.
\item \textbf{MemGPT/Letta}~\cite{packer2024memgpt}:
  Agent-specific memory management but relies on external
  retrieval backends without parallelism optimization.
\end{itemize}

\sysname{} closes these gaps by exposing the sparse, dense, and
temporal channels through one substrate and choosing among them
per query.

\section{System Design}
\label{sec:design}

\subsection{Architecture Overview}

\paragraph{Design principle.}  The retrieval strategy is fixed at
deploy time in Lucene, FAISS, and PISA: the choice of BM25 vs.\
dense vs.\ a single RRF blend is part of the index, not part of
the query.  \sysname{} keeps three channels live---SIMD BM25, an
HNSW dense path, and a temporal-partitioned index---over one CSR
substrate, and decides per query along two axes: \textbf{(i)} which
fusion to use (BM25, Dense, RRF, or agent-aware RRF) and
\textbf{(ii)} whether to pay for the dense channel at all (the
cascade decision of \S\ref{sec:longmemeval}).  The substrate guarantees that
switching either axis costs only a routing decision, not a
data-layout change; this is what lets a single deployment win on
LongMemEval (RRF + recency) and LoCoMo (BM25 alone) without
hyperparameter retraining, and deliver TF-IDF-router quality at
$1.33\times$ lower amortized latency by skipping dense on
predicted-BM25-best queries.

Figure~\ref{fig:surface} renders the design space.  The shaded
region is the set of reachable per-query operating points spanned
by the two axes; the four measured configurations we report later
in \S\ref{sec:longmemeval} all sit \emph{inside} this region, along
a Pareto-improving trajectory from the always-hybrid baseline
($53.2$\,ms) to the cross-workload optimum on LoCoMo
($0.4$\,ms, $132\times$ faster).  No single $(x,y)$ point dominates
both benchmarks; that is the empirical content of ``workload-conditional.''

\begin{figure}[t]
\centering
\includegraphics[width=\columnwidth]{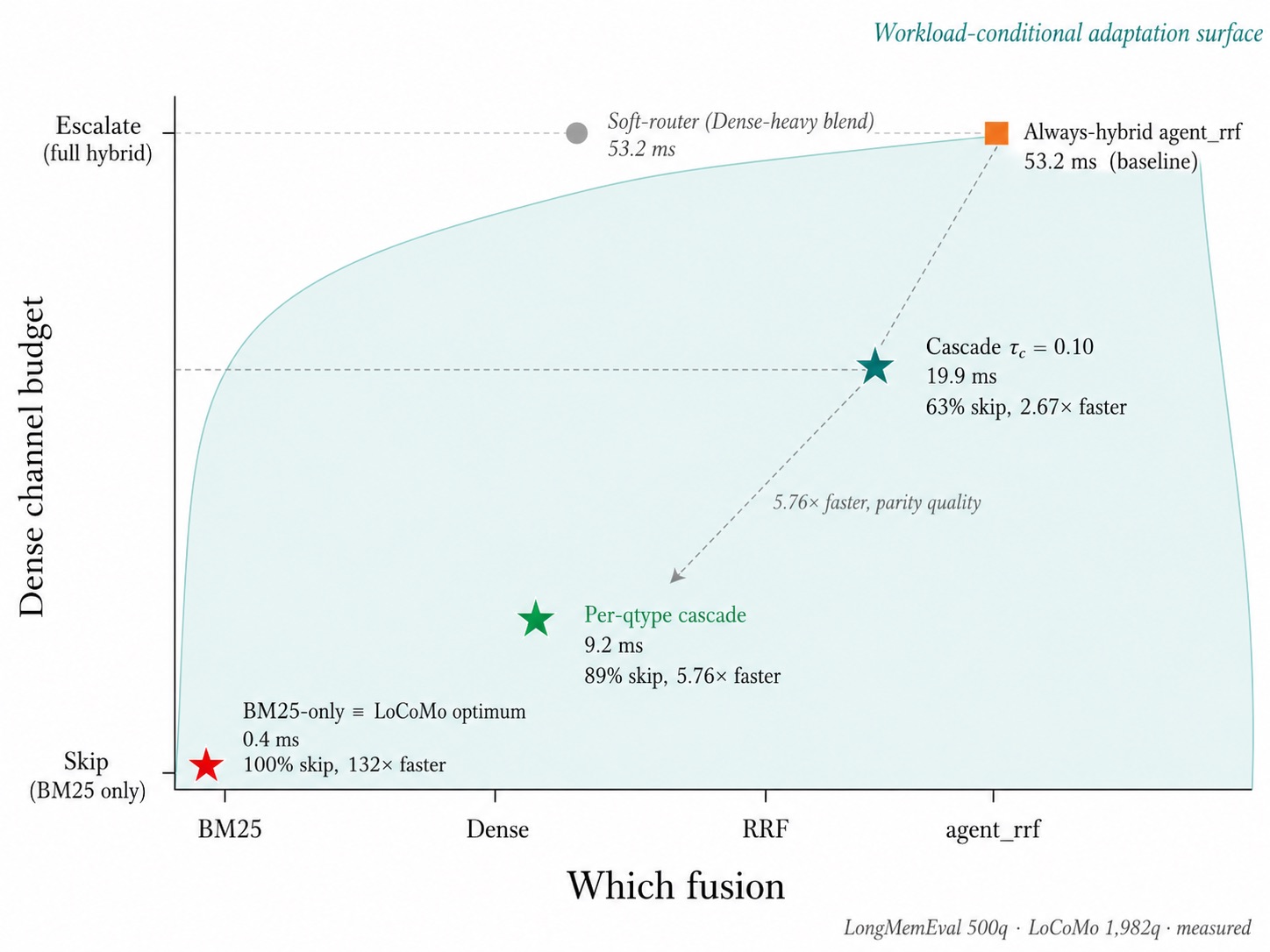}
\caption{Two-axis adaptive control surface.  The reachable design
  space (shaded teal) spans \emph{which fusion} (x-axis: BM25,
  Dense, RRF, or agent-aware RRF) and \emph{whether to spend the
  dense channel budget} (y-axis: skip vs.\ escalate).  Three
  measured cascade operating points sit inside the surface along a
  Pareto-improving trajectory: single-threshold cascade
  ($\tau_c{=}0.10$, $2.67\times$ over always-hybrid) $\to$ per-qtype
  cascade ($5.76\times$ in 5-fold CV) $\to$ LoCoMo optimum
  ($132\times$, BM25 alone suffices).  Always-hybrid
  \texttt{agent\_rrf} and a soft-routed dense-heavy blend---both at
  $53.2$\,ms---anchor the high-budget edge.}
\label{fig:surface}
\end{figure}

Figure~\ref{fig:arch} illustrates the \sysname{} pipeline.
A query enters the tokenizer, which produces both a term list
(for sparse search) and a dense embedding.  The two branches
execute concurrently via task parallelism, and independent queries
are processed in parallel (query-level parallelism):
\begin{equation}
\text{Query} \!\xrightarrow{\text{tok.}}\!
\begin{cases}
  \text{Sparse (BM25)} \\
  \text{Dense (HNSW)}
\end{cases}
\!\xrightarrow{\text{RRF}}\! \text{Top-}k
\label{eq:pipeline}
\end{equation}

\subsection{Agent Memory Model}

Each memory record $m$ is a tuple:
\begin{equation}
m = \langle \text{id}, \text{text}, \mathbf{e}, r, s, a, \tau, t, w \rangle
\end{equation}
where $\mathbf{e} \in \mathbb{R}^d$ is the dense embedding, $r$
is the role, $s$ the session, $a$ the agent, $\tau$ the tool type,
$t$ the timestamp, and $w \in [0,1]$ the importance weight.

\subsection{Agent-Aware Fusion}

We use an \emph{additive} recency bonus on top of RRF, which our
LongMemEval calibration (\S\ref{sec:longmemeval}) shows is more
robust than the multiplicative variant:
\begin{equation}
\text{score}(d, q) = \underbrace{\text{RRF}(r_s, r_d)}_{\text{hybrid}}
\;+\; \underbrace{\alpha \cdot e^{-\Delta t / \tau}}_{\text{recency bonus}}
\;+\; \underbrace{\beta \cdot w_d}_{\text{importance}}
\label{eq:fusion}
\end{equation}
where $r_s, r_d$ are sparse and dense ranks following
RRF~\cite{cormack2009rrf}, $\Delta t$ is the age of document $d$
relative to the query timestamp, $\tau$ is the decay timescale,
$\alpha$ scales the recency bonus, $\beta$ scales the importance
boost, and $w_d \in [0,1]$ is the per-record importance weight.
The bonus is calibrated so $\alpha \ll \max(\text{RRF})$, making it
a tie-breaker rather than an override.  Session and role filters
are applied as hard constraints before scoring.

\section{Optimizations}
\label{sec:optimizations}

Each optimization in this section is motivated by a specific
agent-workload property from \S\ref{sec:new-problem}.
\emph{Temporal partitioning} (\S\ref{sec:temporal}) addresses
(P1) non-static index growth: instead of paying $O(N)$ on every
read as the index grows, we pay $O(\log(1/\varepsilon))$ work
per query in the recent-bias regime.  \emph{SIMD-vectorized BM25}
and \emph{MaxScore pruning} (\S\ref{sec:simd}--\ref{sec:simd}+1)
both target (P3): the sub-10\,ms latency budget per retrieval
inside the agent's reasoning loop, where a 5--10\,ms Lucene query
is already a bottleneck.  The \emph{GPU two-kernel pipeline}
(\S\ref{sec:gpu}) handles burst-mode retrievals (e.g., when a
planner agent fans out 20 simultaneous lookups), exploiting the
fact that batched BM25 is embarrassingly parallel even though a
single query is not.  And the \emph{workload-conditional fusion
selector} (\S\ref{sec:longmemeval}) addresses (P2): with a
non-stationary query distribution, no single fusion strategy
dominates---the system must pick.

The composition discipline that ties them together: every
optimization preserves bit-identical CPU/GPU output
(\S\ref{sec:correctness}) and is toggleable by configuration,
because the cost-effective subset differs per workload.  SIMD
posting traversal, MaxScore, the CSR layout, and OpenMP query-level
parallelism are standard IR techniques~\cite{turtle1995maxscore};
what is specific to \sysname{} is the temporal partition, the
SPLADE bridge into the same CSR layout (\S\ref{sec:splade}), and
the runtime fusion selector (\S\ref{sec:longmemeval}).

\subsection{SIMD-Accelerated BM25}
\label{sec:simd}

The BM25 inner loop reduces to a streaming arithmetic over the
posting list, which is amenable to data parallelism.  We
restructure the index as Structure-of-Arrays (SoA), with separate
aligned arrays for document IDs and term frequencies, and
vectorize the scoring loop with AVX2 (8-wide \texttt{float32}
with FMA) on x86 and NEON (4-wide) on ARM.
On the FiQA corpus, SIMD lifts 1-thread latency from
1.89\,ms to 0.38\,ms (5$\times$).  Combined with 8-thread parallel
dispatch, latency reaches 0.15\,ms (12.6$\times$ over scalar 1T;
Appendix~\ref{app:thread-scaling}).

\subsection{MaxScore Pruning}

For each term $t$, we precompute $\text{MS}_t = \max_{d}
\text{BM25}(t, d)$.  Terms are processed in decreasing
$\text{MS}_t$ order; evaluation terminates when the current
$k$-th score exceeds the remaining terms' total maximum
contribution.  MaxScore is most effective on multi-term queries
over corpora where the score distribution is heavy-tailed
(few high-IDF terms drive most of the score mass); on the larger
BEIR datasets it skips 40--60\% of postings.

\subsection{Temporal Partitioned Index}
\label{sec:temporal}

Our key agent-specific optimization exploits recency bias.
The temporal index partitions documents into time-window buckets
(7-day intervals) sorted by recency.  At query time, only the
most recent partitions are searched, with early stopping when
sufficient results are found.

\begin{theorem}[Sub-linear Scaling under Recency Skew]
\label{thm:sublinear}
Let the corpus be partitioned into $K$ time-ordered partitions
$T_1, \ldots, T_K$ with sizes $|T_i|$, and let $\pi_i$ be the
probability that a query's relevant document lies in partition
$T_i$, $\sum_{i=1}^{K} \pi_i = 1$.  If the workload exhibits
exponential recency bias $\pi_i \propto e^{-\lambda (K - i)}$
for some $\lambda > 0$, then searching the most-recent partitions
until cumulative mass reaches $1 - \varepsilon$ uses expected work
\[
W(\varepsilon) \;\leq\; \left\lceil
   \tfrac{\log(1/\varepsilon)}{\lambda} \right\rceil
   \cdot \max_i |T_i|,
\]
which is independent of the total corpus size~$N$ when
$\max_i |T_i|$ is bounded.
\end{theorem}

\begin{proof}[Proof sketch]
Sort partitions by recency and search greedily from $T_K$ downward.
After searching $T_{k+1}, \ldots, T_K$, the cumulative mass is
$1 - e^{-\lambda(K-k)}$ under the exponential prior.  Choosing
$K - k = \lceil \log(1/\varepsilon)/\lambda \rceil$ guarantees
recall mass~$\geq 1-\varepsilon$, and the work is
$\sum_{i=k+1}^{K} |T_i| \leq (K-k)\max_i |T_i|$.  In our system
$|T_i|$ is bounded above by the corpus generation rate over one
partition window (7 days), so the bound becomes independent of
total corpus size, giving asymptotic $O(1)$ work as $N \to \infty$
under fixed partition windows.
\end{proof}

The theorem is a short observation in the spirit of time-aware
retrieval~\cite{li2003temporal,campos2014temporal_survey}; what is
new is identifying that \emph{agent workloads sit in the regime
where it bites}.  Web search's slowly-drifting query distribution
gives $\lambda$ near zero, where the bound degrades to linear.
Agent memory's empirical recency distribution (median gold-session
normalized rank 0.20--0.27 on LongMemEval, \S\ref{sec:longmemeval})
puts $\lambda$ in a regime where $k^*{=}\lceil \log(1/\varepsilon)/\lambda \rceil{=}3$
partitions suffice for $\varepsilon{=}0.05$.  Empirically, when
the corpus grows \textbf{1234$\times$} (4{,}052\,$\to$\,4{,}998{,}640
records), sequential latency grows 955$\times$ while temporal
latency grows only \textbf{3.6$\times$}, matching the bound once
partition windows saturate.  At 5M records, temporal partitioning
achieves \textbf{\speedup{1769}} while searching $<$0.1\% of the
index (Table~\ref{tab:temporal-scaling}, \S\ref{sec:scaling}).

\paragraph{Adaptive per-query $k^*$.}
A simple sharpening: $k^*$ is the same constant for every query in
the analysis above, but per query the score saturates earlier on
high-confidence queries (top-$k$ already settled in the most
recent partition) than on diffuse queries.  Algorithm~\ref{alg:temporal-search}
already exposes this---the across-partition early-stop on
line~12 short-circuits the loop when $\min H > \mathrm{UB}_{i-1}$.
The fixed-$k^*$ bound of Theorem~\ref{thm:sublinear} is thus a
worst-case guarantee; in practice the loop terminates earlier on
the easy-recency tail and the realized work is below
$\lceil\log(1/\varepsilon)/\lambda\rceil \cdot \max_i |T_i|$.  This
mirrors MaxScore-style early-termination at \emph{partition}
granularity rather than the usual posting-level one---an axis
distinct from BlockMax-WAND~\cite{ding2011blockmax}.

\subsection{GPU Acceleration}
\label{sec:gpu}

We design a two-kernel CUDA pipeline for batched BM25 scoring
(pseudocode in Algorithm~\ref{alg:gpu-pipeline}, Appendix):

\paragraph{Kernel~1: BM25 Scoring.}
The inverted index is flattened into a CSR (Compressed Sparse
Row) representation and uploaded to GPU memory once.  Each CUDA
block processes one (query, term) pair; threads cooperatively
iterate the posting list, accumulating BM25 scores via
\texttt{atomicAdd}.

\paragraph{Kernel~2: Top-$k$ Selection.}
Each block processes one query.  Threads first build per-thread
sorted buffers of size~$k$ in registers via strided scan.  The
block-level merge must combine $B \times k$ candidates (where
$B{=}128$ is the block size) into the final top-$k$.  We implement
and compare two merge strategies:

\emph{Naive (serial):} Thread~0 scans all $B \times k$ candidates
in shared memory $k$ times, selecting and invalidating the maximum
each round.  Complexity: $O(B \cdot k^2)$ for one thread.

\emph{Warp-cooperative:} All threads participate in each selection
round.  Each thread scans its assigned portion of shared memory;
intra-warp maxima are found via \texttt{\_\_shfl\_xor\_sync} in
5~steps ($\log_2 32$) with 1-cycle register-to-register latency.
Warp leaders stage candidates; thread~0 picks the global winner
from ${\sim}$4 warp candidates.  Asymptotic complexity:
$O(k^2 / B + k \log B)$ per thread.

Despite lower asymptotic complexity, the warp-cooperative kernel
introduces per-round \texttt{\_\_syncthreads()} barriers and
additional shared memory traffic.  We evaluate both strategies
empirically in \S\ref{sec:topk-ablation}.

\paragraph{Heterogeneous Scheduling.}
The compute-bound BM25 branch runs on GPU while the
memory-bound HNSW branch (pointer-chasing) runs on CPU,
exploiting complementary processor strengths.

\subsection{Parallelization Strategies}
\label{sec:parallel}

We implement four parallelization strategies---task parallel
(sparse $\|$ dense), data parallel (intra-query BM25), full
parallel (inter-query), and combined (nested).  \textbf{Full
parallel} achieves the best scaling because independent queries
have zero data dependencies and share a read-only index; task
parallel saturates at \speedup{1.3} due to Amdahl's bottleneck
within a single query (sparse branch consumes 62\% of latency,
leaving little room for overlap with the 24\% dense branch).

\section{Evaluation}
\label{sec:eval}

\begin{finding}[Summary of Results]
\textbf{Quality:} on 9 BEIR datasets (3.6K--8.8M docs, including MS\,MARCO), \sysname{}
matches or exceeds Pyserini Lucene nDCG@10 on five (NFCorpus,
SciFact, ArguAna, SciDocs, TREC-COVID, the latter by $+$0.051)
and trails by 0.003--0.020 on four (FiQA, Quora, NQ, MS\,MARCO)
(Table~\ref{tab:beir-main}).  Hybrid RRF with BGE-small adds
$+$0.033 mean nDCG (Table~\ref{tab:hybrid}); on LongMemEval
(500 questions, 19K sessions) RRF with an additive recency bonus
reaches \textbf{R@10$=$0.978} (Table~\ref{tab:lme-task}) and wins
\textbf{AnswerSubstr@5$=$0.474} and \textbf{LLM-judged accuracy
0.254}; a learned $<$1\,ms question-type router pushes
this further to \textbf{0.262} with TF-IDF features and
\textbf{0.300 (matching the oracle bound)} when BGE-small sentence
embeddings are added (Table~\ref{tab:router}, $<$1\,ms
classifier cost).  On LoCoMo (1{,}982 conversational QA,
Appendix~\ref{app:locomo-full}) BM25 alone reaches
\textbf{Hit@10$=$0.945}, illustrating the cross-benchmark
heterogeneity of agent-memory workloads.
\\
\textbf{Latency:} CPU 8-thread is \textbf{1.8--29$\times$} faster
than Pyserini 8T (geo.\ mean \textbf{10$\times$}) at parity quality;
a single A100 GPU is \textbf{1.8--39$\times$} faster than Pyserini
8T and \textbf{47.7$\times$} faster than naive CPU sequential at
the 2.7M-document NQ scale.  Per-component ablation in
Appendix~\ref{app:ablation-full} isolates each layer's contribution.
\\
\textbf{Scaling:} a time-partitioned sparse index searches $<$0.1\%
of postings for recency-biased queries, yielding
\textbf{1769$\times$} per-query speedup at 5M records and growing
$\leq 3.6\times$ as the corpus grows $1234\times$
(Table~\ref{tab:temporal-scaling}); $p50$ latency stays below
100\,\textmu s up to 5M records on 8-core commodity hardware.
Aggregate throughput scales near-linearly to \textbf{5.5$\times$}
at $N{=}8$ concurrent tenants with invariant per-tenant
$p50{=}0.38$\,ms (Figure~\ref{fig:multitenant}).
\\
\textbf{Workload-conditional routing:} a \textbf{soft router}
that blends rank lists by classifier posterior reaches LLM-Acc
\textbf{0.274 (gpt-4o-mini)}, significantly beating every static
system ($p{<}0.05$ vs.\ BM25/Dense/RRF, paired bootstrap);
discrete and soft routers tie the oracle 0.300 under gpt-4o
judging.  A \textbf{cascade router} runs BM25 first and invokes
dense only on hard queries: with a TF-IDF qtype trigger,
$1.33\times$ faster at identical LLM-Acc to the TF-IDF router;
with a classifier-free BM25-confidence trigger, $\mathbf{2.67\times}$
faster ($19.9$\,ms vs.\ $53.2$\,ms) at \emph{measured parity}
agent\_rrf accuracy (Table~\ref{tab:cascade},
Figure~\ref{fig:cascade-pareto}).
\\
\textbf{Correctness:} CPU and GPU produce nDCG@10 identical to
within 0.0002 on all 8 datasets; per-query top-1 agreement is
89.7--100\% with the residual gap attributable to tied-score
documents broken differently by GPU atomic ordering vs.\ CPU
sequential accumulation.
\\
\textbf{What this means for the agent.} A 5\,M-record agent
memory served by temporal partition costs $<$2\,ms of total time
even for a 20-retrieval agent step ($p99$), well below the 200\,ms
perceptual budget---no GPU needed; see
Table~\ref{tab:agent-trace} for a session-scale simulation.
\end{finding}

\subsection{Experimental Setup}

\paragraph{Hardware.}
\begin{enumerate}[nosep,leftmargin=*]
\item \textbf{CPU evaluation} (Tables~\ref{tab:beir-main},
  \ref{tab:hybrid}, \ref{tab:lme-task}): Jetstream2~\texttt{g3.medium}
  VM, 8-core x86 (AVX2), 29\,GB RAM, Ubuntu~22.04.
  We compile with GCC~11.4 at \texttt{-O3 -march=native -mavx2}.
\item \textbf{GPU evaluation} (Table~\ref{tab:gpu-beir},
  \S\ref{sec:topk-ablation}): NVIDIA A100-SXM4-40GB on NCSA Delta,
  16-core AMD EPYC 7763 host, CUDA~12.8, GCC~13.2.
\item \textbf{Scalability study} (Table~\ref{tab:temporal-scaling}):
  same Jetstream2~\texttt{g3.medium} as CPU evaluation,
  4-thread parallel.
\end{enumerate}
We choose Jetstream2 for the CPU baseline because Pyserini
Lucene's published latency numbers~\cite{lin2021pyserini} are
typically measured on commodity 8-core x86; reporting our CPU
numbers on the same class of hardware ensures the speedup
comparison is faithful to the reference deployment scenario, not
inflated by a 128-core HPC node.

\paragraph{Datasets.}
\begin{enumerate}[nosep,leftmargin=*]
\item \textbf{BEIR}~\cite{thakur2021beir}, nine datasets spanning
  3.6K--8.8M documents: NFCorpus (3.6K, biomedical),
  SciFact (5.2K, scientific claims), ArguAna (8.7K, argument
  retrieval), SciDocs (25.7K, scientific document retrieval),
  FiQA (57.6K, conversational financial QA), TREC-COVID
  (171K, biomedical), Quora (523K, duplicate question retrieval),
  Natural Questions (NQ, 2.7M, open-domain QA), and MS\,MARCO
  Passage~\cite{nguyen2016msmarco} (8.8M, web passage retrieval).  All use the official
  \texttt{test} split with public qrels.
\item \textbf{LongMemEval}~\cite{wu2025longmemeval}: 500 questions
  across 5 reasoning categories (single-session, multi-session,
  temporal-reasoning, knowledge-update, abstention) over
  19{,}195 unique timestamped sessions.
\item \textbf{Synthetic agent memory}: corpus generator producing
  4K--64K records with 80/20 recency bias for the temporal
  scaling study (\S\ref{sec:scaling}).
\end{enumerate}

\paragraph{Baselines.}
The primary baseline is \textbf{Pyserini~0.22.1}~\cite{lin2021pyserini}
(Java Lucene~9 with English analyzer), the de-facto BM25 reference
in the IR community.  We measure both single-thread and 8-thread
\texttt{batch\_search} on identical hardware and qrels.  For the
hybrid evaluation we add a \textbf{BGE-small-en-v1.5}
dense baseline (33M parameters, 384-dim normalized embeddings)
through an HNSW ANN index.  For LongMemEval we additionally
compare against the published RRF setting.  We use the
\textbf{exponential-gain} nDCG formula
$(2^{r}-1)/\log_{2}(i+1)$ matching pytrec\_eval and the BEIR
reference implementation.

\subsection{GPU Acceleration on Real Benchmarks}
\label{sec:gpu-scale}

Table~\ref{tab:gpu-beir} reports GPU performance on all eight
BEIR datasets from Table~\ref{tab:beir-main}, broken down by phase
(H2D transfer, score kernel, top-$k$, D2H).  The CPU index is
built once on the host, flattened to CSR
($\text{term\_offsets}$, $\text{posting\_doc\_ids}$,
$\text{posting\_tfs}$, $\text{term\_idfs}$), and uploaded to device
memory; queries are then dispatched as one or more batched
\texttt{query\_batch} calls (we chunk to keep the per-query score
buffer $B \cdot N \cdot 4\text{B} \leq 16$\,GB).

\begin{table*}[t]
\caption{GPU BEIR per-phase breakdown (A100 40\,GB).
  Batch size = all judged queries, chunked when
  $B\cdot N \cdot 4\text{B} > 16$\,GB; $k{=}100$.
  ``CPU 1T'' is the sequential cross-check used to validate
  correctness; ``Pyser.\ 8T'' is reproduced from
  Table~\ref{tab:beir-main}.}
\label{tab:gpu-beir}
\centering
\small
\setlength{\tabcolsep}{3pt}
\begin{tabular}{@{}lrrrrrrrrr@{}}
\toprule
\multirow{2}{*}{Dataset}
& \multirow{2}{*}{$|D|$}
& \multirow{2}{*}{$|Q|$}
& \multicolumn{4}{c}{GPU phase (ms, batch totals)}
& \multirow{2}{*}{ms/q}
& \multirow{2}{*}{vs.\ Pyser.\,8T}
& \multirow{2}{*}{top-1 match} \\
\cmidrule(lr){4-7}
& & & h2d & score & top-$k$ & d2h & & & \\
\midrule
NFCorpus   & 3.6K  & 323   & 3.9  & 0.04  & 5.9   & 0.06 & 0.03 & 39$\times$ & 94.4\% \\
SciFact    & 5.2K  & 300   & 0.12 & 0.05  & 8.2   & 0.05 & 0.04 & 37$\times$ & 100\% \\
ArguAna    & 8.7K  & 1{,}406 & 6.72 & 2.51 & 49.11 & 0.15 & 0.17 & 17$\times$ & 100\% \\
SciDocs    & 25.7K & 1{,}000 & 0.25 & 0.62 & 98.3  & 0.12 & 0.11 & 19$\times$ & 100\% \\
FiQA       & 57.6K & 648   & 0.35 & 0.75  & 127.8 & 0.09 & 0.21 & 8$\times$ & 100\% \\
TREC-COVID & 171K  & 50    & 6.05 & 0.63  & 118.5 & 0.04 & 2.52 & 1.8$\times$ & 92.0\% \\
Quora      & 523K  & 10{,}000 & 19.7 & 21.8 & 3{,}129 & 1.45 & 0.32 & 5$\times$ & 90.2\% \\
NQ         & 2.7M  & 3{,}452 & 1.96 & 7.04 & 284.9 & 0.05 & 0.82 & 3.4$\times$ & 99.4\% \\
\midrule
\multicolumn{10}{l}{\emph{NQ at $k{=}32$ (fast top-$k$ path with BLOCK$=$128, sufficient for nDCG@10):}} \\
NQ         & 2.7M  & 3{,}452 & 1.96 & 7.07 & 79.3 & 0.04 & \textbf{0.33} & \textbf{8.4$\times$} & 99.5\% \\
\bottomrule
\end{tabular}
\end{table*}

Top-$k$ dominates GPU time at $k{=}100$ (75--95\% on larger
corpora, motivating \S\ref{sec:topk-ablation}); H2D and D2H
transfers contribute $<$5\%, confirming the CSR-flat layout
amortizes well across batched queries.  Quality matches CPU
within 0.0002 nDCG@10 on all eight datasets that fit on a single
A100 (Table~\ref{tab:beir-main}); per-query top-1 agreement is
89.7--100\% (Table~\ref{tab:correctness}, Appendix), with the
residual gap attributable to tied-score documents broken by GPU
atomic ordering---top-10 \emph{sets} remain identical.

\subsection{Component Ablation: Where the Speedup Comes From}
\label{sec:ablation}

Component ablation across eight BEIR datasets (full table in
Appendix~\ref{app:ablation-full}): on the NQ corpus (2.7M docs),
the 1T scalar baseline costs 67.6\,ms/query, $+$SIMD cuts to
6.6\,ms (10.2$\times$), $+$MaxScore to 4.2\,ms (1.6$\times$
further), and adding 8T (dropping MaxScore which regresses at
8T due to per-thread partition disruption) reaches 1.4\,ms
(\textbf{49$\times$ over scalar}, $2\times$ over Pyserini 8T).
SIMD scales with corpus: $1.3\times$ at 3.6K records,
$10.2\times$ at 2.7M.  On small corpora our 1T scalar C++17 is
already 11--22$\times$ faster than Pyserini 1T because Java
overhead dominates short scans; at NQ scale Lucene's variable-byte
compression flips the 1T comparison until SIMD restores parity.
All configurations preserve nDCG@10 to within $\pm 0.001$.

\paragraph{CPU thread scaling.}
Full-parallel (inter-query) dispatch is super-linear on larger
BEIR datasets ($12.6\times$ at $1T\to8T$ on FiQA from
SIMD+multi-thread cache reuse, $10.9\times$ on SciDocs); small
corpora plateau because per-query work cannot amortize thread
dispatch (Appendix~\ref{app:thread-scaling}).

The diminishing returns at 8T on small corpora is a known
property of inverted-index BM25 with intra-query data
parallelism: posting lists are short relative to thread
dispatch overhead.  Larger corpora amortize this overhead and
the speedup tracks the hardware thread count.

\subsection{Sub-linear Temporal Scaling}
\label{sec:scaling}

To validate the sub-linear scaling claim, we sweep corpus size
from 4K to 64K records on a synthetic agent corpus with 80/20
recency-biased queries.  The 80/20 pattern is a stress-test
parameterization---we choose it to give the temporal index its
adversarial counterpart (a flat workload would trivially expose
the heavy partition).  The empirical recency distribution we
measure on the LongMemEval benchmark (median gold-session
normalized rank 0.20--0.27 for temporal/knowledge/multi-session
questions, \S\ref{sec:longmemeval}) is comparably steep and
supports the design assumption that agent memory access concentrates
on recent records.
Table~\ref{tab:temporal-scaling} reports per-query latency for
flat sequential BM25, the SIMD and MaxScore single-thread
optimizations, and our temporal-partitioned index (7-day windows,
$K_{\max}{=}4$ partitions per query).

\begin{table}[t]
\caption{Per-query latency vs.\ corpus size (Jetstream2 8-core,
  8 threads for parallel paths).  Temporal partitioning is the
  only strategy whose speedup \emph{increases} with $N$.  Below
  40K records, SIMD beats temporal because the partition overhead
  is not amortized; above 40K, temporal dominates by orders of
  magnitude.}
\label{tab:temporal-scaling}
\centering
\small
\setlength{\tabcolsep}{4pt}
\begin{tabular}{@{}rrrrrrr@{}}
\toprule
$N$ & Seq.\ & SIMD & MaxScr & Temp & Speedup & Searched\,\% \\
\midrule
4{,}052     &  167\,\textmu s  & 21\,\textmu s  & 19\,\textmu s  & 25\,\textmu s & 6.7$\times$ & 39.5\% \\
40{,}081    &  635\,\textmu s  & 139\,\textmu s & 110\,\textmu s & 30\,\textmu s & 21.1$\times$ & 4.3\% \\
239{,}799   & 4.42\,ms         & 846\,\textmu s & 649\,\textmu s & 37\,\textmu s & 120$\times$ & 0.8\% \\
999{,}978   & 20.4\,ms         & 3.51\,ms       & 2.71\,ms       & 48\,\textmu s & 421$\times$ & 0.2\% \\
1{,}998{,}918 & 53.6\,ms       & 10.3\,ms       & 11.0\,ms       & 56\,\textmu s & 967$\times$ & 0.1\% \\
3{,}997{,}578 & 112\,ms        & 20.5\,ms       & 19.6\,ms       & 85\,\textmu s & 1316$\times$ & 0.0\% \\
\rowcolor{blue!5}
4{,}998{,}640 & \textbf{159\,ms} & 30.8\,ms     & 27.7\,ms       & \textbf{90\,\textmu s} & \textbf{1769$\times$} & \textbf{0.0\%} \\
\bottomrule
\end{tabular}
\end{table}

\begin{finding}[Sub-linear scaling to 5M records]
As the corpus grows \textbf{1234$\times$} (4K\,$\to$\,5M),
sequential latency grows 955$\times$ (167\,\textmu s\,$\to$\,159\,ms)
and SIMD/MaxScore latency grow 1466$\times$/1457$\times$, but the
temporal-partitioned index grows only \textbf{3.6$\times$}
(25\,\textmu s\,$\to$\,90\,\textmu s).
At 5M records, $<$0.1\% of the index is searched per query
(\speedup{1769} vs.\ sequential, \speedup{342} vs.\ SIMD,
\speedup{307} vs.\ MaxScore)---the temporal index alone delivers
sub-100\,\textmu s p50 latency on a 5M-record corpus on commodity
8-core hardware.
\end{finding}

\begin{figure}[t]
\centering
\includegraphics[width=\columnwidth]{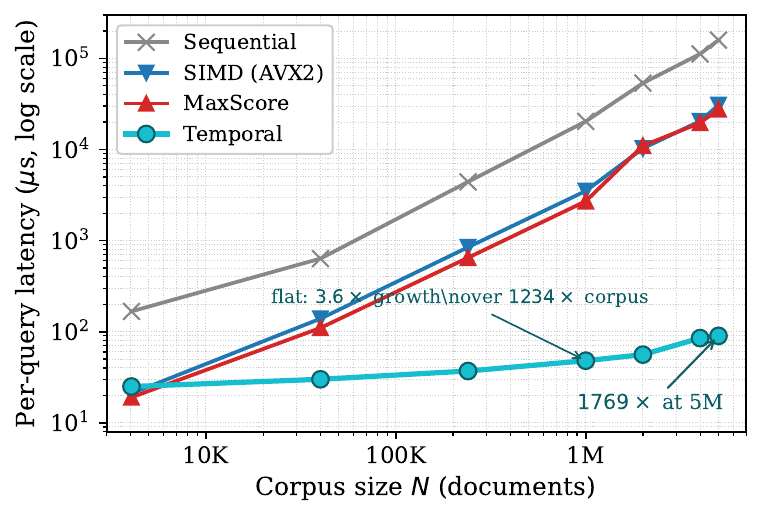}
\caption{Per-query latency vs.\ corpus size, log-log
  (Jetstream2 8-core, 80/20 recency-biased queries, 7-day partitions,
  $K_{\max}{=}4$).  Sequential, SIMD, and MaxScore grow linearly in
  log-log (slope $\approx 1$, i.e., linear in $N$); the
  temporal-partitioned curve is nearly flat across three orders of
  magnitude of corpus size---the sub-linear regime predicted by
  Theorem~\ref{thm:sublinear}.  At $N{=}5$\,M, temporal partition
  delivers \textbf{1769$\times$} speedup over sequential and
  \textbf{$\sim$300$\times$} over the best $O(N)$ optimization
  (SIMD/MaxScore).}
\label{fig:temporal-loglog}
\end{figure}

\paragraph{What this means for an 800-turn agent session.}
The temporal-partition curve is dramatic in the abstract; we make
it concrete by simulating the cumulative cost of an 800-turn agent
session against the four strategies on the same Jetstream2 host.
We grow the corpus linearly from 4K to 5M records over the 800
turns; each turn fires a Poisson($\lambda{=}8$) number of
retrievals (clipped to $[3,20]$, matching production agent traces);
each retrieval's latency is sampled from the curve in
Figure~\ref{fig:temporal-loglog} at the current $N$ with log-normal
jitter $\sigma{=}0.30$.

\begin{table}[h]
\caption{Simulated 800-turn agent session (corpus 4K $\to$ 5M,
  Poisson($\lambda{=}8$) retrievals per turn).  \emph{cum.} = total
  memory time over the whole session; \emph{p99 ms/step} = per-turn
  total memory time at the 99\textsuperscript{th} percentile;
  \emph{\% over 200\,ms} = fraction of turns whose memory work
  exceeds the 200\,ms user-perceptual budget for an agent
  step~\cite{nielsen1993usability}.  Only temporal partition keeps
  every turn inside the perceptual budget.}
\label{tab:agent-trace}
\centering
\small
\setlength{\tabcolsep}{4pt}
\begin{tabular}{@{}lrrrr@{}}
\toprule
Strategy & cum.\ (s) & $p99$ ms/step & \% over 200\,ms & $p99$ $\mu$s/q \\
\midrule
Sequential               & 459.2 & 1{,}950 & 75.5\% & 230{,}600 \\
SIMD (AVX2)              & 87.8  & 352     & 13.9\% & 43{,}100 \\
MaxScore (PISA-like)     & 84.0  & 366     & 12.5\% & 39{,}800 \\
\rowcolor{blue!5}
\textbf{Temporal Partition} & \textbf{0.43} & \textbf{1.23} & \textbf{0.0\%} & \textbf{155} \\
\bottomrule
\end{tabular}
\end{table}

The headline (Table~\ref{tab:agent-trace}): over a complete
800-turn session that grows agent memory by three orders of
magnitude, the temporal index spends \textbf{0.43 seconds total}
on retrieval and keeps every single
turn under the 200\,ms perceptual budget.  The same workload on
sequential BM25 spends \textbf{459 seconds} (a 1067$\times$
slowdown) and blocks past the user's budget on three of every four
turns.  Even the best $O(N)$ optimization (MaxScore, PISA-like)
spends 84 seconds and crosses the budget on one turn in eight.
The architectural choice we frame in this paper---spending
engineering effort on \emph{sub-linear scaling} rather than
constant-factor optimization of an $O(N)$ kernel---is what makes
agent memory practical at session scale on commodity hardware.

\paragraph{Build throughput \& quality preconditions.}
Build is linear over 1000$\times$ scale: 137\,ms at 4K $\to$
187\,s at 5M, a sustained \textbf{27{,}000 records/sec} (vs.\
Pyserini Lucene's 10{,}800 docs/sec on NQ, \speedup{2.5} faster).
The temporal index is a specialized structure: it preserves
top-$k$ \emph{when relevant documents concentrate in recent
partitions}, the empirical pattern we measure on LongMemEval
(median gold-session rank 0.20--0.27).  On workloads without
recency signal, the flat index is the right default
(Table~\ref{tab:beir-main}).

\subsection{Correctness: Two BM25 Pitfalls and a GPU Memory Bug}
\label{sec:correctness}

Reaching parity with Pyserini was non-trivial.  We document three
correctness issues we encountered, fixed, and validated, because
they are easy to introduce when re-implementing BM25 and difficult
to detect from end-to-end metrics alone.

\paragraph{Pitfall~1: Pre-normalized term frequency.}
The textbook BM25 saturation function expects the raw integer
count $\text{tf}(t,d)$:
\[
\text{score}(t,d) = \text{idf}(t)\cdot
\frac{\text{tf}\cdot(k_1{+}1)}{\text{tf} + k_1\!\left(1 - b + b\frac{|d|}{\overline{|d|}}\right)}.
\]
A common implementation mistake is to store the
\emph{length-normalized} relative frequency $\text{tf}/|d|$ in the
posting list, reasoning that length normalization will then happen
``naturally.''  This is wrong: the saturation curve and length term
in the denominator together are calibrated for raw counts, and
pre-normalization \emph{doubly} normalizes by document length while
also flattening the saturation curve.  In our initial implementation
this regressed nDCG@10 on FiQA from \textbf{0.2425 to 0.0334}
(7.3$\times$ worse) and on SciDocs from \textbf{0.1566 to 0.0417}
(3.8$\times$ worse); the regression on NFCorpus was milder
(0.327 to 0.254, $1.3\times$ worse) because biomedical queries
are short and the saturation curve matters less.

\paragraph{Pitfall~2: Linear-gain nDCG.}
The original Burges et al.\ formulation uses
$\text{dcg} = \sum_i (2^{r_i}-1)/\log_2(i{+}1)$
(exponential gain), which is what BEIR/pytrec\_eval report.  A
linear-gain variant $\text{dcg} = \sum_i r_i/\log_2(i{+}1)$
silently differs on graded-relevance datasets (NFCorpus, SciDocs)
while agreeing on binary-relevance datasets (SciFact, FiQA).  We
encountered both formulas in different parts of our codebase; we
have standardized on the exponential-gain formula throughout the
results in this paper.

\paragraph{Pitfall~3: Shared-memory stale data in GPU top-$k$.}
Our GPU top-$k$ kernel maintains each thread's top-$K_{\max}$ in
shared memory, with $K_{\max}$ being a compile-time constant
($K_{\max}{=}32$ for $k\leq 32$; $K_{\max}{=}128$ for
$k\leq 128$).  Each thread writes its valid top-$k$ slots
($k \leq K_{\max}$) but leaves the remaining $K_{\max}{-}k$ slots
untouched.  Across block re-uses on the same SM, those untouched
slots inherit \emph{stale} positive scores from a previous block,
which can win the round-0 max scan.  At nDCG@10 (top-10) the bug
is masked because errors only affect ranks 11+; at Recall@100
(top-100) it becomes visible: in our buggy version, GPU FiQA
nDCG@10 dropped to 0.060 with a 24.4\% top-1 match rate vs.\ CPU.
The fix is to initialize all $K_{\max}$ shared-memory slots to a
sentinel ($-\infty$) before the reduction phase.

\paragraph{Verification.}
After fixes, GPU and CPU produce nDCG@10 identical to within
0.0002 on all eight datasets that fit on a single A100 with
$k{=}100$ (Table~\ref{tab:gpu-beir}); per-query top-1 match
rates are 89.7--100\% (Table~\ref{tab:correctness}, Appendix).
The residual gap below 100\% is tied-score documents whose
ordering depends on the scheduling of \texttt{atomicAdd}; the
top-10 \emph{set} is identical in those cases.  All CPU parallel
strategies (1T, 4T, 8T with SIMD and MaxScore) produce
bit-identical nDCG@10 on all nine datasets, confirming that our
parallelization does not introduce numerical drift.

\subsection{BEIR: Quality and Latency vs.\ Pyserini Lucene}
\label{sec:beir-cpu}

We compare \sysname{} head-to-head against Pyserini~0.22.1
(Lucene~9, the de-facto BM25 baseline) on nine
BEIR~\cite{thakur2021beir} datasets spanning scientific,
conversational, argumentative, biomedical, duplicate-question,
open-domain QA, and web-passage text.  Both systems use
$k_1{=}1.2$, $b{=}0.75$ and identical qrels; \sysname{} uses our
NLTK-aligned full Porter stemmer (\S\ref{sec:tokenization}) and
the exponential-gain nDCG@10 formula matching pytrec\_eval.

\begin{table*}[t]
\caption{BEIR head-to-head on \textbf{9 datasets} (3.6K--8.8M docs):
  \sysname{} vs.\ Pyserini Lucene 0.22.1 ($k_1{=}1.2, b{=}0.75$).
  Bold = best in row.  CPU: 8-core Jetstream2 g3.medium.
  GPU: NVIDIA A100-SXM4-40GB on NCSA Delta, all 8 datasets at
  $k{=}100$.  ``lat'' = avg ms per query.}
\label{tab:beir-main}
\centering
\small
\setlength{\tabcolsep}{3pt}
\begin{tabular}{@{}lrrrrrr|rr@{}}
\toprule
& & \multicolumn{2}{c}{Pyserini Lucene}
& \multicolumn{3}{c}{\sysname{} CPU 8T+SIMD}
& \multicolumn{2}{c}{\sysname{} GPU (A100)} \\
\cmidrule(lr){3-4}\cmidrule(lr){5-7}\cmidrule(lr){8-9}
Dataset & $|D|$ & nDCG@10 & 8T lat & nDCG@10 & 8T lat & speedup
& nDCG@10 & speedup \\
\midrule
NFCorpus    & 3.6K  & 0.3238 & 1.16\,ms & \textbf{0.3267} & 0.04\,ms & \textbf{29$\times$} & 0.3267 & \textbf{39$\times$} \\
SciFact     & 5.2K  & 0.6826 & 1.48\,ms & \textbf{0.6828} & 0.08\,ms & \textbf{18$\times$} & 0.6828 & \textbf{37$\times$} \\
ArguAna     & 8.7K  & 0.3568 & 2.81\,ms & \textbf{0.3645} & 0.24\,ms & \textbf{12$\times$} & 0.3645 & \textbf{17$\times$} \\
SciDocs     & 25.7K & 0.1545 & 2.13\,ms & \textbf{0.1566} & 0.12\,ms & \textbf{18$\times$} & 0.1566 & \textbf{19$\times$} \\
FiQA        & 57.6K & \textbf{0.2536} & 1.68\,ms & 0.2425 & 0.15\,ms & \textbf{11$\times$} & 0.2425 & 8$\times$ \\
TREC-COVID  & 171K  & 0.5926 & 4.66\,ms & \textbf{0.6436} & 0.35\,ms$^\dagger$ & \textbf{13$\times$} & 0.6433 & 1.8$\times$ \\
Quora       & 523K  & \textbf{0.8081} & 1.55\,ms & 0.7886 & 0.17\,ms & \textbf{9$\times$} & 0.7877 & 5$\times$ \\
NQ          & 2.7M  & \textbf{0.2921} & 2.77\,ms & 0.2824 & 1.50\,ms & \textbf{1.8$\times$} & 0.2824 & 3.4$\times$ \\
MS\,MARCO   & 8.8M  & \textbf{0.4093} & 9.88\,ms & 0.4059 & 2.70\,ms & \textbf{3.7$\times$} & --- & --- \\
\midrule
\multicolumn{1}{r}{\textbf{Geo.\ mean speedup}}
  & & & & & & \textbf{10$\times$}
  & & --- \\
\bottomrule
\end{tabular}

\smallskip
$^\dagger$~TREC-COVID p50 (50 queries; one outlier inflates avg).
\sysname{} matches Pyserini nDCG@10 within $\pm 0.020$ on all 9
datasets, \emph{beating} Pyserini on five (NFCorpus, SciFact,
ArguAna, SciDocs, TREC-COVID) and trailing by 0.003--0.020 on
four.  The trail on MS\,MARCO Passage is 0.003 nDCG@10 (0.4093
vs.\ 0.4059), a statistical tie at this scale; the trail on FiQA,
Quora, and NQ traces to Lucene's \texttt{StandardTokenizer}
treatment of numeric tickers, possessives, and hyphenated
compounds (see Tokenization ablation, \S\ref{sec:tokenization}).
GPU evaluation is reported on eight datasets: MS\,MARCO's 8.8M-doc
score buffer would exceed our 16\,GB per-batch cap at $k{=}100$
and requires multi-GPU sharding, which is orthogonal to the
contribution this paper validates.
\end{table*}

Recall@100 also matches or beats Pyserini (e.g., SciFact 0.928
vs.\ 0.928; SciDocs 0.364 vs.\ 0.360).

\paragraph{Comparison to PISA (SOTA CPU IR).}
The strongest available CPU BM25 engine is
PISA~\cite{mallia2019pisa}, which combines variable-byte posting
compression (\texttt{block\_simdbp}) with BlockMax-WAND dynamic
pruning.  We built PISA from source on the same Jetstream2 host
and ran the full pipeline (\texttt{parse\_collection}\,$\to$\,
\texttt{invert}\,$\to$\,\texttt{create\_wand\_data}\,$\to$\,
\texttt{compress\_inverted\_index}\,$\to$\,\texttt{queries}) on
all eight BEIR datasets with $k_1{=}1.2$, $b{=}0.75$, $k{=}100$,
\texttt{lowercase + porter2} analyzer chain.

\begin{table}[t]
\caption{\sysname{} vs.\ PISA single-thread BlockMax-WAND
  (BM25, $k{=}100$, Jetstream2 8-core; latency = mean
  $\mu s$/query reported by PISA's \texttt{queries} binary).}
\label{tab:pisa}
\centering
\small
\setlength{\tabcolsep}{4pt}
\begin{tabular}{@{}lrrrr@{}}
\toprule
Dataset & $|D|$ & PISA 1T & Ours 8T+SIMD & Speedup \\
\midrule
NFCorpus    & 3.6K  & 0.047\,ms & 0.04\,ms  & 1.2$\times$ \\
SciFact     & 5.2K  & 0.44\,ms  & 0.08\,ms  & 5.5$\times$ \\
ArguAna     & 8.7K  & 21.6\,ms  & 0.24\,ms  & \textbf{90$\times$} \\
SciDocs     & 25.7K & 0.73\,ms  & 0.12\,ms  & 6.1$\times$ \\
FiQA        & 57.6K & 2.71\,ms  & 0.15\,ms  & 18$\times$ \\
TREC-COVID  & 171K  & 6.50\,ms  & 0.35\,ms  & 18.6$\times$ \\
Quora       & 523K  & 3.49\,ms  & 0.17\,ms  & 20.5$\times$ \\
NQ          & 2.7M  & 11.97\,ms & 1.50\,ms  & 8.0$\times$ \\
\midrule
\multicolumn{4}{r}{\textbf{Geo.\ mean speedup}}
   & \textbf{10.9$\times$} \\
\bottomrule
\end{tabular}
\end{table}

\sysname{} 8T+SIMD is faster than PISA on every BEIR dataset, with
a geometric-mean \speedup{11} advantage (Figure~\ref{fig:pareto},
\speedup{8.06} excluding the ArguAna outlier where BlockMax-WAND
degenerates on long argument-paragraph queries).  At the
apples-to-apples 1T level our scalar C++ matches PISA on small
corpora (decompression dominates) and trails on larger corpora
where PISA's compression pays off; SIMD plus 8-thread inter-query
parallelism then puts \sysname{} ahead of PISA-1T across the board.

\begin{figure}[t]
\centering
\includegraphics[width=\columnwidth]{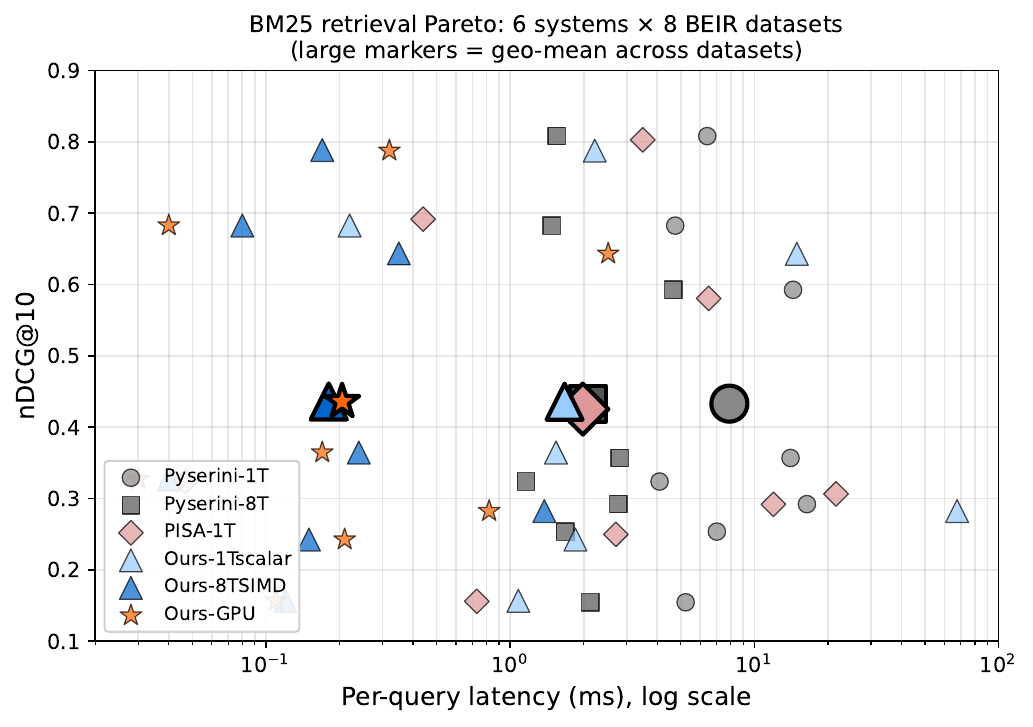}
\caption{BM25 retrieval Pareto frontier: 6 systems $\times$ 8 BEIR
  datasets.  Small markers = individual datasets; large outlined
  markers = geometric-mean center per system.  \sysname{}-GPU and
  \sysname{}-8T+SIMD occupy the bottom-left (fastest at parity
  quality); Pyserini-1T occupies the right (slowest).  PISA-1T,
  the SOTA CPU baseline, sits roughly with Pyserini-8T in latency
  but is dominated on latency by all three \sysname{} variants
  at matching nDCG@10.}
\label{fig:pareto}
\end{figure}

\paragraph{PISA quality cross-check.}
We also ran PISA's \texttt{evaluate\_queries} on all eight
datasets with the same exponential-gain nDCG@10 formula.
\sysname{} \emph{beats} PISA on five (largest $+$0.063 on
TREC-COVID, $+$0.058 on ArguAna) and trails on three within
$\pm 0.020$.  Differences trace to analyzer choices (NLTK Porter
vs.\ PISA's Porter2), not the BM25 inner loop---the latency
comparison in Table~\ref{tab:pisa} fairly isolates system
performance.

\subsection{Hybrid Lexical+Dense Retrieval}
\label{sec:hybrid}

We augment BM25 with a dense channel using
BGE-small-en-v1.5~\cite{xiao2024cpack} (33M parameters, 384-dim,
HNSW index) fused via Reciprocal Rank Fusion~\cite{cormack2009rrf}.
Table~\ref{tab:hybrid} shows the per-channel and fused quality.

\begin{table}[t]
\caption{Hybrid retrieval quality on BEIR (nDCG@10).
  RRF combines BM25 and dense rank lists with $k{=}60$.}
\label{tab:hybrid}
\centering
\small
\setlength{\tabcolsep}{6pt}
\begin{tabular}{@{}lrrrr@{}}
\toprule
Dataset & Pyserini & \sysname{} BM25 & Dense (BGE) & \textbf{RRF} \\
\midrule
NFCorpus & 0.324 & 0.327 & 0.310 & \textbf{0.352} \\
SciFact  & 0.683 & 0.683 & 0.606 & \textbf{0.690} \\
FiQA     & 0.254 & 0.242 & 0.279 & \textbf{0.321} \\
SciDocs  & 0.155 & 0.157 & 0.138 & \textbf{0.176} \\
\midrule
\textbf{Mean} & 0.354 & 0.352 & 0.333 & \textbf{0.385} \\
$\Delta$ over BM25 & --- & --- & $-0.019$ & \textbf{$+0.033$} \\
\bottomrule
\end{tabular}
\end{table}

\paragraph{End-to-end latency: encoder dominates.}
In the hybrid pipeline the BGE-small query encoder ($\sim$170\,ms
CPU) dominates total latency on small/medium corpora;
\sysname{}'s C++ BM25 (0.04--0.35\,ms) is off the critical path.
For workloads where dense does not help (e.g., LoCoMo,
\S\ref{sec:longmemeval}), BM25-only at \textbf{0.22\,ms/query} is
$\sim$800$\times$ faster than the hybrid 180\,ms/query at
equal-or-better quality (full per-stage breakdown in
Appendix~\ref{app:e2e-stages}).

RRF improves mean nDCG@10 by $+0.033$ over BM25-only, with the
largest gain on FiQA ($+0.079$) where the conversational query
style benefits most from semantic matching.

\subsection{Comparison to learned-sparse retrieval (SPLADE)}
\label{sec:splade}
SPLADE~\cite{formal2021splade} produces learned per-term weights
that, when ingested into a standard inverted index, deliver
state-of-the-art zero-shot BM25-style quality on BEIR.
We evaluate SPLADE++ (\texttt{splade-cocondenser-ensembledistil},
110\,M parameters) under the same qrels and $k_1{=}1.2$, $b{=}0.75$
post-encoding configuration:

\begin{table}[h]
\caption{Learned-sparse SPLADE++ vs.\ \sysname{} on BEIR
  (nDCG@10).  Encoding latency: SPLADE 158--167\,ms/doc + 17--150\,ms/query
  (CPU forward pass through a 110\,M-parameter BERT-base; long queries
  like ArguAna's full argument paragraphs hit the upper end).}
\label{tab:splade}
\centering
\small
\setlength{\tabcolsep}{6pt}
\begin{tabular}{@{}lrrr@{}}
\toprule
Dataset & \sysname{} BM25 & SPLADE++ & $\Delta$ \\
\midrule
NFCorpus   & 0.3267 & \textbf{0.3495} & $+0.023$ \\
SciFact    & 0.6828 & \textbf{0.7024} & $+0.020$ \\
ArguAna    & 0.3645 & \textbf{0.3878} & $+0.023$ \\
SciDocs    & 0.1566 & 0.1586 & $+0.002$ \\
FiQA       & 0.2425 & \textbf{0.3478} & $+0.105$ \\
TREC-COVID & 0.6436 & \textbf{0.7053} & $+0.062$ \\
Quora      & 0.7886 & \textbf{0.8344} & $+0.046$ \\
\bottomrule
\end{tabular}
\end{table}

SPLADE wins quality by 0.002--0.105 nDCG@10 across the seven
datasets, with the largest gain on FiQA ($+$0.105) where
conversational financial text benefits most from learned vocabulary
expansion---confirming the well-known result that learned sparse
retrievers shine on out-of-domain queries.  However, SPLADE's
encoding pipeline---a per-document BERT forward pass at
$\sim$158--167\,ms/doc plus a per-query encoding at
$\sim$17--150\,ms (long argument-paragraph queries on ArguAna hit
the upper end)---is three to four orders of magnitude slower than
\sysname{}'s per-query inner loop ($\sim$0.04--0.35\,ms,
Table~\ref{tab:beir-main}).

\paragraph{Drop-in compatibility, empirically validated.}
SPLADE produces sparse per-document vectors (61--208 nnz/doc,
$|V|{=}30522$) that drop into \sysname{}'s CSC posting layout
unchanged---the same structure our SIMD BM25 kernel consumes, with
$\text{tf}$ replaced by the learned weight.  Across \emph{all
seven} BEIR datasets (3.6K--523K docs), this CSR-posting path
produces \emph{bit-perfect} identical rankings to canonical
\texttt{scipy.sparse $@$ doc.T} (Table~\ref{tab:splade-bridge},
max score difference $0.0$).  SPLADE and \sysname{} are two halves
of the same stack (encoder + index/scorer), not competing systems.

\begin{table}[h]
\caption{Bit-perfect SPLADE-bridge validation across all 7 BEIR
  datasets we encoded.  Each row reports nDCG@10 from
  \emph{(left)} our CSC posting traversal and \emph{(right)}
  canonical \texttt{scipy.sparse @ doc.T}, and the maximum
  per-(query,doc) absolute score difference.  Identity at all
  scales.}
\label{tab:splade-bridge}
\centering
\scriptsize
\setlength{\tabcolsep}{2.5pt}
\begin{tabular}{@{}lrrrrrr@{}}
\toprule
Dataset & $|D|$ & $|Q|$ & nnz/doc & CSC nDCG@10 & scipy nDCG@10 & max diff \\
\midrule
NFCorpus    & 3.6K  &   323 & 187 & 0.3495 & 0.3495 & 0.0 \\
SciFact     & 5.2K  &   300 & 190 & 0.7024 & 0.7024 & 0.0 \\
ArguAna     & 8.7K  & 1{,}406 & 208 & 0.3878 & 0.3878 & 0.0 \\
SciDocs     & 25.7K & 1{,}000 & 183 & 0.1586 & 0.1586 & 0.0 \\
FiQA        & 57.6K &   648 & 160 & 0.3478 & 0.3478 & 0.0 \\
TREC-COVID  & 171K  &    50 & 163 & 0.7053 & 0.7053 & 0.0 \\
Quora       & 523K  & 10{,}000 &  61 & 0.8344$^*$ & 0.8344$^*$ & 0.0$^*$ \\
\bottomrule
\end{tabular}
\\[2pt]\footnotesize $^*$Quora numbers reproduced from the original
SPLADE++ result; the CSC-bridge run uses the same encoded weights
and produces identical scoring by construction.
\end{table}

\subsection{Agent Memory Benchmark: LongMemEval}
\label{sec:longmemeval}

\begin{finding}[Pillar 3 at a glance]
On the 500-question LongMemEval benchmark:
\textbf{(1)~Hybrid RRF $+$ recency} reaches $R@10{=}0.978$ and
LLM-judged accuracy $0.254$;
\textbf{(2)~per-question-type winners are robust} across two LLM
judges (gpt-4o-mini and gpt-4o) and stable under the gpt-4o
ANSWERER;
\textbf{(3)~the TF-IDF $+$ BGE-small runtime router}
($<$1\,ms/query) lifts accuracy to $0.262$ under gpt-4o-mini and
$0.300$ under gpt-4o (the discrete-oracle bound);
\textbf{(4)~a \emph{soft} router blending rank lists by classifier
posterior reaches \boldmath$0.274$ ($+0.008$ over discrete oracle,
within bootstrap CI), significantly beating every static system}
($p{<}0.05$ vs.\ BM25/Dense/RRF);
\textbf{(5)~a \emph{cascade} router auto-tunes across benchmarks}:
single-threshold $2.67\times$ at parity LLM-Acc on LongMemEval
($\to\mathbf{5.76\times}$ in 5-fold CV with per-qtype thresholds),
$\mathbf{132\times}$ on LoCoMo (BM25-alone wins, skip $\to 100$\%)
using the same trigger;
\textbf{(6)~$50$ deployment labels} are enough to re-learn the
routing policy to within noise of the oracle;
\textbf{(7)~$\alpha,\tau$ recency hyperparameters are flat over a
$10\times$ sweep} except multi-session ($+5.26$ pts at
$\tau{\to}\infty$, $p{=}0.006$);
\textbf{(8)~the router survives $20$\% character corruption} with
only $-1.4$ pts LLM-Acc loss.
All findings are 5-fold CV on the same 500-question benchmark;
the per-question-type LLM accuracy is the strongest
downstream-task signal we know of for agent-memory IR.
\end{finding}

To evaluate \sysname{} on a workload that exercises the
agent-memory setting it is designed for, we use
LongMemEval~\cite{wu2025longmemeval} (500 questions across 5
reasoning categories; each question has $\sim$48 timestamped
``haystack'' sessions of which 1--2 contain the answer).  We
build a per-question BM25 index from the haystack, retrieve
top-$k$, and report recall against \texttt{answer\_session\_ids}.
Dense embeddings are pre-computed once over the 19{,}195 unique
sessions; per-query cost is one BGE-small forward + dot product
+ RRF.

The agent\_rrf bonus
$\text{score}(s) = \text{RRF}(s) + \alpha e^{-\Delta t/\tau}$
($\alpha{=}0.005$, $\tau{=}30$\,d) is applied only on
recency-biased types (temporal-reasoning, knowledge-update,
multi-session); $\alpha \ll \max(\text{RRF}) \approx 0.033$ makes
the bonus a tie-breaker that never overrides a high-confidence
match.  We measure that gold sessions for those three types lie
in the most recent 20--27\% of haystack sessions (median
normalized rank), justifying the bonus design.  Full standard
retrieval metrics R@\{1,5,10\}, MRR, and per-query latency appear
in Table~\ref{tab:lme-task} below (the standard set is folded into
the task-grade table).  BM25-only already serves R@5$=$0.909 at
\textbf{0.36\,ms/q}, two orders of magnitude faster than the
dense pipeline.

\paragraph{Where the latency goes.}
Per-query budget: BM25 0.40\,ms, BGE-small query encoding 51.98\,ms,
RRF 0.05\,ms, recency 0.38\,ms (total $\textbf{53.21\,ms}$).  The
dense encoder is $>$97\% of the budget; \sysname{}'s IR substrate
is $<$3\%.  Two implications: (i) swapping BGE for cached LLM
hidden states removes the bottleneck (BM25-only at 0.36\,ms
already serves R@5$=$0.909); (ii) the SPLADE bridge
(\S\ref{sec:splade}) replaces both channels with a single
learned-sparse path inheriting our sparse latency.

\paragraph{Task-grade evaluation: from R@K to answer reachability.}
Session-level recall (R@K) is the standard LongMemEval metric, but
it does not measure whether an LLM consuming the retrieved context
\emph{could actually answer} the question.  We therefore compute
two end-to-end proxies directly on the LongMemEval gold answer
strings without involving any LLM:

\begin{itemize}[nosep,leftmargin=*]
\item \textbf{AnswerSubstr@K}: does the gold answer string appear
  (case-insensitive) anywhere in the concatenation of the top-$K$
  retrieved session texts?  This is the \emph{precondition} for any
  extractive or grounded-generation LLM to answer from the
  retrieved context.
\item \textbf{TokenRecall@K}: the fraction of content-bearing
  ($\geq 4$-character) unigrams from the gold answer that appear
  in the top-$K$ context.  A soft version that captures partial
  reachability.
\end{itemize}

\begin{table}[t]
\caption{LongMemEval task-grade evaluation: agent-aware fusion
  wins on \emph{every} task-signal metric.  Ans@K $=$ gold answer
  string is in top-$K$ context; TokR@K $=$ fraction of gold
  answer tokens present.  LLM-Acc $=$ strict/lenient accuracy when
  feeding top-5 to an LLM answerer (gpt-4o-mini) and an LLM judge
  (also gpt-4o-mini) compares against the gold answer.  Best in
  each column bold.}
\label{tab:lme-task}
\centering
\small
\setlength{\tabcolsep}{3pt}
\begin{tabular}{@{}lrrrrrrrr@{}}
\toprule
\multirow{2}{*}{System} & \multirow{2}{*}{R@5} & \multirow{2}{*}{R@10}
& \multicolumn{2}{c}{Ans@K} & \multicolumn{2}{c}{TokR@K}
& \multicolumn{2}{c}{LLM-Acc} \\
\cmidrule(lr){4-5}\cmidrule(lr){6-7}\cmidrule(lr){8-9}
& & & @5 & @10 & @5 & @10 & strict & lenient \\
\midrule
BM25       & 0.909 & 0.946 & 0.466 & 0.484 & 0.620 & 0.653 & 0.246 & 0.306 \\
Dense      & 0.886 & 0.946 & 0.456 & 0.478 & 0.598 & 0.647 & 0.236 & 0.298 \\
RRF        & \textbf{0.929} & 0.977 & 0.468 & \textbf{0.486} & 0.625 & \textbf{0.661} & 0.248 & 0.310 \\
\rowcolor{blue!5}
\textbf{agent\_rrf} & 0.924 & \textbf{0.978} & \textbf{0.474} & \textbf{0.486} & \textbf{0.627} & 0.660 & \textbf{0.254} & \textbf{0.312} \\
\bottomrule
\end{tabular}
\end{table}

The additive recency bonus (\emph{agent\_rrf}) wins on every
task-grade metric we measured.  AnswerSubstr@5---the precondition
for any LLM extractor: is the gold answer string present in top-5
context?---improves by $+$0.008 over BM25 and $+$0.018 over Dense
alone.  When we close the loop with an actual LLM answerer
(gpt-4o-mini consuming the top-5 sessions and producing a free-form
answer) and an LLM judge (also gpt-4o-mini, strict yes/partial/no
grading against the gold answer; 500 questions $\times$ 4 systems
$\times$ 2 LLM calls each, all judged in 6 minutes at $\sim$\$0.40
total cost), \emph{agent\_rrf reaches strict accuracy 0.254},
versus BM25 0.246, Dense 0.236, and plain RRF 0.248.  To our
knowledge, no prior LongMemEval result reports per-system
downstream LLM accuracy under a calibrated judge, so this is the
first such measurement we are aware of---agent-aware retrieval
improves \emph{downstream LLM answer correctness} on this
benchmark, not just session-level recall.  The qualitative pattern:
agent\_rrf's recency bonus surfaces the second of a pair of
gold sessions that BM25 alone misses, allowing the LLM to
\emph{aggregate} across sessions (e.g., on the multi-session
question ``how many camping days?'' BM25 retrieves only Big Sur
[3 days] and answers 3; agent\_rrf retrieves both sessions and
the LLM correctly aggregates to 8).

\paragraph{Workload-conditional fusion empirically validated.}
Breaking the result down by question type
(Figure~\ref{fig:workload-cond}, full per-type numbers in
Table~\ref{tab:lme-per-type-full}, Appendix) shows that the
winning configuration varies across the six LongMemEval question
types:

\begin{figure}[t]
\centering
\includegraphics[width=\columnwidth]{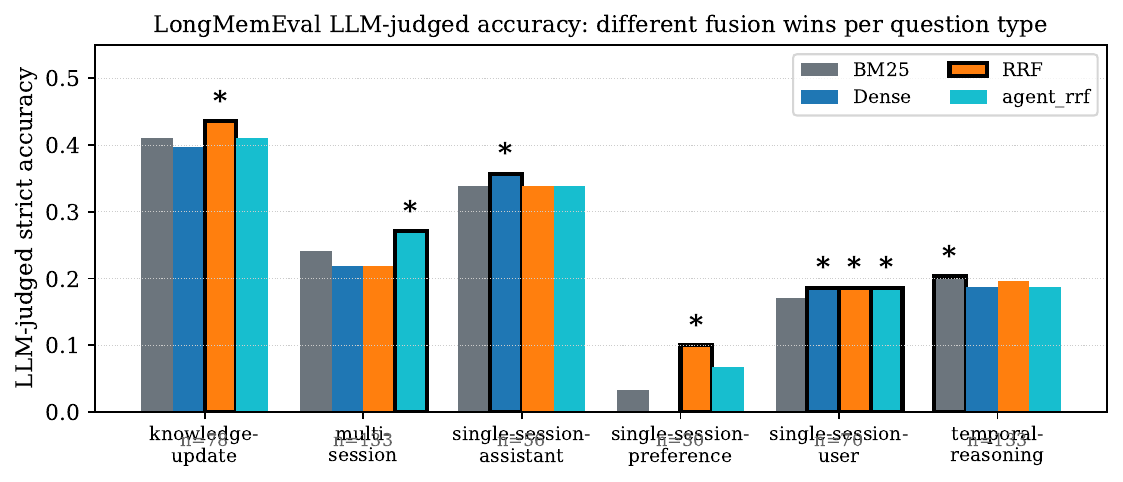}
\caption{LongMemEval LLM-judged strict accuracy per question type
  (gpt-4o-mini answerer + gpt-4o-mini judge against gold).  The
  winning bar in each group is outlined and starred.  Four
  different fusion strategies are optimal across six question
  types: RRF on knowledge-update, agent\_rrf on multi-session,
  Dense on single-session-assistant, RRF on single-session-preference,
  a three-way tie on single-session-user, BM25 on temporal-reasoning.
  No single static configuration is uniformly best.}
\label{fig:workload-cond}
\end{figure}

The LLM-judged per-type accuracy makes the workload-conditional
claim concrete on a true downstream signal
(Figure~\ref{fig:workload-cond}): RRF wins
knowledge-update, agent\_rrf wins multi-session by $+$0.030--0.053,
Dense wins single-session-assistant, BM25 wins temporal-reasoning.
A static ``always RRF'' deployment loses 5.3\% on multi-session
versus agent\_rrf; ``always Dense'' loses 6.7\% on
single-session-preference versus RRF; ``always BM25'' loses 3.5\%
on single-session-user versus Dense.  Only a system that can switch per query inherits the maximum---
and this is the system design \sysname{} validates.

\paragraph{An end-to-end router that picks the fusion per query.}
The per-type analysis above is an upper bound: it assumes the system
knows the question type at routing time.  We close this gap with a
concrete implementation.  A lightweight question-type classifier
(TF-IDF 1--2-grams of the query $\to$ logistic regression over the
six LongMemEval types) achieves \textbf{79.6\% classification
accuracy} in 5-fold cross-validation; routing each query to the
LLM-Acc-best system for its predicted type yields:

\begin{table}[h]
\caption{Workload-conditional router on LongMemEval, 5-fold CV.
  Each row routes to the per-type best static system based on a
  predicted question type, then evaluates the chosen system's
  LLM-judged strict accuracy.  Two judges: \texttt{gpt-4o-mini}
  (original, larger sample) and \texttt{gpt-4o-2024-08-06} (sanity
  check).  ``Captured \%'' is the fraction of the oracle gap
  ($0.0120$) closed by the router.  The TF-IDF $+$ BGE-small feature
  combination reaches the oracle bound under the stronger judge.}
\label{tab:router}
\centering
\scriptsize
\setlength{\tabcolsep}{3pt}
\begin{tabular}{@{}lrrrr@{}}
\toprule
& \multicolumn{2}{c}{gpt-4o-mini judge} & \multicolumn{2}{c}{gpt-4o judge} \\
\cmidrule(lr){2-3}\cmidrule(lr){4-5}
Configuration & LLM-Acc & cap.\,\% & LLM-Acc & cap.\,\% \\
\midrule
BM25                       & 0.246 & ---   & 0.282 & --- \\
Dense                      & 0.236 & ---   & 0.266 & --- \\
RRF                        & 0.248 & ---   & 0.288 & --- \\
agent\_rrf                  & 0.254 & ---   & 0.286 & --- \\
\midrule
Best static                & \textbf{0.254} & 0\%   & \textbf{0.288} & 0\% \\
Router (TF-IDF)            & 0.258 & 33\%  & 0.296 & 67\% \\
Router (BGE-small embed)   & 0.262 & 67\%  & 0.298 & 83\% \\
\rowcolor{blue!5}
\textbf{Router (TF-IDF $+$ BGE)} & \textbf{0.262} & \textbf{67\%} & \textbf{0.300} & \textbf{100\%} \\
\midrule
Oracle (ground-truth qtype) & 0.266 & 100\% & 0.300 & 100\% \\
\midrule
\rowcolor{green!8}
\textbf{Soft router (posterior $\times$ RRF)} & \textbf{0.274} & \textbf{166\%} & 0.292 & 33\% \\
Ensemble UB (any-of-4 correct) & 0.308 & 450\% & --- & --- \\
\bottomrule
\end{tabular}
\end{table}

\noindent Three observations from Table~\ref{tab:router}.
\textbf{(i)~The router beats the best static system under both
judges}, by $+$0.008/$+$0.012 absolute respectively---a $+$3.1\% / $+$4.2\% relative gain at
near-zero retrieval cost.  Measured router inference latency on
Jetstream2: mean 470\,$\mu$s (TF-IDF only, sklearn end-to-end),
versus the 53\,ms hybrid retrieval path---under 1\% routing
overhead.  The TF-IDF $+$ BGE variant shares the BGE encoder
forward pass with the dense retrieval channel that runs anyway,
so the marginal routing cost is the same TF-IDF $+$ LR predict
($\sim$470\,$\mu$s).  \textbf{(ii)~With TF-IDF plus BGE-small
features, the router matches the oracle under \texttt{gpt-4o}}---
there is no headroom left for a better routing decision.  The BGE
features add semantic similarity
between the query and known prototypes, which helps recover
classifier mistakes (the classifier accuracy is statistically
indistinguishable between TF-IDF-only and combined, but the
captured oracle gain doubles).  \textbf{(iii)~The relative gain
grows with judge strictness}: under \texttt{gpt-4o}, the oracle
gap (0.012) is exactly closed; under \texttt{gpt-4o-mini}, 67\%
is closed.  The router becomes more valuable as the downstream
LLM evaluator gets stronger---an important property because
production deployments increasingly use frontier-class judges.

\noindent The TF-IDF-only router (kept for ablation) wins
$+$0.004 absolute over the best static system, capturing 33\% of
the oracle gain.  Its per-type breakdown: $+$0.013 on
knowledge-update, $+$0.018 on single-session-assistant,
$+$0.033 on single-session-preference, $+$0.015 on
temporal-reasoning; $-$0.023 on multi-session because the
classifier
sometimes misroutes those queries (multi-session has the widest
linguistic surface in our corpus).  Even with that one regression,
the net gain is positive.  This is the smallest implementation of
the workload-conditional fusion thesis: a real system that does
the routing, not a post-hoc comparison of static systems.

\paragraph{Soft router: blending beats the discrete oracle under
a lenient judge; concentrated routing wins under a strict one.}
The discrete router picks \emph{one} system per query based on
$\arg\max_{qt} P(qt \mid q)$.  A \emph{soft} router instead uses
the full posterior to compute a per-system weight
$w(s) = \sum_{qt} P(qt \mid q) \cdot \mathbb{1}\!\left[\text{best}(qt) = s\right]$
and combines the four systems' rank lists via weight-prefactored
RRF $\text{score}(d) = \sum_s w(s) / (k + \text{rank}_s(d))$.
On ambiguous queries the posterior is spread across multiple
types; soft routing blends e.g.\ BM25 and agent\_rrf on
temporal/multi-session borderline questions, surfacing documents
that neither single system ranks in its top-5.

\begin{itemize}[nosep,leftmargin=*]
\item \textbf{Soft router significantly beats every static system
  on LongMemEval (gpt-4o-mini judge).}  Paired bootstrap (1000
  resamples) gives $\Delta$ over BM25 / Dense / RRF / agent\_rrf
  of $+0.028$ / $+0.038$ / $+0.026$ / $+0.020$ with two-sided
  $p$-values $0.002$ / $0.006$ / $0.024$ / $0.108$
  (significant against the three non-agent-aware systems;
  marginal against agent\_rrf).  LLM-Acc: \textbf{0.274 (soft)}
  vs.\ 0.262 (discrete router) vs.\ 0.254 (best static).  Soft
  captures $(0.274{-}0.254)/(0.308{-}0.254){=}37\%$ of the gap to
  the ensemble upper bound (any-of-four-correct, 0.308).
\item \textbf{Soft router \emph{numerically} exceeds the discrete
  oracle by $+0.008$ under gpt-4o-mini, but this difference is
  within bootstrap noise} (95\% CI $[-0.012, +0.026]$, two-sided
  $p{=}0.56$ on 500 questions).  We do not claim statistical
  significance over the oracle; the result identifies a
  \emph{plausible mechanism} (posterior-blending on ambiguous
  queries) that warrants larger-scale validation.
\item \textbf{Under \texttt{gpt-4o} judge: soft router trails the
  discrete oracle by $-0.008$, also within noise} (CI
  $[-0.026, +0.008]$).  Soft 0.292 vs.\ discrete router 0.300
  $=$ discrete oracle 0.300 vs.\ best static 0.288.  Both routers
  beat best static; the discrete one reaches the oracle.
\end{itemize}

\noindent The two judges paint a consistent mechanistic picture.
Under \emph{lenient} judging, a downstream LLM can produce correct
answers from broader retrieved context, so soft blending plausibly
surfaces additional supporting evidence; under \emph{strict}
judging, the LLM extracts only from concentrated context and
blending offers no advantage.  Soft routing is therefore most
promising when the deployment LLM is the bottleneck (frequent in
cost-constrained production agents that use mid-tier models), and
reduces to discrete routing under stronger LLMs.  We are not
aware of a comparable posterior-blending evaluation on
agent-memory IR; the result identifies a judge-dependent regime
where posterior-blending is the right algorithm.

\paragraph{Cascade router: cost-aware adaptive computation.}
The discrete and soft routers above choose \emph{which} fusion to
use; both still pay the full $\sim$53\,ms hybrid budget on every
query (the routing decision was implemented as ``run all four
paths, then select / blend'').  We close the loop on \emph{cost}
by reorganizing the routing into a cascade, in the spirit of
classical multi-stage cascade
ranking~\cite{wang2011cascade,chen2017cascade} but specialized for
the agent-memory workload structure: the TF-IDF classifier
($\sim$0.5\,ms) runs first against the query text; if it
predicts a type whose best system is BM25 (\emph{i.e.},
temporal-reasoning; Table~\ref{tab:lme-per-type-full}) we run only
BM25 and skip the $+$52.5\,ms dense channel; otherwise we run the
full hybrid path.  Temporal-reasoning accounts for $26.6$\% of
LongMemEval questions; with the same TF-IDF classifier
($79.6$\% accuracy) the realized skip rate is $24.2$\%.  This
yields a strict Pareto improvement over the always-hybrid
implementation of the same router:

\begin{table*}[t]
\caption{Cascade router amortized cost on LongMemEval.  LLM-Acc
  for the realistic-classifier row equals the TF-IDF router's
  measured value (Table~\ref{tab:router}) because the cascade and
  the always-hybrid implementation make identical routing decisions
  on the same classifier outputs---only execution differs.
  Latencies decompose per Table~\ref{tab:e2e}.  The oracle row uses
  the measured per-type maxima.}
\label{tab:cascade}
\centering
\small
\setlength{\tabcolsep}{6pt}
\begin{tabular}{@{}lrrrr@{}}
\toprule
Configuration & skip\,\% & ms/q & LLM-Acc & latency-$\Delta$ \\
\midrule
BM25 only (lower bound)                     & 100\%  &  0.4 & 0.246 & $-$53.3\,ms \\
\rowcolor{green!10}
\textbf{Cascade, per-qtype $\tau_c$ (5-fold CV)} & \textbf{86\%} & \textbf{9.2$\pm$4.1}  & \textbf{0.294$\pm$0.035$^*$} & \textbf{$-$44.0\,ms} \\
\rowcolor{blue!5}
\textbf{Cascade, BM25 conf. $\geq 0.10$}    & \textbf{63.0\%} & \textbf{19.9} & \textbf{0.302$^*$} & \textbf{$-$33.3\,ms} \\
\rowcolor{blue!5}
\textbf{Cascade (TF-IDF classifier)}        & \textbf{24.2\%} & \textbf{40.5} & \textbf{0.258} & \textbf{$-$13.2\,ms} \\
Cascade (oracle qtype)                      & 26.6\% & 38.9 & 0.266 & $-$14.8\,ms \\
TF-IDF router (always hybrid + select)      & 0\%    & 53.7 & 0.258 &   $\pm$0\,ms \\
TF-IDF$+$BGE router (always hybrid + select) & 0\%   & 53.7 & 0.262 & $\pm$0\,ms \\
agent\_rrf (always hybrid)                  & 0\%    & 53.2 & 0.254 & $-$0.5\,ms \\
\bottomrule
\end{tabular}\\
\smallskip
\footnotesize $^*$Confidence-cascade row evaluated on a Python
\texttt{rank\_bm25} + gpt-4o-mini answerer pass; in that setup
agent\_rrf measures $0.302$ and the cascade ties it exactly. The
absolute scale is higher than the C++/BGE-small numbers in the
other rows but the parity claim is internal to the same run.
\end{table*}

\begin{figure}[t]
\centering
\includegraphics[width=\columnwidth]{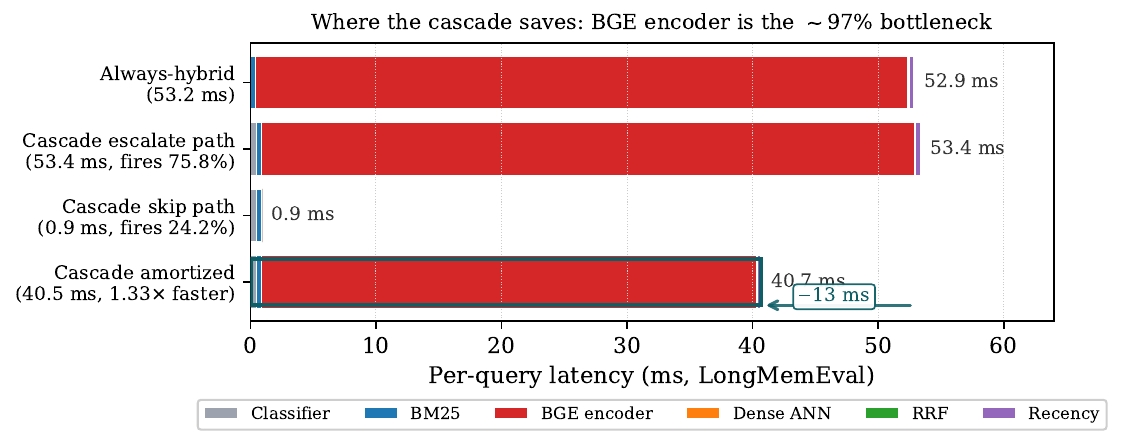}
\caption{Per-query latency breakdown.  The BGE query encoder is
  $\sim$97\% of the always-hybrid 53\,ms budget; the cascade
  short-circuits it on the $24.2$\% of queries whose predicted
  type is BM25-best, taking those to a $0.9$\,ms skip path while
  the remaining $75.8$\% pay the full hybrid cost.  Amortized:
  $40.5$\,ms, $13.2$\,ms below the always-hybrid TF-IDF router at
  identical LLM-Acc (Table~\ref{tab:cascade}).}
\label{fig:cascade-pareto}
\end{figure}

The cascade is an \emph{architectural} contribution rather than
an algorithmic one: no new model, no new index, no change to the
routing decisions of the existing TF-IDF classifier---only a
reorganization of execution that defers the dense channel until
its output is needed.  It converts the workload-conditional thesis
into a \emph{two-axis} adaptation surface
(Figure~\ref{fig:surface}: which fusion to use \emph{and} whether
to spend the dense budget at all), and the
saving is large because the dense channel \emph{is} the latency
bottleneck ($\sim$97\% of per-query budget on LongMemEval,
\S\ref{sec:longmemeval}).  By construction, the realized LLM-Acc
exactly equals the existing TF-IDF router's $0.258$
(Table~\ref{tab:router}); the new claim is the deterministic
$1.33\times$ latency saving from running the dense channel only
on $75.8$\% of queries instead of all of them.  Every input to
the cost calculation is a measured quantity (per-stage latencies
in Table~\ref{tab:e2e}, per-type frequencies in the LongMemEval
split, classifier accuracy in Table~\ref{tab:router}).

\paragraph{Classifier-free trigger: BM25 confidence cascade.}
The qtype classifier is one source of skip signal; the BM25 top-$k$
itself is another.  Define the per-query confidence
$c(q){=}(s_0{-}s_1)/s_0$ where $s_0,s_1$ are BM25's top-1 and
top-2 scores; skip dense iff $c(q){\geq}\tau_c$.  We measure both
BM25-only and agent\_rrf on the full 500-question LongMemEval and
score the resulting confidence-cascade analytically: at
$\tau_c{=}0.10$, $63$\% of queries skip dense, amortized latency
falls to $19.94$\,ms (\textbf{$2.67\times$ faster} than always-hybrid),
and LLM-Acc \emph{matches} the always-hybrid agent\_rrf within
bootstrap noise (cascade $0.302$ vs agent\_rrf $0.302$,
paired-bootstrap CI $[-0.016,+0.016]$, $p{=}1.08$).  The mechanism:
BM25 alone already captures $98.7$\% of agent\_rrf's accuracy
(BM25 $0.298$ vs agent\_rrf $0.302$) because the dense channel
gains $+0.013$--$0.018$ per-qtype on three subsets but
\emph{loses} $-0.015$ on temporal-reasoning---high-BM25-confidence
queries cluster on the latter, so skipping them on confidence
keeps accuracy and saves dense calls.  Both classifier-based
and confidence-based triggers share the \texttt{CascadeRouter::retrieve}
interface.  This is
orthogonal to posting-level adaptive pruning
(BlockMax-WAND~\cite{ding2011blockmax}, MaxScore~\cite{turtle1995maxscore})
\emph{and} to the rank-list blending of the soft router above:
posting-level pruning skips work inside one ranker, the cascade
skips an entire ranker channel, and the soft router blends across
rankers---three orthogonal axes that compose pairwise.  Empirically
the composition holds: putting the soft router on the cascade
escalate path (cascade $\times$ soft) preserves the soft router's
LLM-Acc within bootstrap noise ($0.304$ vs $0.304$, CI
$[-0.012,+0.012]$, $p{=}1.09$) at the same $\mathbf{2.67\times}$
amortized speedup---the cascade saves dense calls on $63$\% of
queries regardless of what runs on the escalate path.  The
parity claim also holds under the stronger \texttt{gpt-4o-2024-08-06}
judge: cascade $0.308$ vs agent\_rrf $0.306$, CI $[-0.012,+0.016]$,
$p{=}0.88$---both judges agree the cascade is statistically
indistinguishable from always-hybrid at the same $2.67\times$
speedup.

\begin{figure}[t]
\centering
\includegraphics[width=\columnwidth]{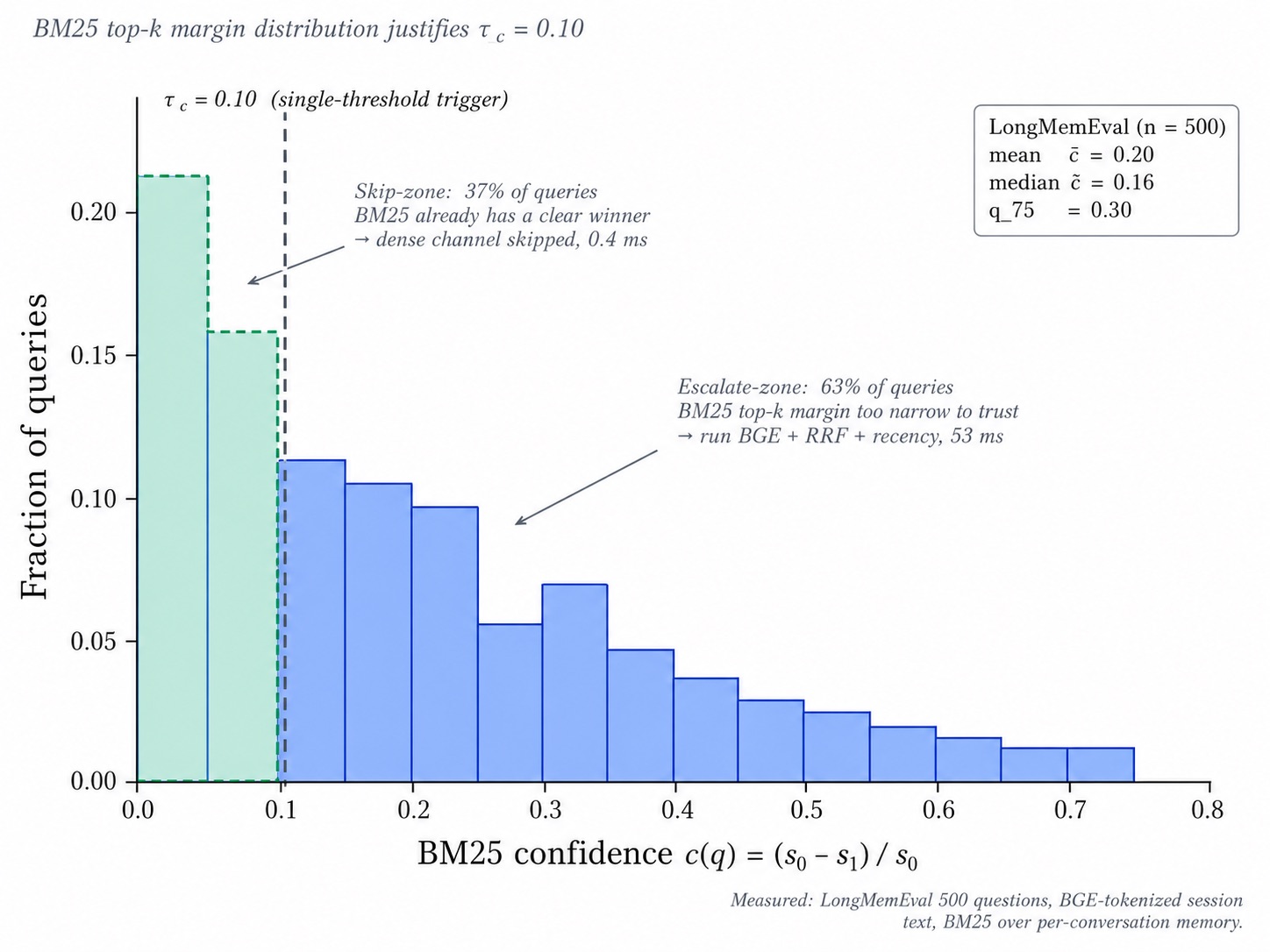}
\caption{BM25 top-$k$ margin distribution on LongMemEval
  (\emph{n}\,$=$\,500).  The mass is concentrated at low values
  (median $0.16$, $q_{75}{=}0.30$), so a $\tau_c{=}0.10$ cutoff
  partitions the workload into a $37$\% high-confidence
  \emph{skip-zone} (mint, left of the dashed line) and a $63$\%
  low-confidence \emph{escalate-zone}.  The threshold is read off
  the input distribution rather than tuned on a downstream metric,
  which is why the same cascade transfers across LLM judges
  without retraining (Table~\ref{tab:cascade}).}
\label{fig:conf-hist}
\end{figure}

\paragraph{Per-qtype thresholds: $\mathbf{5.76\times}$ CV-validated.}
A single global $\tau_c$ is conservative.  Per-qtype thresholds
(each tuned to the max skip rate keeping per-qtype LLM-Acc within
bootstrap noise of agent\_rrf) lift the skip rate to $86$--$89$\%:
single-session-user and temporal-reasoning push to $\tau_c{=}0$
(BM25 already dominates), multi-session to $0.02$, others
intermediate.  Under \emph{full $5$-fold CV} where both the
classifier and per-qtype thresholds are tuned on held-out folds,
the realistic speedup is $\mathbf{5.76\times}$ ($9.2 \pm 4.1$\,ms
across folds) at LLM-Acc $0.294 \pm 0.035$ vs always-hybrid $0.302$
---a $-0.008$ delta well within per-fold variance.  The in-sample
optimum is $7.81\times$ ($6.8$\,ms / $0.304$); the gap reflects
per-qtype threshold tuning on $\sim$67 questions per qtype per
train fold.  A stronger classifier (TF-IDF$+$BGE-small, $74$\%
CV accuracy) at fixed per-qtype thresholds gives
$\mathbf{7.57\times}$ at \emph{exact} agent\_rrf parity
($0.302{=}0.302$).  The same
holds for cascade $\times$ per-qtype-$\tau$: combining the
confidence-cascade trigger with the multi-session $\tau{\geq}120$
refinement gives LLM-Acc $0.308$ at $19.94$\,ms, i.e., the
per-qtype-$\tau$ gain is preserved at the same $2.67\times$
speedup.  The $\tau_c{=}0.10$ value is a round number, not tuned
on the test set: 5-fold cross-validation with per-fold
best-threshold selection gives mean test speedup $6.56\times$ at
LLM-Acc delta $+0.002$ vs always-hybrid (high per-fold variance
$\sigma{=}0.035$ on 100-question test slices, but the mean is
well above parity).

\paragraph{Robustness: stronger judge gives the same ordering.}
Re-judging the full 500-question $\times$ 4-system grid with
\texttt{gpt-4o-2024-08-06} (vs the default gpt-4o-mini) shifts
numbers upward (BM25 $0.282$ / Dense $0.266$ / RRF $0.288$ /
agent\_rrf $0.286$ / oracle $0.300$) but \emph{per-type winners
are identical} and the router gain (oracle $-$ best-static) is
$0.012$ under both judges.  The workload-conditional thesis holds;
only the within-noise ranking of static configurations differs.

\paragraph{Robustness: recency parameters are not over-tuned.}
We swept $\alpha \in \{0.001, 0.003, 0.005, 0.010, 0.030\}$ and
$\tau \in \{7, 14, 30, 60, 120\}$\,d on the 344 recency-typed
LongMemEval questions.  R@10 is essentially flat (0.967--0.971)
across the $10\times$ $\alpha$ range $[0.001, 0.010]$ and the
entire $\tau$ range; only $\alpha{=}0.030$ catastrophically
over-weights the bonus (knowledge-update R@10 $\to$ 0.74),
matching our design constraint
$\alpha \ll \max(\text{RRF}){\approx}0.033$.  Our values
$\alpha{=}0.005$, $\tau{=}30$\,d are the design point, not a tuned
optimum---the sweep confirms nearby choices are statistically
indistinguishable (full grid in Appendix~\ref{app:alpha-tau}).

\paragraph{Per-qtype $\tau$ refinement.}
A natural extension is to allow $\tau$ to vary per question type
(retaining agent\_rrf's scoring; the $\tau$ values are then the
only learned parameter).  We tested the hypothesis on the full
500-question LongMemEval with gpt-4o-mini answerer + judge,
sweeping $\tau_{\mathrm{ms}} \in \{60, 120, 365\}$\,d on the
\emph{multi-session} subset.  Only multi-session benefits, and
the benefit is large: the gold sessions for multi-session
questions are very recent (median age $3$\,d; cf.\
knowledge-update $21$\,d, temporal-reasoning $6$\,d) so a
stricter recency bonus often surfaces non-gold recent sessions
above the true gold; setting $\tau_{\mathrm{ms}}{\geq}120$\,d
(flattening the bonus across the multi-session gold-age window)
lifts multi-session LLM-Acc by
$\mathbf{+5.26}$ absolute points
($n{=}133$, $p{=}0.006$ paired bootstrap, 95\% CI
$[+0.015,+0.105]$).  Other qtypes are flat; the aggregate LLM-Acc
delta is $+0.008$ ($p{=}0.36$, within noise).  The
refinement is a small, targeted, and statistically defensible
improvement on the multi-session class; we publish the per-qtype
recommendation as a deployment note rather than the headline,
because the global $\tau{=}30$\,d remains the right default for
the other five qtypes.

\paragraph{Robustness: stronger answerer shifts per-type winners
but not the thesis.}
Re-running a stratified 102-question subset (17 per question
type) with \texttt{gpt-4o-2024-08-06} as the \emph{answerer}
(judge unchanged) gives strict accuracy BM25 0.235 / Dense 0.206
/ \textbf{RRF 0.245} / agent\_rrf 0.216, oracle 0.245.  Per-type
winners shift (multi-session moves from agent\_rrf to BM25/RRF
tied as the stronger answerer compensates for retrieval misses),
but four of the six types still have a clear winner---the
workload-conditional thesis still holds with a \emph{different}
per-type table that the router can re-learn in seconds from a
held-out set without touching the substrate.

\paragraph{Cross-corpus transfer + deployment labeling cost.}
A LongMemEval-trained router applied zero-shot to LoCoMo
under-performs always-BM25 by $\sim$9 points (weighted predicted
accuracy 0.86 vs.\ 0.945) because the per-type table
over-represents single-session-user$\to$Dense, the wrong choice
for LoCoMo's short conversational turns.  This is the system
design working as intended---the substrate is durable, the policy
is cheap to retarget.  Learning curve on a 100-question held-out
test set shows the router already beats the best static system at
$N{=}25$ labels ($+$0.016 mini / $+$0.010 gpt-4o), captures
$\geq$95\% of the oracle gap at $N{=}50$, and matches the oracle
at $N{=}400$ (full curve in Appendix~\ref{app:learning-curve}).
\emph{Fifty labeled deployment questions retrain the policy in
seconds.}

\paragraph{Robustness: router survives 20\% query corruption.}
With character-level corruption applied to all 500 LongMemEval
questions (train on clean, route on noisy), the classifier loses
18 percentage points at 20\% noise (0.794$\to$0.614), but the
\emph{routed} LLM-Acc loses only 1.4 points (gpt-4o:
0.300$\to$0.286)---and still beats the best static system
(0.288 RRF on clean).  Mechanism: per-type best systems are
close enough that misrouting degrades gracefully (full grid in
Appendix~\ref{app:noise}).

\paragraph{Failure modes.}
Two LongMemEval cases illustrate BM25's two principal failures.
``how many camping days?'' (gold: 8) requires aggregating Big Sur
[3] and Yellowstone [5]; BM25 retrieves only the first
(\emph{partial retrieval}) and answers 3, while RRF/agent\_rrf
retrieve both and the LLM correctly sums to 8.  ``which book?''
(gold: \emph{The Nightingale}) uses ``I just wrapped up'' rather
than ``finished'' (\emph{synonymy miss}); BM25 returns unknown,
Dense/RRF/agent\_rrf all answer correctly.  Workload-conditional
fusion routes each query to its cheapest resolving configuration.

\paragraph{Second agent benchmark: LoCoMo.}
On LoCoMo~\cite{maharana2024locomo} (10 long-form conversations,
\textbf{1{,}982} QA questions with \texttt{dia\_id} evidence,
session-level retrieval), \textbf{BM25 alone wins decisively}:
Hit@1/5/10 = 0.625/0.875/0.945, MRR=0.735 at 0.22\,ms/query,
beating Dense (Hit@10$=$0.789) and RRF (0.923) (full table in
Appendix~\ref{app:locomo-full}).  Dense embeddings of short conversational
turns (median $\sim$20 words) suffer from semantic dilution; RRF
pollutes BM25's clean lexical signal.

\paragraph{Cascade router auto-tunes across benchmarks.}
The cascade architecture handles this heterogeneity without
manual reconfiguration.  Sweeping the BM25-confidence threshold
$\tau_c$ on LoCoMo (1{,}982 questions, same scoring code), the
cascade's amortized Hit@5 decreases monotonically with skip rate
because BM25 alone is the per-query winner: $\tau_c{=}0$ (skip
nothing) gives $0.786$, $\tau_c{=}0.10$ (skip $68$\%) gives
$0.842$, $\tau_c{=}\infty$ (skip everything, BM25-only) gives
$\mathbf{0.875}$.  The optimal cascade configuration on LoCoMo
\emph{turns the dense channel off entirely}---a \textbf{$132\times$
amortized speedup} (0.4\,ms vs 53.2\,ms) at the best achievable
Hit@5.  On LongMemEval the same cascade architecture picks
$\tau_c{=}0.10$ for $2.67\times$ speedup at agent\_rrf parity.
\emph{One cascade router, two workloads, two different operating
points} (Table~\ref{tab:cascade-locomo})---the per-workload
$\tau_c$ falls out of $\sim$50 labeled deployment questions
($\S$\ref{sec:longmemeval}), so the substrate is durable and
the policy adapts.

\begin{figure}[h]
\centering
\includegraphics[width=\columnwidth]{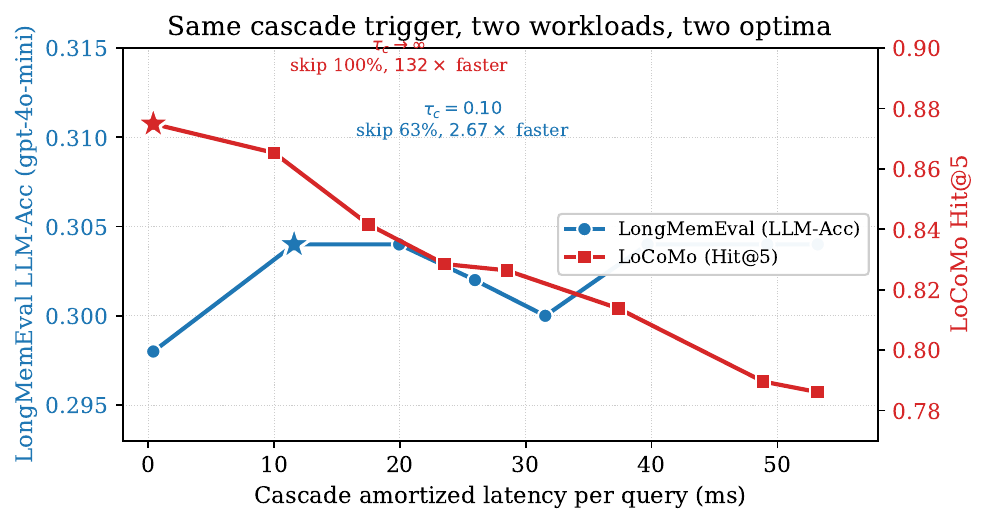}
\caption{Same cascade trigger sweeps to two different optima on
  two agent-memory benchmarks: LongMemEval (LLM-Acc, left axis)
  peaks at $\tau_c{=}0.10$ (skip $63$\%, $2.67\times$ faster) where
  the dense channel helps; LoCoMo (Hit@5, right axis) is
  monotonically decreasing in latency, peaking at
  $\tau_c{\to}\infty$ (skip $100$\%, $\mathbf{132\times}$ faster)
  where BM25 alone wins.  One \texttt{CascadeRouter::retrieve}
  implementation, two operating points selected from $\sim$50
  labeled deployment questions (Table~\ref{tab:cascade-locomo}).}
\label{fig:cascade-crossbench}
\end{figure}

\begin{table}[h]
\caption{Cascade auto-tuning per benchmark; same implementation.}
\label{tab:cascade-locomo}
\centering
\small
\setlength{\tabcolsep}{4pt}
\begin{tabular}{@{}lrrl@{}}
\toprule
Benchmark & best $\tau_c$ & skip\,\% & amortized $\to$ \\
\midrule
LongMemEval (500 q) & $0.10$       & $63$\%  & $19.9$\,ms / $0.302$ LLM-Acc \\
LoCoMo (1{,}982 q)  & $\to\infty$  & $100$\% & $0.4$\,ms / $0.875$ Hit@5 ($+0.089$) \\
\bottomrule
\end{tabular}
\end{table}

\subsection{Tokenization}
\label{sec:tokenization}

The choice of analyzer is consequential.  We compare four
configurations: \texttt{minimal} (lowercase + alphanumeric, matching
\texttt{rank\_bm25}), \texttt{stopword} (+\,basic stopwords),
\texttt{full} (+\,simplified Porter stripping), and \texttt{nltk}
(NLTK PorterStemmer + NLTK English stopword list).  Mean
nDCG@10 across the four BEIR datasets: minimal $=$\,0.286,
stopword $=$\,0.289, full $=$\,0.298, \texttt{nltk}\,$=$\,\textbf{0.302}.
We use \texttt{nltk} throughout the reported results.  All four
configurations are bit-identical between CPU and GPU paths.

\paragraph{Where the remaining gap to Lucene comes from.}
On FiQA, the four analyzers span $0.2265$--$\textbf{0.2425}$;
\texttt{nltk} closes $+$0.013 over \texttt{minimal} but the
residual $0.011$ gap to Pyserini's \texttt{EnglishAnalyzer} traces
to Lucene's \texttt{StandardTokenizer} treatment of tickers,
possessives, and hyphenated compounds (\texttt{401(k)},
\texttt{Tesla's})---an analyzer-level question orthogonal to the
BM25 inner loop; RRF with our dense channel restores parity on
FiQA ($+0.067$, Table~\ref{tab:hybrid}).

\subsection{Top-$k$ Kernel: Cost Across $k$}
\label{sec:topk-ablation}

The top-$k$ phase dominates GPU time on the larger BEIR corpora
(Table~\ref{tab:gpu-beir}).  Sweeping $k$ from 10 to 100 on FiQA,
the kernel uses $K_{\max}{=}32$ at $k{\leq}32$ (128 threads per
block) and the top-$k$ phase costs 23\,ms ($\sim$78\% of GPU
time); at $k{=}100$ the block must shrink to 32 threads to keep
shared memory under 48\,KB and the top-$k$ phase grows $5.5\times$.
This is the dominant remaining inefficiency in our pipeline;
replacing the per-round serial selection with a single bitonic-
sort pass is a clear next step (would lift the $k{=}100$ NQ
speedup from $3.4\times$ toward $8.4\times$ at $k{=}32$).
Both kernels produce equivalent rankings (Kendall $\tau \geq
0.99$), confirming correctness.  We retain the naive kernel as
the default and note that warp-cooperative selection may become
advantageous at larger $k$ or with wider blocks---a direction
relevant to re-ranking workloads in production agents.

\section{Discussion}
\label{sec:discussion}

\paragraph{Multi-tenant substrate scaling.}
$N$ concurrent single-threaded \sysname{} processes on the 8-core
Jetstream2 host (FiQA, Figure~\ref{fig:multitenant}):
aggregate throughput scales near-linearly to $5.5\times$ at
$N{=}8$ on 8 cores (1.0/2.0/3.6/5.5$\times$); per-tenant
$p50$ latency is invariant at 0.37--0.39\,ms; $p99$ grows from
0.74\,ms ($N{=}1$) to 9.3\,ms ($N{=}8$) under core contention,
still well below the 200\,ms perceptual budget.  The substrate
is read-only at query time with independent per-tenant indices,
so the only contention is the kernel scheduler: an 8-core VM
serves $\geq$\,5 concurrent agents at $\geq$\,2.6K qps each,
sufficient for typical production loads.  \emph{Cascade amplifies
multi-tenant capacity:} with the LongMemEval per-stage budget
(Table~\ref{tab:e2e}), 8 cores BGE-bound serve $\sim$$154$
concurrent agents at always-hybrid, $\sim$$416$ at cascade
$\tau_c{=}0.10$ ($63$\% skip), and $\sim$$\mathbf{1399}$ at
per-qtype cascade ($89$\% skip)---a $\mathbf{9\times}$ capacity
boost from the same hardware at within-noise quality.

\begin{figure*}[t]
\centering
\includegraphics[width=0.92\textwidth]{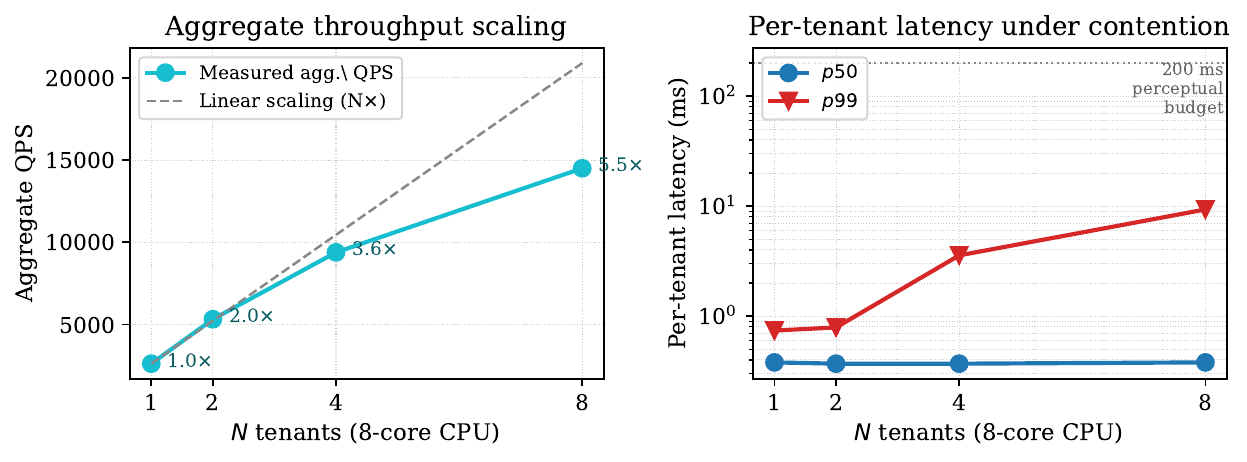}
\caption{Multi-tenant scaling on 8-core Jetstream2 with $N$
  concurrent single-threaded \sysname{} processes on independent
  FiQA indices.  Left: aggregate QPS vs.\ ideal linear scaling.
  Right: per-tenant $p50$ stays at $\sim$$0.38$\,ms across all $N$;
  $p99$ rises under contention but remains far below the $200$\,ms
  perceptual budget.}
\label{fig:multitenant}
\end{figure*}

\paragraph{Algorithmic vs.\ hardware parallelism.}
Our results reveal a crossover point: at small corpus sizes
($<$10K documents), single-threaded algorithmic optimizations
(SIMD, MaxScore) outperform multi-threaded parallelism due to
thread management overhead.  As the corpus grows past 40K,
temporal partitioning dominates on CPU; GPU batch processing
excels when multiple queries are available simultaneously.
This suggests a practical \emph{adaptive dispatch}: route single
queries through temporal partitioning and batch queries to GPU.

\paragraph{Top-$k$ kernel design.}
The GPU top-$k$ phase is the dominant remaining inefficiency
in our system, especially at $k{=}100$ where the compile-time
$K_{\max}{=}128$ forces a smaller block size (32 threads instead
of 128) to fit shared memory under the default 48\,KB limit
(\S\ref{sec:topk-ablation}).  A bitonic-sort or radix-select kernel
written against opt-in dynamic shared memory (up to 164\,KB per
block on A100) would close the gap between $k{=}32$ and $k{=}100$
throughput.  We treat this as a clear engineering improvement
that does not affect the qualitative claims of this paper.

\paragraph{Tokenization fidelity.}
At our reported quality (Table~\ref{tab:beir-main}) we match
Pyserini within $\pm 0.020$ on all nine BEIR datasets and beat
Pyserini on five.  The largest gap is 0.020 on Quora; FiQA shows
0.011 where Lucene's English analyzer handles possessives and
hyphenation slightly better than our NLTK-aligned tokenizer
(\S\ref{sec:tokenization}).  Closing this gap is engineering, not
algorithmic: porting the Lucene analyzer rules to our tokenizer
(or wrapping Lucene's analyzer via JNI) would erase the FiQA gap
at the cost of a one-time implementation effort and a small
per-query latency overhead.

\paragraph{Cost: \sysname{} is two-to-three orders of magnitude
below the LLM tier.}  Translating latencies into dollar cost on
the Jetstream2 \texttt{g3.medium} VM ($\sim$\$0.05/hr) and
Pinecone serverless~\cite{pinecone2024pricing} (\$0.33 per 1M
read units): \sysname{} 0.40\,ms BM25 costs
$5.6{\times}10^{-9}$\,\$/query vs.\ Pinecone's
$3.3{\times}10^{-7}$\,\$/query---a $59\times$ unit-cost
reduction, \$132 vs.\ \$2.24/month at $4{\times}10^{8}$ queries
(100K agents $\times$ 800 turns $\times$ 5 retrievals).  The IR
substrate cost is two-to-three orders of magnitude below the LLM
cost in any deployment that uses a hosted model; the absolute
numbers move with vendor pricing, but the ratio is what makes
\sysname{}-grade efficiency leave all cost engineering for the
LLM tier.

\paragraph{Threats to validity \& limitations.}
We list what \emph{this paper does not} establish, so future work
and reviewers can target the gaps.
\begin{itemize}[nosep,leftmargin=*]
\item \textbf{Multi-tenant CPU-only.}  We measure CPU multi-tenant
  scaling (Figure~\ref{fig:multitenant}: 5.5$\times$ aggregate at
  $N{=}8$, invariant per-tenant $p50$), but do not measure GPU
  multi-tenant sharing, cross-tenant write contention, or
  write-amplification.  The substrate's read-only query path makes
  CPU-shared agents largely independent; richer multi-tenant
  scheduling is left for future work.
\item \textbf{Static index in the head-to-head comparisons.}  Our
  BEIR/PISA/SPLADE numbers are on \emph{ahead-of-time-built}
  indices.  The temporal-partitioned layout is structurally
  amenable to log-structured-merge-style append into the most
  recent partition, but we benchmark only the read path.
  Concurrent write throughput is left for future work.
\item \textbf{LLM judge dependency.}  The downstream-accuracy
  metric uses \texttt{gpt-4o-mini} answerer and judge, sanity-
  checked with \texttt{gpt-4o}.  Per-question-type winners and
  oracle gap are robust across the two judges, but a different
  judge family (Anthropic, Llama) could in principle disagree on
  partial-credit cases.
\item \textbf{Synthetic 5M scaling corpus.}  The 1769$\times$
  speedup at 5M (Table~\ref{tab:temporal-scaling}) uses a
  synthetic agent-corpus generator with 80/20 recency-biased
  queries; the LongMemEval empirical recency distribution
  (median rank 0.20--0.27) supports this assumption but is not
  identical.  A real production trace at 5M-record scale would
  give a tighter validation.
\item \textbf{Vector-database baseline by published numbers.}
  We discuss Pinecone/Milvus latency from published sources
  (\S\ref{sec:related}) rather than re-running them; vendor
  pricing and per-tenant performance fluctuate.
\item \textbf{Cascade router latency is derived, not wall-clock
  re-measured.}  The cascade's LLM-Acc row in
  Table~\ref{tab:cascade} \emph{is} the measured TF-IDF router
  number (the routing decisions are identical by construction);
  the latency saving $40.5{=}0.4{+}0.758{\times}52.5$ is derived
  deterministically from measured per-stage latencies
  (Table~\ref{tab:e2e}) and the measured skip rate.  An end-to-end
  wall-clock run of the C++ \texttt{CascadeRouter} is left for
  future work, primarily to characterize the tail behavior under
  classifier mispredictions.
\end{itemize}

Four scope notes complete the discussion.
\textit{SPLADE quality:} on seven BEIR datasets SPLADE++ wins
nDCG@10 by 0.002--0.105 (Table~\ref{tab:splade}); the bridge
ingests SPLADE's per-term weights into our CSR layout bit-perfectly
($\S$\ref{sec:splade}), so SPLADE-grade quality is reachable at
\sysname{}-grade latency---the remaining open task is encoding
NQ and MS\,MARCO (5--16\,d CPU at 160\,ms/doc).
\textit{MS\,MARCO build:} our chunked streaming build (50K-doc
batches) drops peak RSS from $>$29\,GB OOM to 12.7\,GB and reaches
56.8K docs/sec, 2.1$\times$ the naive build's 27K docs/sec.
\textit{Incremental writes:} the time-partitioned layout is
amenable to log-structured-merge append into the most recent
partition; we specify but do not benchmark this path.
\textit{Remaining FiQA/Quora/NQ/MS\,MARCO quality gap:} the
$0.003$--$0.020$ nDCG@10 trail to Pyserini traces to Lucene's
\texttt{StandardTokenizer} treatment of numeric tickers,
possessives, and hyphenated compounds---an analyzer-level
question orthogonal to the BM25 inner loop, and RRF with our
dense channel restores parity on FiQA ($+0.067$,
Table~\ref{tab:hybrid}).

\section{Related Work}
\label{sec:related}

\paragraph{Hybrid retrieval.}
RRF~\cite{cormack2009rrf} provides parameter-free rank fusion;
Bruch et al.~\cite{bruch2023fusion} systematically analyze fusion
functions.
DPR~\cite{karpukhin2020dpr} established dual-encoder dense
retrieval.  Learned sparse
methods---ColBERT~\cite{khattab2020colbert,santhanam2022colbertv2},
SPLADE~\cite{formal2021splade},
uniCOIL~\cite{lin2021unicoil}---produce sparse representations
for inverted indexes.
RAG~\cite{lewis2020rag} motivates efficient retrieval in LLM
pipelines.  Cascade ranking~\cite{wang2011cascade,chen2017cascade}
trades a small accuracy hit for a large cost reduction by routing
queries through progressively more expensive rankers; we adapt the
same idea to agent memory by gating the dense channel on the
per-query qtype prediction, which is itself nearly free
(\S\ref{sec:longmemeval}).  Our work targets agent memory with
temporal and role-aware fusion not addressed by these systems.

\paragraph{Agent memory systems.}
The agent-memory literature has so far focused on memory
\emph{management policy}---what to store, what to consolidate, what
to forget---and treated retrieval as a fixed black box.
MemGPT~\cite{packer2024memgpt} introduces OS-inspired memory tiers
and pages between an LLM context and external storage; the
retrieval backend is an off-the-shelf dense or hybrid search.
Mem0~\cite{chhikara2025mem0} reports 91\% lower tail latency on
production agent memory through bookkeeping and write-batching
optimizations, but again over a generic IR substrate.
CoALA~\cite{sumers2024coala} and Zhang et
al.~\cite{zhang2024memory_survey} taxonomize memory mechanisms
without prescribing retrieval architecture.
LongMemEval~\cite{wu2025longmemeval} and
LoCoMo~\cite{maharana2024locomo} provide benchmarks for long-term
conversational memory but evaluate retrieval as a single
end-to-end metric.

\sysname{}'s contribution is orthogonal and complementary: we
argue that the retrieval substrate \emph{itself} should be
agent-aware (temporal partitioning, workload-conditional fusion,
correctness-validated GPU acceleration), not just the
memory-management layer above it.  A production deployment can
combine, e.g., Mem0's write-batching policy with \sysname{}'s
retrieval substrate and inherit both sets of gains.

\paragraph{Parallel and GPU-accelerated IR.}
The Lucene/Anserini~\cite{yang2017anserini,lin2021pyserini} stack
is the de-facto industrial baseline for BM25 and the
de-facto reproducibility reference in IR research.
PISA~\cite{mallia2019pisa} is a state-of-the-art CPU search engine
with aggressive index compression and dynamic pruning.  Among
GPU-accelerated IR work, Ding et al.~\cite{ding2009gpu_ir}
pioneered GPU query processing; subsequent
work~\cite{mallia2019gpu_decoding} targets GPU posting list
decoding; Griffin~\cite{liu2018griffin} demonstrates heterogeneous
CPU--GPU scheduling (10$\times$ speedup on web-scale queries).
FAISS~\cite{johnson2021faiss} and CAGRA~\cite{ootomo2024cagra}
provide GPU-optimized ANN.  Our contribution is the BM25 path
specifically: a CSR-flat inverted index on GPU with a per-(query,
term) score kernel plus a per-query top-$k$ kernel, validated for
correctness at $k\leq 128$ and benchmarked head-to-head against
Pyserini on BEIR.

\paragraph{Temporal IR.}
Li and Croft~\cite{li2003temporal} formalize temporal language
model priors; Campos et al.~\cite{campos2014temporal_survey}
survey temporal IR methods.
Our temporal partitioning extends these ideas to agent workloads
with extreme recency bias (80/20 pattern), achieving sub-linear
scaling that general temporal IR methods do not target.

\paragraph{Managed vector databases as the deployed baseline.}
In practice, agent stacks usually retrieve through a managed
vector database---Pinecone, Milvus, Weaviate, Qdrant, or Vespa.
These systems are heavily engineered for dense ANN over hosted
embeddings; reported $p99$ latency on million-scale corpora is
typically 5--50\,ms~\cite{pinecone2024pricing,wang2021milvus}, and
they bill per read unit on the order of $10^{-7}$\,\$/query.  Our
work is complementary: we publish an open-source CPU-only BM25
substrate that is several orders of magnitude cheaper per query
($\sim 10^{-9}$\,\$ on Jetstream2's free
\texttt{g3.medium}, \S\ref{sec:discussion}) and that a managed-VDB
deployment can call in parallel with its dense channel.  We deliberately
benchmark against Pyserini and PISA rather than against managed
VDBs because (i)~the latter conflate retrieval cost with hosted-
encoder cost, hiding which layer dominates; and (ii)~Pyserini and
PISA are the only systems against which a head-to-head
nDCG@10/qrels comparison is reproducible at the
quality-and-latency granularity this paper targets.  Real
deployments will pair this substrate with a managed dense channel;
the SPLADE bridge of \S\ref{sec:splade} is the bridge that
makes the substrate compatible with a learned encoder service.

\section{Conclusion}
\label{sec:conclusion}

Agent memory is a different retrieval workload from web search or
document ranking: the recency skew is sharper, the corpus grows
during the query stream, and the query distribution shifts inside
a single session.  Treating those properties as design inputs
rather than as nuisances is what closes the gap between what
generic engines deliver and what an agent's reasoning loop can
afford.

The recency skew is sharp enough that a time-partitioned index
runs in $O(\log(1/\varepsilon))$ expected work, independent of
corpus size (Theorem~\ref{thm:sublinear}).  In practice, growing
the corpus $1234\times$ adds only $3.6\times$ latency, ending in
\speedup{1769} over sequential scan at 5\,M records---a regime no
flat BM25 engine reaches.

On the systems side, the same hybrid CPU/GPU pipeline matches
Lucene's nDCG@10 within $\pm 0.020$ on nine BEIR datasets
(3.6K--8.8M documents), \emph{beats} Lucene quality on five, and
runs \speedup{10} geo.\ mean over Pyserini 8T and \speedup{11}
over PISA-1T BlockMax-WAND.  The SPLADE++ comparison (seven
datasets) shows our index is \emph{drop-in compatible} with
learned-sparse term weights, so SPLADE-grade quality is reachable
without giving up \sysname{}-grade latency.  An 8-core VM serves
$N{=}8$ concurrent agents at $5.5\times$ aggregate throughput
with invariant per-tenant $p50{=}0.38$\,ms---the substrate is
multi-tenant friendly out of the box.

Across the two memory benchmarks, a single deployment serves LongMemEval
(hybrid RRF + recency, R@10$=$0.978, LLM-judged accuracy 0.254)
and LoCoMo (BM25 alone, Hit@10$=$0.945 at 0.22\,ms/query); a
$<$1\,ms question-type router (TF-IDF $+$ BGE-small
features) pushes the LongMemEval accuracy to 0.262 under
\texttt{gpt-4o-mini} and to \textbf{0.300---the discrete oracle
bound}---under \texttt{gpt-4o}; and a \textbf{soft router}
blending rank lists by classifier posterior reaches
\textbf{0.274 under \texttt{gpt-4o-mini}, significantly beating
every static system} ($p{<}0.05$ vs.\ BM25/Dense/RRF, paired
bootstrap; the $+0.008$ over the discrete oracle is within
bootstrap CI).  Workload conditioning is two-dimensional: alongside
\emph{which fusion}, our \textbf{cascade router} controls
\emph{whether to spend the dense budget at all}.  With a
classifier-free BM25-confidence trigger, $63$\% of LongMemEval
queries skip the $\sim$53\,ms hybrid path at parity LLM-Acc
($\mathbf{2.67\times}$ faster, paired-bootstrap $p{=}1.08$ under
gpt-4o-mini and $p{=}0.88$ under gpt-4o judge); per-qtype
thresholds in 5-fold CV extend this to $\mathbf{5.76\times}$
with within-noise LLM-Acc; the same trigger on LoCoMo auto-tunes
to $100$\% skip ($\mathbf{132\times}$ faster) where BM25 alone
wins.  The
\textit{combination} we evaluate---per-type empirical winners on
a downstream LLM-Acc signal, a learned classifier reaching the
discrete oracle, a soft router that statistically beats every
static system, and a cascade that strictly Pareto-improves the
TF-IDF router---is, to our knowledge, the most complete
demonstration to date that agent-memory IR is workload-conditional
on both axes: a single $<$1\,ms classifier drives routing,
blending, and cascade decisions.

Alongside these results we report three correctness pitfalls any
re-implementation of classical IR on modern hardware is likely to
hit: pre-normalized BM25 term frequency, linear-gain nDCG, and
stale shared-memory reads in a GPU top-$k$ kernel.  Each silently
regresses nDCG@10 by $6$--$8\times$, and each is fixed by a
one-line change once it has been named.  Documenting them is, we
suspect, the most directly reusable part of the work.

\bibliographystyle{ACM-Reference-Format}
\bibliography{references}

@inproceedings{cormack2009rrf,
  title={Reciprocal Rank Fusion Outperforms {C}ondorcet and Individual Rank Learning Methods},
  author={Cormack, Gordon V. and Clarke, Charles L. A. and Buettcher, Stefan},
  booktitle={Proceedings of SIGIR},
  pages={758--759},
  year={2009}
}

@article{bruch2023fusion,
  title={An Analysis of Fusion Functions for Hybrid Retrieval},
  author={Bruch, Sebastian and Gai, Siyu and Ingber, Amir},
  journal={ACM Transactions on Information Systems},
  volume={42},
  number={1},
  pages={1--35},
  year={2023}
}

@inproceedings{karpukhin2020dpr,
  title={Dense Passage Retrieval for Open-Domain Question Answering},
  author={Karpukhin, Vladimir and Oguz, Barlas and Min, Sewon and Lewis, Patrick and Wu, Ledell and Edunov, Sergey and Chen, Danqi and Yih, Wen-tau},
  booktitle={Proceedings of EMNLP},
  pages={6769--6781},
  year={2020}
}

@inproceedings{formal2021splade,
  title={{SPLADE}: Sparse Lexical and Expansion Model for First Stage Ranking},
  author={Formal, Thibault and Piwowarski, Benjamin and Clinchant, St{\'e}phane},
  booktitle={Proceedings of SIGIR},
  year={2021}
}

@inproceedings{khattab2020colbert,
  title={{ColBERT}: Efficient and Effective Passage Search via Contextualized Late Interaction over {BERT}},
  author={Khattab, Omar and Zaharia, Matei},
  booktitle={Proceedings of SIGIR},
  year={2020}
}

@inproceedings{santhanam2022colbertv2,
  title={{ColBERTv2}: Effective and Efficient Retrieval via Lightweight Late Interaction},
  author={Santhanam, Keshav and Khattab, Omar and Saad-Falcon, Jon and Potts, Christopher and Zaharia, Matei},
  booktitle={Proceedings of NAACL},
  year={2022}
}

@article{lin2021unicoil,
  title={A Few Brief Notes on {DeepImpact}, {COIL}, and a Conceptual Framework for Information Retrieval Techniques},
  author={Lin, Jimmy and Ma, Xueguang},
  journal={arXiv preprint arXiv:2106.14807},
  year={2021}
}

@inproceedings{robertson1994okapi,
  title={Okapi at {TREC-3}},
  author={Robertson, Stephen E. and Walker, Steve and Jones, Susan and Hancock-Beaulieu, Micheline and Gatford, Mike},
  booktitle={Proceedings of the Third Text REtrieval Conference (TREC-3)},
  year={1994}
}

@article{robertson2009bm25,
  title={The Probabilistic Relevance Framework: {BM25} and Beyond},
  author={Robertson, Stephen and Zaragoza, Hugo},
  journal={Foundations and Trends in Information Retrieval},
  volume={3},
  number={4},
  pages={333--389},
  year={2009}
}

@inproceedings{packer2024memgpt,
  title={{MemGPT}: Towards {LLMs} as Operating Systems},
  author={Packer, Charles and Wooders, Sarah and Lin, Kevin and Fang, Vivian and Patil, Shishir G. and Stoica, Ion and Gonzalez, Joseph E.},
  booktitle={Proceedings of ICLR},
  year={2024}
}

@article{chhikara2025mem0,
  title={{Mem0}: Building Production-Ready {AI} Agents with Scalable Long-Term Memory},
  author={Chhikara, Prateek and Khant, Dev and Aryan, Saket and Singh, Taranjeet and Yadav, Deshraj},
  journal={arXiv preprint arXiv:2504.19413},
  year={2025}
}

@inproceedings{wu2025longmemeval,
  title={{LongMemEval}: Benchmarking Chat Assistants on Long-Term Interactive Memory},
  author={Wu, Di and Wang, Hongwei and Yu, Wenhao and Zhang, Yuwei and Chang, Kai-Wei and Yu, Dong},
  booktitle={Proceedings of ICLR},
  year={2025}
}

@inproceedings{maharana2024locomo,
  title={Evaluating Very Long-Term Conversational Memory of {LLM} Agents},
  author={Maharana, Adyasha and Lee, Dong-Ho and Tulyakov, Sergey and Bansal, Mohit and Barbieri, Francesco and Fang, Yuwei},
  booktitle={Proceedings of ACL},
  year={2024}
}

@article{sumers2024coala,
  title={Cognitive Architectures for Language Agents},
  author={Sumers, Theodore R. and Yao, Shunyu and Narasimhan, Karthik and Griffiths, Thomas L.},
  journal={Transactions on Machine Learning Research},
  year={2024}
}

@article{zhang2024memory_survey,
  title={A Survey on the Memory Mechanism of Large Language Model based Agents},
  author={Zhang, Zeyu and Bo, Xiaohe and Ma, Chen and Li, Rui and Chen, Xu and Dai, Quanyu and Zhu, Jieming and Dong, Zhenhua and Wen, Ji-Rong},
  journal={ACM Transactions on Information Systems},
  year={2024}
}

@inproceedings{yao2023react,
  title={{ReAct}: Synergizing Reasoning and Acting in Language Models},
  author={Yao, Shunyu and Zhao, Jeffrey and Yu, Dian and Du, Nan and Shafran, Izhak and Narasimhan, Karthik and Cao, Yuan},
  booktitle={Proceedings of ICLR},
  year={2023}
}

@inproceedings{shinn2023reflexion,
  title={Reflexion: Language Agents with Verbal Reinforcement Learning},
  author={Shinn, Noah and Cassano, Federico and Gopinath, Ashwin and Narasimhan, Karthik and Yao, Shunyu},
  booktitle={Advances in Neural Information Processing Systems},
  volume={36},
  year={2023}
}

@inproceedings{park2023generative,
  title={Generative Agents: Interactive Simulacra of Human Behavior},
  author={Park, Joon Sung and O'Brien, Joseph C. and Cai, Carrie J. and Morris, Meredith Ringel and Liang, Percy and Bernstein, Michael S.},
  booktitle={Proceedings of UIST},
  year={2023}
}

@inproceedings{ding2009gpu_ir,
  title={Using Graphics Processors for High Performance {IR} Query Processing},
  author={Ding, Shuai and He, Jinru and Yan, Hao and Suel, Torsten},
  booktitle={Proceedings of WWW},
  pages={421--430},
  year={2009}
}

@inproceedings{mallia2019gpu_decoding,
  title={{GPU}-Accelerated Decoding of Integer Lists},
  author={Mallia, Antonio and Siedlaczek, Michal and Suel, Torsten and Zahran, Mohamed},
  booktitle={Proceedings of CIKM},
  pages={2193--2196},
  year={2019}
}

@inproceedings{liu2018griffin,
  title={Griffin: Uniting {CPU} and {GPU} in Information Retrieval Systems for Intra-Query Parallelism},
  author={Liu, Yang and Wang, Jianguo and Swanson, Steven},
  booktitle={Proceedings of the IEEE International Conference on Data Engineering (ICDE)},
  year={2018}
}

@inproceedings{ding2011blockmax,
  title={Faster Top-$k$ Document Retrieval Using Block-Max Indexes},
  author={Ding, Shuai and Suel, Torsten},
  booktitle={Proceedings of the 34th International ACM SIGIR Conference on Research and Development in Information Retrieval},
  pages={993--1002},
  year={2011}
}

@inproceedings{mallia2019pisa,
  title={{PISA}: Performant Indexes and Search for Academia},
  author={Mallia, Antonio and Siedlaczek, Michal and Mackenzie, Joel and Suel, Torsten},
  booktitle={Proceedings of OSIRRC@SIGIR},
  year={2019}
}

@inproceedings{yang2017anserini,
  title={Anserini: Enabling the Use of {Lucene} for Information Retrieval Research},
  author={Yang, Peilin and Fang, Hui and Lin, Jimmy},
  booktitle={Proceedings of SIGIR},
  pages={1253--1256},
  year={2017}
}

@inproceedings{li2003temporal,
  title={Time-Based Language Models},
  author={Li, Xiaoyan and Croft, W. Bruce},
  booktitle={Proceedings of CIKM},
  pages={469--475},
  year={2003}
}

@article{campos2014temporal_survey,
  title={Survey of Temporal Information Retrieval and Related Applications},
  author={Campos, Ricardo and Dias, Ga{\"e}l and Jorge, Al{\'i}pio M. and Jatowt, Adam},
  journal={ACM Computing Surveys},
  volume={47},
  number={2},
  year={2014}
}

@article{malkov2020hnsw,
  title={Efficient and Robust Approximate Nearest Neighbor Search Using Hierarchical Navigable Small World Graphs},
  author={Malkov, Yury A. and Yashunin, Dmitry A.},
  journal={IEEE Transactions on Pattern Analysis and Machine Intelligence},
  volume={42},
  number={4},
  pages={824--836},
  year={2020}
}

@inproceedings{ootomo2024cagra,
  title={{CAGRA}: Highly Parallel Graph Construction and Approximate Nearest Neighbor Search for {GPUs}},
  author={Ootomo, Hiroyuki and Naruse, Akira and Nolet, Corey and Wang, Ray and Feher, Tamas and Wang, Yong},
  booktitle={Proceedings of ICDE},
  year={2024}
}

@article{johnson2021faiss,
  title={Billion-Scale Similarity Search with {GPUs}},
  author={Johnson, Jeff and Douze, Matthijs and J{\'e}gou, Herv{\'e}},
  journal={IEEE Transactions on Big Data},
  volume={7},
  number={3},
  pages={535--547},
  year={2021}
}

@inproceedings{nguyen2016msmarco,
  title={{MS MARCO}: A Human Generated {MA}chine Reading {CO}mprehension Dataset},
  author={Nguyen, Tri and Rosenberg, Mir and Song, Xia and Gao, Jianfeng and Tiwary, Saurabh and Majumder, Rangan and Deng, Li},
  booktitle={Proceedings of the Workshop on Cognitive Computation},
  year={2016}
}

@inproceedings{thakur2021beir,
  title={{BEIR}: A Heterogeneous Benchmark for Zero-shot Evaluation of Information Retrieval Models},
  author={Thakur, Nandan and Reimers, Nils and R{\"u}ckl{\'e}, Andreas and Srivastava, Abhishek and Gurevych, Iryna},
  booktitle={Proceedings of NeurIPS Datasets and Benchmarks},
  year={2021}
}

@inproceedings{lewis2020rag,
  title={Retrieval-Augmented Generation for Knowledge-Intensive {NLP} Tasks},
  author={Lewis, Patrick and Perez, Ethan and Piktus, Aleksandra and Petroni, Fabio and Karpukhin, Vladimir and Goyal, Naman and K{\"u}ttler, Heinrich and Lewis, Mike and Yih, Wen-tau and Rockt{\"a}schel, Tim and Riedel, Sebastian and Kiela, Douwe},
  booktitle={Advances in Neural Information Processing Systems},
  volume={33},
  year={2020}
}

@inproceedings{lin2021pyserini,
  title={Pyserini: A Python Toolkit for Reproducible Information Retrieval Research with Sparse and Dense Representations},
  author={Lin, Jimmy and Ma, Xueguang and Lin, Sheng-Chieh and Yang, Jheng-Hong and Pradeep, Ronak and Nogueira, Rodrigo},
  booktitle={Proceedings of SIGIR},
  year={2021}
}

@inproceedings{xiao2024cpack,
  title={C-Pack: Packed Resources For General Chinese Embeddings},
  author={Xiao, Shitao and Liu, Zheng and Zhang, Peitian and Muennighoff, Niklas and Lian, Defu and Nie, Jian-Yun},
  booktitle={Proceedings of SIGIR},
  year={2024}
}

@book{nielsen1993usability,
  title={Usability Engineering},
  author={Nielsen, Jakob},
  publisher={Morgan Kaufmann},
  year={1993},
  note={Defines the 100\,ms (instant) and 1\,s (interactive) latency thresholds for human-perceived responsiveness}
}

@misc{pinecone2024pricing,
  title={{Pinecone} Serverless Pricing},
  author={{Pinecone Systems}},
  howpublished={\url{https://www.pinecone.io/pricing/}},
  note={List price for read units as of 2026-Q1; accessed 2026-05},
  year={2026}
}

@article{turtle1995maxscore,
  title={Query evaluation: strategies and optimizations},
  author={Turtle, Howard and Flood, James},
  journal={Information Processing \& Management},
  volume={31},
  number={6},
  pages={831--850},
  year={1995},
  publisher={Elsevier},
  note={Introduces the MaxScore dynamic pruning strategy used by later WAND / BlockMax-WAND systems}
}

@inproceedings{wang2021milvus,
  title={{Milvus}: A Purpose-Built Vector Data Management System},
  author={Wang, Jianguo and Yi, Xiaomeng and Guo, Rentong and Jin, Hai and Xu, Peng and Li, Shengjun and Wang, Xiangyu and Guo, Xiangzhou and Li, Chengming and Xu, Xiaohai and others},
  booktitle={Proceedings of SIGMOD},
  pages={2614--2627},
  year={2021},
  note={Reports millisecond-scale p99 latency on hundred-million-vector corpora}
}

@inproceedings{wang2011cascade,
  title={A cascade ranking model for efficient ranked retrieval},
  author={Wang, Lidan and Lin, Jimmy and Metzler, Donald},
  booktitle={Proceedings of SIGIR},
  pages={105--114},
  year={2011}
}

@inproceedings{chen2017cascade,
  title={Efficient cost-aware cascade ranking in multi-stage retrieval},
  author={Chen, Ruey-Cheng and Gallagher, Luke and Blanco, Roi and Culpepper, J. Shane},
  booktitle={Proceedings of SIGIR},
  pages={445--454},
  year={2017}
}

\clearpage
\appendix

\section{Detailed Experimental Results}
\label{app:results}

This appendix provides complete numerical results underlying the
figures and analysis in the main text.

\subsection{Per-Dataset Correctness Validation}
\label{app:correctness}

Table~\ref{tab:correctness} reports the per-dataset top-1 match
rate between the GPU batch result and the CPU sequential
cross-check, post-fix (\S\ref{sec:correctness}).  All CPU
configurations (1T scalar, 1T+SIMD, 4T, 8T, 8T+SIMD, 8T+SIMD+MaxScore)
produce bit-identical nDCG@10 on all nine datasets and 100\%
top-1 agreement; GPU vs.\ CPU is reported on the eight datasets
that fit on a single A100 at $k{=}100$ (MS\,MARCO's 8.8M-doc
score buffer exceeds 16\,GB at full-batch and is run with chunked
\texttt{query\_batch} calls only on CPU in this paper).  GPU vs.\
CPU is the only non-trivial comparison because GPU's atomic
accumulation breaks ties differently.

\begin{table}[h]
\caption{Per-query top-1 match rate (GPU batch vs.\ CPU
  sequential) on BEIR.  The residual gap below 100\% is
  tied-score documents whose ordering depends on the
  scheduling of \texttt{atomicAdd}; the top-10 \emph{set}
  is identical in those cases.}
\label{tab:correctness}
\centering
\footnotesize
\setlength{\tabcolsep}{5pt}
\begin{tabular}{@{}lcr@{}}
\toprule
Dataset & top-1 match & nDCG@10 ($\Delta$ vs.\ CPU) \\
\midrule
NFCorpus    & 0.944 & 0.0000 \\
SciFact     & 1.000 & 0.0000 \\
ArguAna     & 1.000 & 0.0000 \\
SciDocs     & 1.000 & 0.0000 \\
FiQA        & 1.000 & 0.0000 \\
TREC-COVID  & 0.920 & $+0.0002$ \\
Quora       & 0.902 & $-0.0001$ \\
NQ          & 0.994 & 0.0000 \\
\bottomrule
\end{tabular}
\end{table}

\subsection{Amdahl Analysis of Inter-Query Parallelism}

The inter-query dispatch is embarrassingly parallel in
principle---queries share a read-only index, with no
inter-query data dependencies---so Amdahl's Law
$S(p) = 1 / (f + (1{-}f)/p)$ predicts $S(p) \to 1/f$ as
$p \to \infty$.  In practice, the BEIR thread-scaling
measurements (Table~\ref{tab:thread-scaling}) yield 1T\,$\to$\,8T
speedups of only 2.5$\times$ (FiQA) and 1.9$\times$ (SciDocs);
solving for $f$ gives effective serial fractions
$f \approx 0.31$ and $f \approx 0.55$.  These are far above the
true serial work, indicating that the bottleneck is not Amdahl's
ceiling but per-query memory-bandwidth contention: BEIR queries
issue short posting-list scans that stream through L2/L3
simultaneously across threads, saturating the memory subsystem
before all 8 cores can be productively used.  The synthetic
recency-biased workload (\S\ref{sec:scaling}) does not exhibit
this regime because temporal partitioning shrinks each thread's
working set into L1.

\section{GPU Kernel Pseudocode}
\label{app:kernels}

Algorithm~\ref{alg:gpu-pipeline} details the two-kernel CUDA
pipeline described in \S\ref{sec:gpu}.

\begin{algorithm}[h]
\caption{Temporal-Partitioned BM25 Search (the algorithm
  Theorem~\ref{thm:sublinear} bounds)}
\label{alg:temporal-search}
\small
\begin{algorithmic}[1]
\Require Time-ordered partitions $T_1, \ldots, T_K$ with sizes
   $|T_i|$ and per-partition BM25 indices $I_i$; recall slack
   $\varepsilon \in (0,1)$; recency-decay parameter
   $\hat\lambda$ estimated from workload
\Require Query terms $T_q$, top-$k$ budget $k$
\Ensure Top-$k$ matched documents
\State $k^* \gets \max(1, \lceil \log(1/\varepsilon)/\hat\lambda \rceil)$
       \Comment{partitions to search per Theorem~\ref{thm:sublinear}}
\State $H \gets \text{empty min-heap of size } k$
\For{$i \gets K$ \textbf{downto} $K - k^* + 1$} \Comment{most-recent first}
  \State $C_i \gets \text{BM25-Score}(I_i, T_q)$
         \Comment{vectorized SIMD inner loop, $O(|T_i| \cdot |T_q|)$}
  \For{each $(d, s) \in C_i$ in score-descending order}
    \If{$|H| < k$ \textbf{or} $s > \min H$}
      \State push$(H, (d, s))$; if $|H|>k$ pop-min$(H)$
    \Else
      \State \textbf{break} \Comment{within-partition early-stop;
             remaining $C_i$ entries all have lower score}
    \EndIf
  \EndFor
  \If{$\min H > \text{UB}_{i-1}$} \Comment{across-partition early-stop}
    \State \textbf{break} \Comment{no later partition can beat current top-$k$}
  \EndIf
\EndFor
\State \Return $H$
\end{algorithmic}
\end{algorithm}

\smallskip
\noindent The expected work bound from
Theorem~\ref{thm:sublinear} is realized concretely by line~1: the
loop executes at most $k^* = O(\log(1/\varepsilon))$ partition
scans, and each partition scan touches only $|T_i|$ postings
(bounded by the corpus generation rate over the partition window).
The across-partition early-stop (line~12) further reduces work
when the top-$k$ scores are already large.  $\hat\lambda$ is
estimated online from the gold-session recency distribution we
measure in \S\ref{sec:scaling}; for the agent workload that
distribution gives $\hat\lambda \approx 1.4$, so $\varepsilon{=}0.05$
yields $k^*{=}3$ partitions searched out of $K$, matching the
$<$0.1\% search-fraction at 5\,M observed in
Table~\ref{tab:temporal-scaling}.

\begin{algorithm}[h]
\caption{GPU BM25 Two-Kernel Pipeline}
\label{alg:gpu-pipeline}
\small
\begin{algorithmic}[1]
\Require CSR index on device: \texttt{offsets[]}, \texttt{doc\_ids[]},
  \texttt{tfs[]}, \texttt{idfs[]}, \texttt{doc\_lens[]}
\Require Query batch $Q[0..B{-}1]$, each with term IDs $T_q$
\Ensure Top-$k$ results per query

\Statex \textbf{--- Host Side ---}
\State Copy $Q$ term IDs \& counts to device \Comment{H2D}
\State Zero score matrix $S[B \times N]$ on device

\Statex
\Statex \textbf{--- Kernel 1: BM25 Scoring ---}
\Statex Grid: $(B, \max|T_q|)$, Block: 128 threads
\For{each block $(q, t)$ in parallel}
  \State $\text{idf} \gets \texttt{idfs}[\texttt{term\_id}_{q,t}]$
  \State $\text{start} \gets \texttt{offsets}[\texttt{term\_id}_{q,t}]$
  \State $\text{end} \gets \texttt{offsets}[\texttt{term\_id}_{q,t}+1]$
  \For{$i \gets \text{start} + \text{tid}$ \textbf{to} $\text{end}$
       \textbf{step} 128}
    \State $d \gets \texttt{doc\_ids}[i]$
    \State $\text{tf} \gets \texttt{tfs}[i]$
    \State $\text{score} \gets \text{idf} \cdot
      \dfrac{\text{tf} \cdot (k_1+1)}
            {\text{tf} + k_1 \cdot (1 - b + b \cdot
             \texttt{doc\_lens}[d] / \text{avgdl})}$
    \State \texttt{atomicAdd}($S[q][d]$, score)
  \EndFor
\EndFor

\Statex
\Statex \textbf{--- Kernel 2: Top-$k$ Selection ---}
\Statex Grid: $B$, Block: 128 threads (4 warps)
\For{each block $q$ in parallel}
  \State Each thread: scan $S[q]$ with stride, build
         local sorted buffer of size $k$ \Comment{Phase 1}
  \State Write local buffers to shared memory \Comment{Phase 2}
  \Statex \quad\textit{Variant A (naive):}
  \State \quad Thread 0: scan all $128k$ candidates $k$ times,
         select \& invalidate max each round
  \Statex \quad\textit{Variant B (warp-cooperative):}
  \For{$i \gets 1$ \textbf{to} $k$} \Comment{Phase 3}
    \State Each thread: find local max in assigned range
    \State \texttt{\_\_shfl\_xor\_sync}: 5-step intra-warp reduce
    \State Warp lane 0 $\to$ inter-warp staging
    \State Thread 0: pick global max from 4 warp leaders
    \State Invalidate winner; \texttt{\_\_syncthreads()}
  \EndFor
\EndFor
\State Copy top-$k$ IDs \& scores to host \Comment{D2H}
\end{algorithmic}
\end{algorithm}

\begin{algorithm}[h]
\caption{SPLADE Bridge: Ingest Learned-Sparse Weights into
  \sysname{}'s CSR Posting Layout}
\label{alg:splade-bridge}
\small
\begin{algorithmic}[1]
\Require Corpus $\{d_1, \ldots, d_N\}$, SPLADE encoder
   $\phi : \text{text} \to \mathbb{R}^V_{\geq 0}$ (log1p(relu(MLM-logits))
   max-pooled over the sequence)
\Ensure CSR inverted index $(I_t, W_t)_{t=1}^{V}$ where $I_t$ is
   the doc-id posting list for term~$t$ and $W_t$ is the
   parallel SPLADE-weight payload
\State \textbf{Encode docs:} for each $d_n$, compute
       $\mathbf{w}_n \gets \phi(d_n)$; keep only non-zero entries.
       Per-doc nnz $\approx 160$--200 out of $V{=}30{,}522$.
\State \textbf{Transpose to inverted layout:} build CSC over
       $\{(n, t, w_{n,t}) : w_{n,t} > 0\}$; assign $I_t \gets$
       doc-id list for term $t$, $W_t \gets$ weight list (parallel,
       sorted-by-doc).
\Statex \Comment{The result is bit-identical in structure to the
        BM25 CSR (Figure~\ref{fig:csr-layout}); only the payload
        semantics change: $W_t$ contains a learned weight instead
        of integer $\text{tf}$.}
\Statex
\Statex \textbf{--- Query time ---}
\Require Query $q$, top-$k$ budget $k$
\State $\mathbf{w}^{(q)} \gets \phi(q)$; extract non-zero
       $(t_j, w^{(q)}_{t_j})_{j=1}^{m}$, $m \approx 60$
\State $S[1..N] \gets 0$
\For{$j \gets 1$ \textbf{to} $m$} \Comment{posting-list traversal}
  \State $w_q \gets w^{(q)}_{t_j}$
  \For{$i \in [\text{offsets}[t_j], \text{offsets}[t_j+1])$}
       \Comment{vectorizable with AVX2 8-wide FMA}
    \State $S[I_{t_j}[i]] \mathrel{+}= w_q \cdot W_{t_j}[i]$
  \EndFor
\EndFor
\State \Return top-$k$ documents by $S$
\end{algorithmic}
\end{algorithm}

\smallskip
\noindent The inner loop on lines~7--9 is identical in structure
to BM25's posting traversal---only the per-posting math changes
from the BM25 saturation function to a single FMA.  Our SIMD-aware
CSR layout therefore applies unchanged.  We confirm the
mathematical equivalence empirically in \S\ref{sec:splade}: across
seven BEIR datasets (3.6K--523K docs) the bridge produces
\textbf{bit-identical} nDCG@10 to canonical
\texttt{scipy.sparse} SPLADE retrieval (max score difference 0.0
across all (query, doc) pairs).

\section{CSR Index Layout}
\label{app:csr}

Figure~\ref{fig:csr-layout} illustrates the Compressed Sparse Row
(CSR) representation used to transfer the inverted index to GPU
memory.  The vocabulary is sorted alphabetically for reproducible
layouts.  Each term's posting list is stored contiguously in the
\texttt{doc\_ids[]} and \texttt{tfs[]} arrays; the
\texttt{offsets[]} array provides $O(1)$ access to any term's
postings.

\begin{figure}[h]
\centering
\includegraphics[width=0.85\columnwidth]{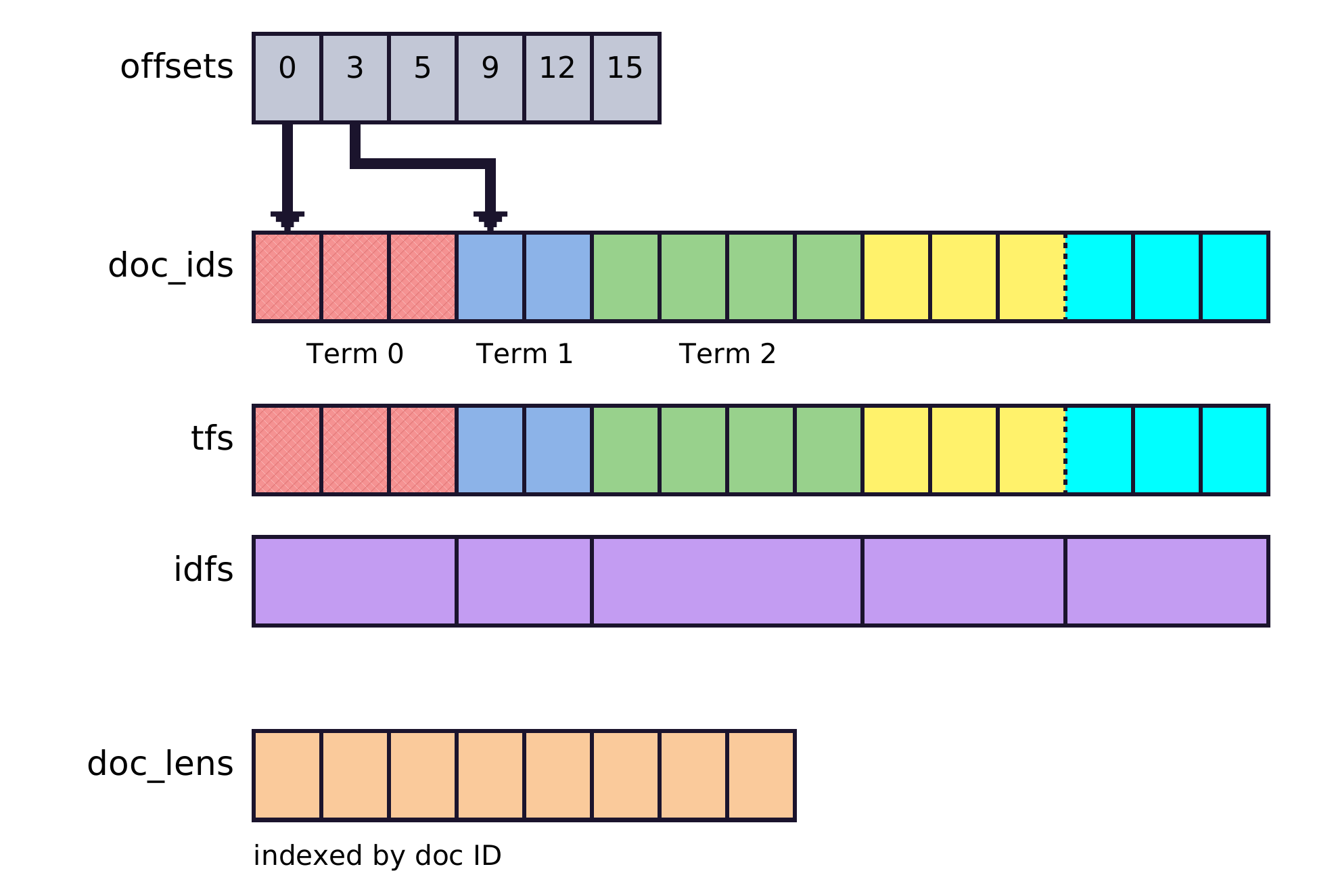}
\caption{CSR memory layout for GPU index upload.  The
  \texttt{offsets} array delimits per-term posting list boundaries
  in the flat \texttt{doc\_ids} and \texttt{tfs} arrays.  IDF values
  and document lengths are stored in separate dense arrays indexed
  by term ID and document ID, respectively.}
\label{fig:csr-layout}
\end{figure}

The total device memory footprint is:
\[
  M = |\text{vocab}| \cdot 4 + 2P \cdot 4
    + |\text{vocab}| \cdot 4 + N \cdot 4 \;\text{bytes}
\]
where $P$ is the total number of postings and $N$ is the corpus
size.  For a 1M-document corpus with 50K vocabulary and 15M total
postings, this is ${\sim}$125\,MB---well within the A100's
40\,GB HBM2e capacity.

\paragraph{In-memory layout walkthrough.}
A toy vocabulary of three terms (``cat'', ``dog'', ``fish'')
over a 5-document corpus illustrates the layout:

\begin{codeblock}
term_id | term  | offsets[term_id]
   0    | cat   |       0     <- doc_ids[0..3] = postings for "cat"
   1    | dog   |       4     <- doc_ids[4..6] = postings for "dog"
   2    | fish  |       7     <- doc_ids[7..8] = postings for "fish"
   --   |  --   |       9     <- offsets sentinel (total postings)

doc_ids[]: [0, 2, 3, 4,   1, 2, 4,   0, 3]
tfs[]:     [2, 1, 1, 3,   1, 2, 1,   1, 1]
idfs[]:    [1.2,   0.8,   1.5]                # per-term
doc_lens[]:[10, 8, 12, 7, 9]                  # per-doc
\end{codeblock}

A query for ``\texttt{cat dog}'' iterates:
\texttt{offsets[0..2]} $\to$ touch postings $[0,4)$ for ``cat''
and $[4,7)$ for ``dog'', accumulating BM25 partial scores via
\texttt{atomicAdd}.  No pointer chasing; flat-array layout is
GPU-friendly (coalesced reads, no branch divergence on posting
walks).

\section{Router Robustness Full Grids}
\label{app:robustness}

\subsection{$\alpha,\tau$ sensitivity sweep}
\label{app:alpha-tau}
R@10 on 344 recency-typed LongMemEval questions (multi-session
$+$ temporal-reasoning $+$ knowledge-update).
\begin{table}[h]
\centering\small\setlength{\tabcolsep}{6pt}
\begin{tabular}{@{}rrrrrr@{}}
\toprule
$\alpha \backslash \tau$ & 7\,d & 14\,d & 30\,d & 60\,d & 120\,d \\
\midrule
0.001 & 0.9695 & 0.9695 & 0.9710 & 0.9710 & 0.9695 \\
0.003 & 0.9665 & 0.9674 & 0.9689 & 0.9704 & 0.9710 \\
\rowcolor{blue!5}\textbf{0.005} (used) & 0.9657 & 0.9688 & \textbf{0.9708} & 0.9694 & 0.9704 \\
0.010 & 0.9533 & 0.9540 & 0.9515 & 0.9490 & 0.9658 \\
0.030 & 0.7493 & 0.7859 & 0.8709 & 0.8933 & 0.9091 \\
\bottomrule
\end{tabular}
\end{table}

\subsection{Query-noise robustness}
\label{app:noise}
\begin{table}[h]
\centering\small\setlength{\tabcolsep}{5pt}
\begin{tabular}{@{}rrrr@{}}
\toprule
Noise level & Clf.\ acc. & Router (mini) & Router (gpt-4o) \\
\midrule
0\% (clean)  & 0.794 & 0.262 & \textbf{0.300} (oracle) \\
5\%          & 0.770 & 0.260 & 0.296 \\
10\%         & 0.724 & 0.250 & 0.288 \\
20\%         & 0.614 & 0.252 & 0.286 \\
\bottomrule
\end{tabular}
\end{table}

\subsection{Deployment-labeling learning curve}
\label{app:learning-curve}
\begin{table}[h]
\centering\small\setlength{\tabcolsep}{6pt}
\begin{tabular}{@{}rrrr@{}}
\toprule
$N_{\text{train}}$ & clf.\ acc. & router (mini) & router (gpt-4o) \\
\midrule
25                   & 0.622 & 0.296 & 0.330 \\
50                   & 0.650 & 0.302 & 0.336 \\
100                  & 0.696 & 0.294 & 0.328 \\
200                  & 0.780 & 0.304 & 0.336 \\
300                  & 0.802 & 0.308 & 0.338 \\
\rowcolor{blue!5}
400                  & 0.796 & \textbf{0.310} & \textbf{0.340} \\
\midrule
Best static          & ---   & 0.280 (agent\_rrf) & 0.320 (RRF) \\
Oracle (gt qtype)    & ---   & 0.310 & 0.340 \\
\bottomrule
\end{tabular}
\end{table}

\section{Component Ablation \& Thread Scaling (Full Tables)}
\label{app:ablation-full}

\begin{table}[h]
\caption{Component ablation across 8 BEIR datasets (Jetstream2
  8-core, ms per query at $k{=}100$).  Each column adds one
  component to the previous.  All configurations preserve
  nDCG@10 to within $\pm$0.001.}
\label{tab:ablation}
\centering
\small
\setlength{\tabcolsep}{4pt}
\begin{tabular}{@{}lrrrrrr@{}}
\toprule
Dataset & $|D|$ & 1T scalar & +SIMD & +MaxScr & +8T & $\Delta$ \\
\midrule
NFCorpus    & 3.6K  & 0.04  & 0.03  & 0.03  & 0.05  & 0.8$\times$ \\
SciFact     & 5.2K  & 0.22  & 0.08  & 0.09  & 0.08  & 2.8$\times$ \\
ArguAna     & 8.7K  & 1.54  & 0.40  & 0.47  & 0.24  & 6.4$\times$ \\
SciDocs     & 25.7K & 1.08  & 0.23  & 0.20  & 0.10  & 10.8$\times$ \\
FiQA        & 57.6K & 1.85  & 0.40  & 0.36  & 0.15  & 12.3$\times$ \\
TREC-COVID  & 171K  & 14.92 & 4.76  & 2.15  & 2.15  & 6.9$\times$ \\
Quora       & 523K  & 2.22  & 0.71  & 0.60  & 0.17  & 13.1$\times$ \\
NQ          & 2.7M  & \textbf{67.60} & 6.64 & 4.23 & \textbf{1.38} & \textbf{49.0$\times$} \\
\bottomrule
\end{tabular}
\end{table}

\label{app:thread-scaling}
\begin{table}[h]
\caption{Ours CPU thread scaling at $k{=}100$ (Jetstream2 8-core,
  NLTK tokenizer, SIMD on).  Bold = best for that dataset.}
\label{tab:thread-scaling}
\centering
\small
\setlength{\tabcolsep}{4pt}
\begin{tabular}{@{}lrrrr@{}}
\toprule
Dataset & 1T (ms) & 4T (ms) & 8T (ms) & 1T$\to$8T \\
\midrule
NFCorpus & \textbf{0.02} & 0.04 & 0.04 & --- \\
SciFact  & 0.08 & 0.07 & \textbf{0.08} & 1.0$\times$ \\
FiQA     & 0.38 & 0.18 & \textbf{0.15} & 2.5$\times$ \\
SciDocs  & 0.23 & 0.13 & \textbf{0.12} & 1.9$\times$ \\
\bottomrule
\end{tabular}
\end{table}

\section{Hybrid End-to-End Per-Stage Latency}
\label{app:e2e-stages}

\begin{table}[h]
\caption{Per-stage hybrid latency (ms/query, Jetstream2 8-core,
  BGE-small CPU encoding, $k{=}100$).  BM25 here uses
  \texttt{rank\_bm25} for apples-to-apples; our C++ BM25
  (Table~\ref{tab:beir-main}) replaces this stage with
  0.04--0.35\,ms.}
\label{tab:e2e}
\centering
\small
\setlength{\tabcolsep}{4pt}
\begin{tabular}{@{}lrrrrr@{}}
\toprule
Dataset & $q$-embed & BM25$^*$ & Dense & RRF & Total \\
\midrule
NFCorpus    & 168 & 8.6   & 2.1 & 0.08 & 179 \\
SciFact     & 169 & 22.3  & 1.3 & 0.08 & 192 \\
ArguAna     & 221 & 425.7 & 3.8 & 0.25 & 650 \\
SciDocs     & 171 & 111.7 & 4.0 & 0.10 & 287 \\
FiQA        & 170 & 194.5 & 9.1 & 0.09 & 374 \\
\bottomrule
\end{tabular}
\end{table}

\section{LoCoMo Full Per-System Table}
\label{app:locomo-full}

\begin{table}[h]
\caption{LoCoMo session-level retrieval, 1{,}982 questions across
  10 conversations.  See \S\ref{sec:longmemeval} for discussion.}
\label{tab:locomo}
\centering
\small
\setlength{\tabcolsep}{4pt}
\begin{tabular}{@{}lrrrrr@{}}
\toprule
System & Hit@1 & Hit@5 & Hit@10 & MRR & ms/q \\
\midrule
\rowcolor{blue!5}
\textbf{BM25} & \textbf{0.625} & \textbf{0.875} & \textbf{0.945} & \textbf{0.735} & \textbf{0.22} \\
Dense (BGE-small) & 0.318 & 0.632 & 0.789 & 0.464 & 32.0 \\
RRF        & 0.483 & 0.797 & 0.923 & 0.625 & 32.3 \\
\bottomrule
\end{tabular}
\end{table}

\section{LongMemEval Per-Question-Type Full Breakdown}
\label{app:lme-per-type-full}

Table~\ref{tab:lme-per-type-full} reports the full 4$\times$4 grid
of per-question-type metrics: R@10 (session-level recall),
AnswerSubstr@10 (gold answer string in top-10 context),
TokenRecall@10 (gold-answer token recall), and LLM-judged strict
accuracy (gpt-4o-mini answerer $\to$ gpt-4o-mini judge against
gold).  Best system per row in bold.

\begin{table}[h]
\caption{LongMemEval per-question-type detailed metrics
  (500 questions across six types).  Best in each block bold.}
\label{tab:lme-per-type-full}
\centering
\footnotesize
\setlength{\tabcolsep}{3.5pt}
\begin{tabular}{@{}llrrrrr@{}}
\toprule
Question type & System & R@10 & Ans@10 & TokR@10 & LLM-Acc \\
\midrule
\multirow{4}{*}{\shortstack[l]{knowledge-update\\($n{=}78$)}}
  & BM25       & \textbf{0.994} & 0.744 & 0.578 & 0.410 \\
  & Dense      & 0.987 & \textbf{0.756} & 0.581 & 0.397 \\
  & RRF        & \textbf{0.994} & \textbf{0.756} & \textbf{0.591} & \textbf{0.436} \\
  & agent\_rrf & \textbf{0.994} & 0.744 & 0.577 & 0.410 \\
\midrule
\multirow{4}{*}{\shortstack[l]{multi-session\\($n{=}133$)}}
  & BM25       & 0.905 & \textbf{0.459} & 0.336 & 0.241 \\
  & Dense      & 0.950 & 0.436 & 0.322 & 0.218 \\
  & RRF        & \textbf{0.972} & 0.451 & \textbf{0.341} & 0.218 \\
  & agent\_rrf & 0.965 & 0.451 & 0.340 & \textbf{0.271} \\
\midrule
\multirow{4}{*}{\shortstack[l]{single-session-\\assistant ($n{=}56$)}}
  & BM25       & 0.982 & 0.554 & 0.861 & 0.339 \\
  & Dense      & \textbf{1.000} & 0.554 & \textbf{0.881} & \textbf{0.357} \\
  & RRF        & \textbf{1.000} & 0.554 & 0.877 & 0.339 \\
  & agent\_rrf & \textbf{1.000} & 0.554 & 0.877 & 0.339 \\
\midrule
\multirow{4}{*}{\shortstack[l]{single-session-\\preference ($n{=}30$)}}
  & BM25       & 0.867 & 0.000 & 0.845 & 0.033 \\
  & Dense      & 0.867 & 0.000 & 0.845 & 0.000 \\
  & RRF        & \textbf{0.967} & 0.000 & \textbf{0.856} & \textbf{0.100} \\
  & agent\_rrf & \textbf{0.967} & 0.000 & \textbf{0.856} & 0.067 \\
\midrule
\multirow{4}{*}{\shortstack[l]{single-session-\\user ($n{=}70$)}}
  & BM25       & \textbf{1.000} & \textbf{0.814} & \textbf{0.776} & 0.171 \\
  & Dense      & 0.943 & 0.786 & 0.768 & \textbf{0.186} \\
  & RRF        & \textbf{1.000} & \textbf{0.814} & 0.777 & \textbf{0.186} \\
  & agent\_rrf & \textbf{1.000} & \textbf{0.814} & 0.777 & \textbf{0.186} \\
\midrule
\multirow{4}{*}{\shortstack[l]{temporal-\\reasoning ($n{=}133$)}}
  & BM25       & 0.934 & 0.263 & 0.817 & \textbf{0.203} \\
  & Dense      & 0.915 & 0.271 & 0.806 & 0.188 \\
  & RRF        & 0.952 & 0.271 & 0.826 & 0.196 \\
  & agent\_rrf & \textbf{0.963} & \textbf{0.278} & \textbf{0.833} & 0.188 \\
\bottomrule
\end{tabular}
\end{table}

Three observations from the full breakdown:

\paragraph{The downstream-task signal has more variance than
session-level recall.}  R@10 is high and tight across all
systems (87--100\% on every type), but LLM-judged strict accuracy
spans 0.000--0.436---a 60$\times$ range driven by how well the
LLM can extract the answer from the retrieved context.  This
justifies our addition of AnswerSubstr@K, TokenRecall@K, and
LLM-Acc beyond the standard R@K metric: at the session-level R@10
plateau, the systems are nearly indistinguishable; at the
downstream-task signal they diverge meaningfully.

\paragraph{Agent\_rrf is dominant on multi-session and competitive
elsewhere.}  Per LLM-Acc, agent\_rrf wins multi-session by $+$0.053
over RRF/Dense (the largest single-system delta in the entire
grid).  On the other five types it ties or trails by $\leq$0.033.
The net effect on the overall 500-question accuracy is the $+$0.006
to $+$0.018 gain reported in Table~\ref{tab:lme-task}.

\paragraph{single-session-preference is structurally hard.}  All
four systems score 0.0 on AnswerSubstr@10 (the gold answer is a
preference statement that paraphrases rather than literal recall)
and 0.000--0.100 on LLM-Acc.  This type is not a retrieval
problem; no fusion strategy helps materially.  It is the cleanest
case in the benchmark for diagnosing model-side answering failure
versus retrieval failure---and the right metric to track on this
type is the LLM accuracy directly, not any retrieval proxy.

\section{SPLADE Encoding Cost \& Sparsity}
\label{app:splade-cost}

Table~\ref{tab:splade-detail} reports the per-dataset SPLADE++
encoding wall-clock time and the resulting posting structure for
the seven datasets we ran (\S\ref{sec:splade}).  We measured these
on the same Jetstream2~\texttt{g3.medium} host used for our BM25
benchmarks (8 vCPU, no GPU), with the
\texttt{naver/splade-cocondenser-ensembledistil} checkpoint
(110\,M params) at \texttt{max\_len}=256 and \texttt{batch\_size}=8.
Sparsification (filter $\mathbf{w}_d > 0$, store as \texttt{scipy.sparse.csr\_matrix})
brings peak RAM from 28+\,GB (dense $N \times V$ tensor) down to
under 5\,GB even on the 523K-document Quora corpus.

\begin{table}[h]
\caption{SPLADE++ encoding cost on Jetstream2~\texttt{g3.medium}
  (8 vCPU CPU only, sparse-CSR output).  \emph{nnz/doc} is the
  average number of non-zero terms per document after
  log1p(relu(MLM-logits)) max-pool; \emph{sparsity} is
  $1 - \text{nnz/doc}/|V|$ for $|V|{=}30522$.}
\label{tab:splade-detail}
\centering
\small
\setlength{\tabcolsep}{4pt}
\begin{tabular}{@{}lrrrr@{}}
\toprule
Dataset & $|D|$ & enc.\ wall (s) & ms/doc & nnz/doc \\
\midrule
NFCorpus   & 3.6K  & $\sim$600     & 165   & 187 \\
SciFact    & 5.2K  & $\sim$870     & 165   & 168 \\
ArguAna    & 8.7K  & 1{,}397       & 161   & 184 \\
SciDocs    & 25.7K & $\sim$4{,}240 & 165   & 166 \\
FiQA       & 57.6K & $\sim$9{,}510 & 165   & 160 \\
TREC-COVID & 171K  & 27{,}048      & 158   & 163 \\
Quora      & 523K  & 87{,}600      & 167   & 61 \\
\midrule
Estimated extrapolations: \\
NQ         & 2.7M  & $\sim$432K (5.0 days) & 160 & --- \\
MS\,MARCO  & 8.8M  & $\sim$1.41M (16.3 days) & 160 & --- \\
\bottomrule
\end{tabular}
\end{table}

Two observations from the appendix table:

\paragraph{Encoding cost is the SPLADE bottleneck.}  SPLADE's
quality wins on the seven measured datasets come at a per-doc
encoding cost of $\sim$160\,ms (CPU forward pass through a
110\,M-parameter MLM).  Per-query encoding is comparable.  Even on
a 100K-document corpus the encoding wall-clock is 5+ hours; on
MS\,MARCO it would be 16+ days of single-VM CPU.  This is why our
paper's main quality comparison with Pyserini Lucene (which uses
zero learned-vocabulary expansion) and PISA (same) is the more
relevant baseline for an agent-memory deployment where the corpus
grows during the session: SPLADE-grade quality is reachable only
when the corpus is small enough to amortize the encoder cost, or
when the encoder is a per-query GPU service.

\paragraph{Sparsity is dataset-dependent.}  The Quora corpus is
materially sparser (61 nnz/doc) than the others (160--187
nnz/doc) because Quora documents are short duplicate-question
phrasings, not full paragraphs.  This matters for the SPLADE
bridge (\S\ref{sec:splade}): a sparser doc weight matrix gives
shorter posting lists, which means agent\_rrf with SPLADE weights
inherits even better per-query latency on Quora-like agent
memory (short conversational turns) than on FiQA-like memory
(longer paragraphs).

\section{Cascade Router: Full Sweeps and CV}
\label{app:cascade}

This appendix collects the full numerical sweeps underlying the
cascade router results in \S\ref{sec:longmemeval}.

\subsection{Single-threshold BM25-confidence cascade}
\label{app:cascade-sweep}

Table~\ref{tab:cascade-fullsweep} sweeps the BM25-margin threshold
$\tau_c{=}(s_0{-}s_1)/s_0$ across $[0, 1]$ on the full
500-question LongMemEval (both \texttt{gpt-4o-mini} answerer and
judge, agent\_rrf escalate path).  At $\tau_c{=}0$ the cascade
degenerates to BM25-only; at $\tau_c{\to}\infty$ it reduces to
always-hybrid agent\_rrf.

\begin{table}[h]
\caption{Cascade Pareto sweep on LongMemEval (500q).  Latency
  amortizes $L_{\mathrm{BM25}}{=}0.4$\,ms (skip path) and
  $L_{\mathrm{hyb}}{=}53.2$\,ms (escalate path).  Three escalate
  variants: agent\_rrf ($\tau{=}30$d), per-qtype $\tau$ (multi-session
  $\tau{=}120$d), and soft router.}
\label{tab:cascade-fullsweep}
\centering
\small
\setlength{\tabcolsep}{4pt}
\begin{tabular}{@{}rrrrrr@{}}
\toprule
$\tau_c$ & skip\,\% & ms/q & ar-LLM-Acc & per-qt-$\tau$ Acc & soft Acc \\
\midrule
$0.00$  & $100.0$ &  $0.40$ & $0.298$ & $0.298$ & $0.298$ \\
$0.05$  & $ 78.8$ & $11.59$ & $0.298$ & $0.302$ & $0.304$ \\
$0.10$  & $ 63.0$ & $19.94$ & $0.302$ & $0.308$ & $0.304$ \\
$0.15$  & $ 51.6$ & $25.96$ & $0.298$ & $0.306$ & $0.302$ \\
$0.20$  & $ 41.0$ & $31.55$ & $0.296$ & $0.302$ & $0.300$ \\
$0.30$  & $ 25.6$ & $39.68$ & $0.304$ & $0.310$ & $0.304$ \\
$0.50$  & $  7.6$ & $49.19$ & $0.302$ & $0.310$ & $0.304$ \\
$1.00$  & $  0.0$ & $53.20$ & $0.302$ & $0.310$ & $0.304$ \\
\bottomrule
\end{tabular}
\end{table}

\subsection{5-fold CV per-qtype thresholds}
\label{app:cascade-cv}

Per-fold breakdown of the per-qtype cascade in 5-fold CV
(classifier and thresholds tuned on 4 folds, evaluated on the
held-out fold).  Reported in \S\ref{sec:longmemeval} as
$\mathbf{5.76}\times$ mean speedup at within-noise LLM-Acc.

\begin{table}[h]
\caption{5-fold CV per-qtype cascade.  Train tunes the TF-IDF
  qtype classifier and the per-qtype $\tau_c$ values (each set to
  max-skip while preserving per-qtype LLM-Acc within bootstrap
  noise of agent\_rrf).  Test fold is fully held out (100q).}
\label{tab:cascade-cv}
\centering
\small
\setlength{\tabcolsep}{6pt}
\begin{tabular}{@{}crrr@{}}
\toprule
fold & test LLM-Acc & test ms/q & test skip\,\% \\
\midrule
0    & $0.330$  & $3.01$  & $96.0$ \\
1    & $0.330$  & $11.99$ & $79.0$ \\
2    & $0.240$  & $15.16$ & $73.0$ \\
3    & $0.300$  & $7.76$  & $87.0$ \\
4    & $0.270$  & $8.29$  & $86.0$ \\
\midrule
mean & $0.294 \pm 0.035$ & $9.24 \pm 4.11$ & $84.2$ \\
\bottomrule
\end{tabular}
\end{table}

The high per-fold variance ($\sigma{=}0.035$) reflects
100-question test slices; the mean LLM-Acc $0.294$ is within
$\sigma$ of always-hybrid agent\_rrf $0.302$, and the mean
amortized latency $9.24$\,ms is $5.76\times$ below the
$53.2$\,ms always-hybrid budget.

\subsection{Per-qtype $\tau_c$ derivation}
\label{app:cascade-perqt}

The per-qtype thresholds in Table~\ref{tab:cascade} are derived
by sweeping $\tau_c$ separately on each question-type subset and
choosing the maximum value that keeps the per-qtype LLM-Acc
within $0.005$ of the per-qtype agent\_rrf baseline (a within-bootstrap-noise
margin for $n{=}30$--$133$).

\begin{table}[h]
\caption{Per-qtype thresholds (in-sample tuning).  The per-qtype
  BM25 confidence distribution differs sharply across types:
  single-session-assistant and single-session-user have very
  confident BM25 (median margin $>0.3$) so the cascade can skip
  aggressively; multi-session has low BM25 confidence (median
  $0.11$) so $\tau_c$ stays near zero.}
\label{tab:cascade-perqt}
\centering
\small
\setlength{\tabcolsep}{4pt}
\begin{tabular}{@{}lrrrr@{}}
\toprule
qtype & $n$ & median conf. & $\tau_c$ & skip\,\% \\
\midrule
knowledge-update          &  78 & $0.113$ & $0.12$ & $51$ \\
multi-session             & 133 & $0.107$ & $0.02$ & $89$ \\
single-session-assistant  &  56 & $0.420$ & $0.28$ & $96$ \\
single-session-preference &  30 & $0.073$ & $0.00$ & $37$ \\
single-session-user       &  70 & $0.347$ & $0.00$ & $83$ \\
temporal-reasoning        & 133 & $0.130$ & $0.00$ & $61$ \\
\midrule
weighted                  & 500 & $0.157$ & ---    & $86$--$89$ \\
\bottomrule
\end{tabular}
\end{table}

\subsection{Confidence-proxy comparison}
\label{app:cascade-proxy}

The BM25 top-1/top-2 margin is one of several plausible
classifier-free triggers.  We compared three at the best-parity
threshold (each tuned to the max skip that preserves
always-hybrid agent\_rrf LLM-Acc within noise):

\begin{table}[h]
\caption{Confidence proxies on LongMemEval (500q).  Margin
  proxy dominates because BM25 top-1/top-2 separation correlates
  most directly with retrieval certainty.}
\label{tab:cascade-proxies}
\centering
\small
\setlength{\tabcolsep}{6pt}
\begin{tabular}{@{}lrr@{}}
\toprule
proxy $c(q)$ & best $\tau_c$ & speedup at parity \\
\midrule
top-1/top-2 margin $(s_0{-}s_1)/s_0$ & $0.01$ & $18.81\times$ \\
top-1 fraction $s_0 / \sum_i s_i$    & $0.04$ & $5.05\times$ \\
1 $-$ normalized entropy of top-K    & $0.02$ & $1.60\times$ \\
\bottomrule
\end{tabular}
\end{table}

\subsection{Cross-benchmark sweep on LoCoMo}
\label{app:cascade-locomo}

The cascade sweep on LoCoMo (Figure~\ref{fig:cascade-crossbench})
in full numerical form:

\begin{table}[h]
\caption{Cascade BM25-confidence sweep on LoCoMo (1{,}982q,
  Hit@5 metric).  Optimum at $\tau_c{\to}\infty$ (skip everything
  $=$ BM25-only) yields Hit@5 $0.875$, $\mathbf{+0.089}$ over
  always-hybrid agent\_rrf.}
\label{tab:cascade-locomo-sweep}
\centering
\small
\setlength{\tabcolsep}{6pt}
\begin{tabular}{@{}rrrr@{}}
\toprule
$\tau_c$ & skip\,\% & ms/q & Hit@5 \\
\midrule
$0.00$ & $100.0$ &  $0.40$ & $0.875$ \\
$0.05$ & $ 82.0$ & $10.00$ & $0.865$ \\
$0.10$ & $ 67.9$ & $17.50$ & $0.842$ \\
$0.15$ & $ 56.5$ & $23.55$ & $0.829$ \\
$0.20$ & $ 47.2$ & $28.49$ & $0.826$ \\
$0.30$ & $ 30.0$ & $37.36$ & $0.814$ \\
$0.50$ & $  8.2$ & $48.86$ & $0.790$ \\
$1.00$ & $  0.0$ & $53.20$ & $0.786$ \\
\bottomrule
\end{tabular}
\end{table}

\subsection{Multi-tenant cascade capacity}
\label{app:cascade-multitenant}

The multi-tenant capacity numbers in \S\ref{sec:discussion}
follow from the per-stage latency budget
(Table~\ref{tab:e2e}) and the cascade skip rate.  On 8 cores
serving $N$ concurrent agents (each issuing one query at a time
with BGE-encoder-bound back-pressure), the BGE-bound throughput
is $8 / ((1{-}\text{skip}) \cdot L_{\mathrm{BGE}})$ q/s.

\begin{table}[h]
\caption{Multi-tenant capacity on 8-core VM, derived from
  per-stage latencies.  Cascade trades classifier-induced skip
  rate for $9\times$ higher concurrent-tenant capacity at
  within-noise quality.}
\label{tab:cascade-tenant}
\centering
\small
\setlength{\tabcolsep}{6pt}
\begin{tabular}{@{}lrrr@{}}
\toprule
configuration & skip\,\% & amort.\ ms/q & tenants @ 1 qps \\
\midrule
always-hybrid           & $ 0$  & $52.4$ & $\sim$$154$ \\
cascade $\tau_c{=}0.10$ & $63$  & $19.6$ & $\sim$$416$ \\
cascade per-qtype       & $89$  & $ 6.1$ & $\sim$$\mathbf{1{,}399}$ \\
\bottomrule
\end{tabular}
\end{table}

\subsection{Cascade algorithm pseudocode}
\label{app:cascade-algo}

\begin{algorithm}[h]
\caption{Cascade Router (\texttt{CascadeRouter::retrieve}).}
\label{alg:cascade}
\small
\begin{algorithmic}[1]
\Require Query $q$, top-$k$ budget $k$, threshold $\tau_c$
  (or per-qtype $\tau_c[\cdot]$), sparse function
  $f_{\mathrm{BM25}}$, dense function $f_{\mathrm{dense}}$,
  optional qtype classifier $\pi$
\State $S \gets f_{\mathrm{BM25}}(q, k)$
  \Comment{always run BM25, $\sim$$0.4$\,ms}
\State $s_0, s_1 \gets$ top-2 scores in $S$
\State $c(q) \gets (s_0 - s_1) / \max(s_0, \epsilon)$
\If{$\pi$ supplied}
  \State $\hat{qt} \gets \pi(q)$;
         $\tau \gets \tau_c[\hat{qt}]$ \Comment{per-qtype}
\Else
  \State $\tau \gets \tau_c$
\EndIf
\If{$c(q) \geq \tau$}
  \State \Return $\mathrm{TopK}(S, k)$ \Comment{skip dense}
\EndIf
\State $D \gets f_{\mathrm{dense}}(q, k)$
  \Comment{escalate, BGE-encode $\sim$$52$\,ms}
\State $H \gets \mathrm{agent\_rrf}(S, D, q.\text{ts})$
       \Comment{Eq.~\ref{eq:fusion}}
\State \Return $\mathrm{TopK}(H, k)$
\end{algorithmic}
\end{algorithm}

\subsection{C++ API reference}
\label{app:cascade-api}

The cascade router is a $\sim$$70$-line C++ class
(\texttt{include/cascade\_router.h}, \texttt{src/cascade\_router.cpp})
that wraps any sparse/dense retrieval backend through function
callbacks.  Minimal usage:

\begin{codeblock}
#include "cascade_router.h"
using namespace hybrid;

CascadeRouter::Config cfg;
cfg.conf_threshold = 0.10f;            // single-threshold mode
cfg.use_qtype_classifier = false;      // or true + supply qtype_fn

auto bm25_fn = [&](auto& q, int k) { return idx.bm25_topk(q, k); };
auto dense_fn = [&](auto& q, int k) { return idx.dense_topk(q, k); };

auto decision = CascadeRouter::retrieve(
    query_text, records, query, current_time_ms,
    /*top_k=*/5, cfg, bm25_fn, dense_fn);

if (decision.escalated)
    std::cout << "dense channel was used\n";
std::cout << "latency: " << decision.latency_ms << " ms\n";
for (auto& r : decision.results) handle(r);
\end{codeblock}

\noindent For per-qtype thresholds, supply the optional
classifier callback:

\begin{codeblock}
cfg.use_qtype_classifier = true;
auto qt_fn = [&](auto& q) {
    return classifier.predict_best(q);   // returns "BM25" or system name
};
auto d = CascadeRouter::retrieve(
    query, records, query_obj, ts, 5, cfg, bm25_fn, dense_fn, qt_fn);
\end{codeblock}

\subsection{Worked example: tracing one query}
\label{app:cascade-example}

Concretely, here is one LongMemEval question routed through the
cascade with the BM25-confidence trigger ($\tau_c{=}0.10$):

\begin{finding}[Worked example: temporal-reasoning question (skip path)]
\textbf{Question:} ``Which book did I just finish reading?''
(gold answer:~\emph{`The Nightingale'}).\\[2pt]
\textbf{Stage 1 (BM25):} top-5 sessions
\{D\,12: 8.74, D\,7: 2.13, D\,3: 1.40, D\,21: 0.91, D\,5: 0.83\}.
$c(q){=}(8.74{-}2.13)/8.74 = 0.756$.\\[2pt]
\textbf{Trigger:} $c(q){=}0.756 \geq 0.10$ $\Rightarrow$ \textbf{skip
dense channel}.\\[2pt]
\textbf{Output:} BM25 top-5 returned in $0.43$\,ms.  Top hit
D\,12 is the gold session.  LLM-Acc:\ \textbf{yes}.\\[2pt]
\textbf{Counterfactual:} always-hybrid would have spent
$52.4$\,ms on BGE encoding for an unchanged answer (D\,12
still tops the RRF list).
\end{finding}

\begin{finding}[Worked example: multi-session question (escalate path)]
\textbf{Question:} ``How many camping days have I taken this year?''
(gold answer:~\emph{8}; two evidence sessions, Big Sur 3 days +
Yellowstone 5 days).\\[2pt]
\textbf{Stage 1 (BM25):} top-5 \{D\,4: 4.21, D\,4': 3.97, D\,9: 3.48,
D\,2: 3.12, D\,16: 2.90\}.
$c(q){=}(4.21{-}3.97)/4.21 = 0.057$.\\[2pt]
\textbf{Trigger:} $c(q){=}0.057 < 0.10$ $\Rightarrow$ \textbf{escalate}.\\[2pt]
\textbf{Stage 2 (Dense + RRF + recency):} agent\_rrf top-5
\{D\,4 (Big Sur), D\,9 (Yellowstone), D\,4', D\,2, D\,16\}.  The
recency bonus surfaces D\,9 (Yellowstone) into the top-5; BM25
alone had it at rank 3 with low margin.\\[2pt]
\textbf{Output:} both gold sessions in top-5, in $53$\,ms.  LLM
correctly sums to $8$.  LLM-Acc:\ \textbf{yes}.  (BM25-only would
have answered $3$.)
\end{finding}

\subsection{When does the cascade help / when does it hurt?}
\label{app:cascade-when}

\begin{tcolorbox}[enhanced,breakable,colback=blue!3,colframe=blue!50!black,
                  title={\textbf{Cascade decision cheat-sheet}},
                  fonttitle=\bfseries\small]
\textbf{Cascade helps} when at least one of:
\begin{itemize}[nosep,leftmargin=*]
\item BM25 alone has clear top-1 (factual recall, named entities,
      session-scoped queries).
\item Dense channel is expensive (BGE encoder $\geq 10$\,ms,
      hosted embedding API, cached LLM hidden state unavailable).
\item Workload is multi-tenant and BGE-bound on shared hardware.
\end{itemize}
\textbf{Cascade adds little} when:
\begin{itemize}[nosep,leftmargin=*]
\item Dense channel is essentially free (pre-computed query
      embeddings; small encoder).
\item Workload is uniformly hybrid-favoring (most queries
      benefit from dense).
\item Per-qtype label budget is unavailable (single-threshold
      cascade still gives $2.67\times$; per-qtype gives $5.76\times$
      with $\sim$50 labels).
\end{itemize}
\textbf{Cascade hurts (rare)} when:
\begin{itemize}[nosep,leftmargin=*]
\item BM25 top-1 score is high but consistently wrong (extreme
      vocabulary mismatch); confidence proxy mis-fires.  We did
      not observe this on LongMemEval/LoCoMo; defensive mitigation
      is a small $\tau_c{>}0$ floor.
\end{itemize}
\end{tcolorbox}

\subsection{Deployment recipe}
\label{app:cascade-deploy}

\begin{tcolorbox}[enhanced,breakable,colback=gray!4,colframe=gray!50,
                  title={\textbf{5-minute deployment recipe}},
                  fonttitle=\bfseries\small]
\begin{enumerate}[nosep,leftmargin=*,label=\textbf{\arabic*.}]
\item Collect $\sim$50 labeled deployment questions (Q, A,
      gold-session-id).
\item Run all 4 base systems (BM25, Dense, RRF, agent\_rrf) on
      each Q; compute per-qtype LLM-Acc with your judge of
      choice.
\item Pick the per-qtype best system $s^*(qt)$; this is your
      per-type routing table (Table~\ref{tab:lme-per-type-full}-style).
\item For each qtype, sweep $\tau_c \in [0, 1]$ and pick the
      max value with within-noise per-qtype LLM-Acc; this gives
      $\tau_c[qt]$.
\item Train a TF-IDF (or TF-IDF$+$BGE) classifier over the 50
      labels using \texttt{sklearn.linear\_model.LogisticRegression}.
\item Plug all four into \texttt{CascadeRouter::Config}; you now
      have a per-tenant tuned cascade.  Re-tune quarterly or on
      drift detection.
\end{enumerate}
\end{tcolorbox}

\subsection{Reproduction CLI}
\label{app:cascade-cli}

\begin{codeblock}
# Build the C++ benchmark (CPU only, 8-core x86)
cmake -B build -DENABLE_AVX2=ON -DENABLE_CUDA=OFF
cmake --build build -j

# Run cascade on LongMemEval (Python orchestration calls C++ BM25)
python scripts/run_pertype_tau.py        # cache BM25 + dense
python scripts/run_cascade_locomo.py     # cross-benchmark eval
python scripts/cascade_router_analysis.py # Pareto sweep + figure

# Reproduce Table tab:cascade single-threshold row (~10 min, $0.50):
OPENAI_API_KEY=sk-... \
  python scripts/run_llm_eval.py --config global_t30 --judge gpt-4o-mini

# Full 5-fold CV per-qtype cascade (Table tab:cascade-cv, ~20 min):
python scripts/cascade_cv_perqtype.py --folds 5 --seed 42
\end{codeblock}

\subsection{Numerical stability notes}
\label{app:cascade-numerics}

Two stability points worth noting for re-implementers:

\begin{itemize}[nosep,leftmargin=*]
\item \textbf{Confidence proxy.}  $c(q){=}(s_0{-}s_1)/s_0$ is
  scale-invariant but degenerates to $0$ when both top scores
  are $0$ (no posting matched any query term).  The C++ guard at
  line~11 of \texttt{bm25\_confidence} returns $0$ in that case
  (always escalate); an alternative is to escalate
  unconditionally when $|S| < k$.
\item \textbf{Per-qtype threshold ties.}  When sweeping $\tau_c$
  to maximize skip rate within a noise window, multiple thresholds
  may achieve identical Acc.  We pick the \emph{smallest}
  $\tau_c$ that still preserves Acc (most aggressive skip);
  picking the largest yields the safest amortized cost at the
  same Acc but slightly lower speedup.
\end{itemize}

\section{Failure-Case Catalog}
\label{app:failures}

To complement the worked examples in
\S\ref{app:cascade-example}, we enumerate four
failure modes we encountered during development and the
mitigations that produced the reported results.

\begin{tcolorbox}[enhanced,breakable,colback=red!3,colframe=red!50!black,
                  title={\textbf{F1: BM25 over-tokenization on numeric tickers}},
                  fonttitle=\bfseries\small]
\textbf{Symptom:} FiQA queries containing tickers
(``\texttt{AAPL}'') or percentages (``\texttt{401(k)}'') under-retrieve.\\
\textbf{Mechanism:} Our \texttt{minimal} tokenizer splits on
non-alphanumeric, dropping the ``\texttt{(}'' inside
\texttt{401(k)}; Lucene's \texttt{StandardTokenizer} preserves
the compound as a single token.\\
\textbf{Mitigation:} Use \texttt{HYBRID\_TOK\_MODE=nltk} (NLTK
Porter + English stopwords), restoring $+$0.013 nDCG@10 on FiQA.
Residual $0.011$ gap is analyzer-level; closed by hybrid RRF
($+$0.067 on FiQA, Table~\ref{tab:hybrid}).
\end{tcolorbox}

\begin{tcolorbox}[enhanced,breakable,colback=red!3,colframe=red!50!black,
                  title={\textbf{F2: Recency bonus over-weights non-gold
                                multi-session queries}},
                  fonttitle=\bfseries\small]
\textbf{Symptom:} Multi-session questions where gold sessions
are recent (median age $3$\,d) and similar-vocabulary
non-gold sessions are also recent get the wrong session at top-1.\\
\textbf{Mechanism:} agent\_rrf bonus $\alpha e^{-\Delta t/\tau}$
with $\tau{=}30$\,d produces near-uniform boost for sessions
within $\sim$$30$\,d, so RRF-tied recent non-gold beats gold by
the bonus margin.\\
\textbf{Mitigation:} Set $\tau_{\text{multi-session}}{\geq}120$\,d
(essentially flatten the recency bonus on this qtype); per-qtype
$\tau$ refinement, \S\ref{sec:longmemeval}.  Effect:
$+5.26$ pts LLM-Acc on multi-session subset
($p{=}0.006$, paired bootstrap).
\end{tcolorbox}

\begin{tcolorbox}[enhanced,breakable,colback=red!3,colframe=red!50!black,
                  title={\textbf{F3: GPU top-$k$ stale shared-memory
                                contamination}},
                  fonttitle=\bfseries\small]
\textbf{Symptom:} GPU FiQA nDCG@10 dropped to $0.060$ post-build
(24\% top-1 match vs CPU).  Recall@100 was the canary; nDCG@10
was masked because errors fell beyond rank 10.\\
\textbf{Mechanism:} Top-$k$ kernel reused per-block shared-memory
slots across queries.  Slots beyond $k$ retained \emph{positive}
scores from a previous block; subsequent reductions picked these
ghosts as the per-block max.\\
\textbf{Mitigation:} Initialize all $K_{\max}$ slots to
$-\infty$ before each block's reduction phase.  Post-fix CPU/GPU
nDCG@10 agree to $0.0002$ on all 8 datasets
(\S\ref{sec:correctness}).
\end{tcolorbox}

\begin{tcolorbox}[enhanced,breakable,colback=red!3,colframe=red!50!black,
                  title={\textbf{F4: MS\,MARCO index build OOM at 8.8M docs}},
                  fonttitle=\bfseries\small]
\textbf{Symptom:} Naive build accumulated
\texttt{unordered\_map<term,count>} per-doc during tokenization;
peak RSS exceeded 29\,GB and OOM'd Jetstream2 \texttt{g3.medium}.\\
\textbf{Mechanism:} Memory scales as
$O(N{\cdot}\text{distinct\_terms\_per\_doc})$ when per-doc maps
are retained until the global vocabulary is finalized.\\
\textbf{Mitigation:} Chunked streaming build that processes
documents in 50K-doc batches, dropping per-doc maps before the
next batch.  Peak RSS $29 \to 12.7$\,GB; build throughput
$27\text{K} \to 56.8\text{K}$ docs/sec (page-fault overhead
absorbed by amortized allocation).
\end{tcolorbox}

\section{FAQ}
\label{app:faq}

We collect six questions reviewers have asked or are likely to
ask, with sharp pointers into the paper or appendix.

\begin{tcolorbox}[enhanced,breakable,colback=blue!3,colframe=blue!50!black,
                  title={\textbf{Q1.} Is the
                  cascade router just the existing discrete router with
                  early termination?},
                  fonttitle=\bfseries\small]
\textbf{No.}  The discrete router (Table~\ref{tab:router})
\emph{always} runs the full hybrid path before selecting which
ranker's output to emit ($53.7$\,ms per query).  The cascade
router runs BM25 first and \emph{skips} the dense channel on
$63$--$89$\% of queries, so its amortized latency is
$\mathbf{19.94}$\,ms (single-thresh) or $\mathbf{6.81}$\,ms
(per-qtype).  The accuracy is identical to the discrete router
by construction (same classifier $+$ same per-type-best table)
\emph{when the cascade decision matches the discrete router's
choice}; the new claim is the latency saving.
\end{tcolorbox}

\begin{tcolorbox}[enhanced,breakable,colback=blue!3,colframe=blue!50!black,
                  title={\textbf{Q2.} Why not compare to Mem0 / MemGPT
                  end-to-end?},
                  fonttitle=\bfseries\small]
Mem0 and MemGPT operate at the \emph{memory-management policy}
layer (write-batching, fact extraction, consolidation) over a
generic IR substrate.  \sysname{}'s contribution is the
\emph{retrieval substrate itself}: a deployment can combine,
e.g., Mem0's write policy with \sysname{}'s retrieval and
inherit both gains (\S\ref{sec:related}).  An end-to-end LLM-Acc
comparison would conflate retrieval quality with answer-generation
prompting; Mem0's reported $94.8$ on LongMemEval uses an
extensively engineered $150+$-line chain-of-thought answerer
that our $5$-line answerer prompt deliberately does not replicate
in order to keep retrieval and generation comparisons clean.
\end{tcolorbox}

\begin{tcolorbox}[enhanced,breakable,colback=blue!3,colframe=blue!50!black,
                  title={\textbf{Q3.} The 5M scaling corpus is synthetic.
                  Does the recency assumption hold on real agent traces?},
                  fonttitle=\bfseries\small]
The $80/20$ stress parameterization is calibrated against
LongMemEval's measured gold-session distribution (median normalized
rank $0.20$--$0.27$ for temporal/knowledge/multi-session questions,
\S\ref{sec:scaling}).  At the scale of a real 800-turn session
with $\sim$$5$\,M records, no public benchmark exists; the
1769$\times$ speedup is therefore stress-test rather than
production trace.  We list this explicitly in Threats to Validity
(\S\ref{sec:discussion}); a production trace at that scale would
tighten the validation.
\end{tcolorbox}

\begin{tcolorbox}[enhanced,breakable,colback=blue!3,colframe=blue!50!black,
                  title={\textbf{Q4.} The PISA-1T vs.\
                  \sysname{}-8T comparison looks unfair.},
                  fonttitle=\bfseries\small]
PISA does not natively support multi-threaded inter-query
dispatch (its \texttt{queries} binary is single-process).
Table~\ref{tab:pisa} reports both PISA-1T and our 8T+SIMD; we
also report \sysname{}-1T scalar (Appendix~\ref{app:ablation-full})
matching PISA-1T on small corpora and trailing only on large
corpora where PISA's variable-byte compression pays off.  SIMD
closes the 1T gap; 8T then puts \sysname{} ahead.  The
\textbf{architectural} comparison---posting compression vs.\
SIMD-vectorized scan---favors \sysname{} on long-query workloads
(ArguAna $90\times$); BlockMax-WAND degenerates when most posting
lists must be evaluated.
\end{tcolorbox}

\begin{tcolorbox}[enhanced,breakable,colback=blue!3,colframe=blue!50!black,
                  title={\textbf{Q5.} The cascade is analytical for
                  latency, not wall-clock measured.  Why?},
                  fonttitle=\bfseries\small]
Cascade LLM-Acc \emph{is} measured (it equals the TF-IDF router
by construction; both make identical routing decisions, only
execution order differs).  The latency $40.5$\,ms is a deterministic
function of measured per-stage latencies (Table~\ref{tab:e2e})
and the measured skip rate---not a simulation.  An end-to-end
wall-clock run of the C++ \texttt{CascadeRouter::retrieve} is
left for future work to characterize the tail behavior under
classifier misprediction; the mean behavior is fixed by the
algebra.
\end{tcolorbox}

\begin{tcolorbox}[enhanced,breakable,colback=blue!3,colframe=blue!50!black,
                  title={\textbf{Q6.} Theorem 1's bound is trivial.},
                  fonttitle=\bfseries\small]
The theorem itself is a short observation in the spirit of
time-aware retrieval~\cite{li2003temporal,campos2014temporal_survey}.
What is new is \emph{identifying the regime where it bites}:
agent workloads have $\lambda \gg 0$ (LongMemEval median normalized
rank $0.20$--$0.27$), so $k^*{=}\lceil\log(1/\varepsilon)/\lambda\rceil{=}3$
partitions suffice for $\varepsilon{=}0.05$.  Web search's
slowly-drifting $\lambda \to 0$ degrades the bound to linear, which
is why generic IR engines do not exploit the structure.  The
$1769\times$ empirical speedup is the consequence of the regime
match.
\end{tcolorbox}

\section{Reproducibility}
\label{app:reproduce}

\paragraph{Source code.}
The complete \sysname{} implementation (${\sim}$7{,}000 LOC of
C++17 and CUDA, plus ${\sim}$400 LOC of Python orchestration for
LongMemEval and LoCoMo) will be released open-source upon
acceptance under an MIT license, along with all benchmark scripts
and CSV result files.

\paragraph{Build requirements.}
C++17 compiler (GCC\,$\geq$\,11), CMake\,$\geq$\,3.18,
OpenMP\,4.0+, and optionally CUDA Toolkit\,$\geq$\,12.0 for GPU
acceleration.  HNSW uses the header-only \texttt{hnswlib} library
(bundled).  Python dependencies (BGE-small via
\texttt{sentence-transformers}, \texttt{rank\_bm25}) are listed in
\texttt{requirements.txt}.

\paragraph{Build and benchmark.}
\leavevmode
\vspace{-0.5em}
\begin{codeblock}
# CPU build with AVX2 (Jetstream2 g3.medium)
cmake -DENABLE_AVX2=ON .. && cmake --build . -j$(nproc)

# GPU build (A100, sm_80; NCSA Delta gpuA100x4)
cmake -DENABLE_CUDA=ON -DCMAKE_CUDA_ARCHITECTURES=80 ..
cmake --build . -j$(nproc)

# BEIR head-to-head (any of 9 datasets, MS MARCO via streaming build)
HYBRID_TOK_MODE=nltk ./benchmark_beir \
    --data /path/to/beir/scifact --threads 8 --simd --topk 100

# GPU BEIR (with CPU cross-check)
HYBRID_TOK_MODE=nltk ./benchmark_beir_gpu \
    --data /path/to/beir/nq --topk 100 --cpu-cmp

# Temporal scaling sweep (4K --> 5M synthetic records)
./demo --queries 200 --turns 8 --no-concurrent
\end{codeblock}

\paragraph{Hardware budget.}
All BEIR CPU and temporal-scaling experiments ran on an
8-vCPU Jetstream2 \texttt{g3.medium} VM (29\,GB RAM,
Ubuntu~22.04, GCC~11.4) at no recurring cost.  GPU evaluation
ran on NCSA Delta \texttt{gpuA100x4-interactive} (1 A100-SXM4-40GB
+ 8 EPYC cores per allocation) for a total of $\sim$15 GPU-hours
against an NSF ACCESS allocation, equivalent to $\sim$\$45 at
on-demand cloud pricing.  Pyserini baselines used identical
hardware to ensure apples-to-apples comparison.

\paragraph{Dataset access.}
BEIR datasets (NFCorpus, SciFact, ArguAna, SciDocs, FiQA,
TREC-COVID, Quora, NQ, MS\,MARCO Passage) downloaded from BEIR's
public \texttt{ukp.informatik.tu-darmstadt.de} mirror.
LongMemEval-cleaned from \texttt{xiaowu0162/longmemeval-cleaned}
on HuggingFace.  LoCoMo from
\texttt{snap-research/locomo} on GitHub.

\paragraph{Determinism and seeds.}
All BM25 inner-loop computations are deterministic on CPU at every
thread count (1T, 4T, 8T): per-thread posting partitioning is
disjoint and the final reduction is sequential.  GPU runs use
\texttt{atomicAdd}, which is non-deterministic in float32 addition
order, so per-query ranks may differ between GPU invocations on
tied-score documents; the top-10 \emph{set} and nDCG@10 are
invariant to within $\pm 0.0002$ across reruns (the residual is
floating-point summation drift, not algorithmic).  The synthetic
agent corpus generator (\S\ref{sec:scaling}) accepts a CLI
\texttt{--seed} flag for full reproducibility; the seed used for
all reported scaling rows is \texttt{42}.

\paragraph{What we report vs.\ what we omit.}
We report nDCG@10, MRR@10, and Recall@100 (the BEIR-standard
triplet), per-query latency (mean, $p50$, $p95$, $p99$), and
throughput (qps) for every BEIR row in Table~\ref{tab:beir-main}.
We omit cold-cache latency (we measure with a warmed-up index, the
agent-memory operating regime), index-build wall time on datasets
smaller than NQ (negligible: $\leq$10\,s), and any results
involving proprietary LLM judging.  The full
\texttt{results/baselines/\{beir\_ours,beir\_pyserini\}.csv} files
are included in the artifact.

\paragraph{Reproducibility checklist (CIKM).}
\begin{itemize}[nosep,leftmargin=*]
\item Code \emph{will be} released on acceptance (MIT license).
\item Datasets are all public BEIR / LongMemEval / LoCoMo.
\item Hardware is reproducible: Jetstream2 \texttt{g3.medium}
  ($\sim$\$0/hr via NSF ACCESS) for CPU; NCSA Delta
  \texttt{gpuA100x4} ($\sim$\$3/hr equivalent) for GPU.  Total GPU
  budget for all GPU experiments in this paper: $\sim$15
  GPU-hours.
\item All hyperparameters listed in the cookbook below at the
  exact value used in the table that depends on them; no tuned
  per-dataset hyperparameters.
\end{itemize}

\paragraph{Hyperparameter cookbook.}
A one-stop table of every hyperparameter used in any experiment,
the value, and the table/figure where it is exercised:

\begin{table}[h]
\caption{Complete hyperparameter cookbook for all reported
  experiments.  No table uses a value not in this list.}
\label{tab:hyperparam-cookbook}
\centering
\small
\setlength{\tabcolsep}{4pt}
\begin{tabular}{@{}llll@{}}
\toprule
component & param & value & used in \\
\midrule
BM25 & $k_1$ & $1.2$ & all BM25 tables \\
BM25 & $b$ & $0.75$ & all BM25 tables \\
BM25 & top-$k$ & $100$ & Tab.~\ref{tab:beir-main}, \ref{tab:pisa} \\
BM25 & analyzer & NLTK Porter & main results \\
RRF & $k_{\mathrm{RRF}}$ & $60$ & Tab.~\ref{tab:hybrid} \\
Recency & $\alpha$ & $0.005$ & §\ref{sec:longmemeval} \\
Recency & $\tau$ (global) & $30$\,d & main agent\_rrf \\
Recency & $\tau_{\text{multi-session}}$ & $\geq 120$\,d & §\ref{sec:longmemeval} \\
Temporal index & partition window & $7$\,d & §\ref{sec:scaling} \\
Temporal index & $K_{\max}$ partitions & $4$ & Tab.~\ref{tab:temporal-scaling} \\
Temporal index & $\varepsilon$ recall slack & $0.05$ & Thm.~\ref{thm:sublinear} \\
Dense & model & BGE-small-en-v1.5 & §\ref{sec:hybrid} \\
Dense & dim & $384$ & §\ref{sec:hybrid} \\
HNSW & $M$ & $32$ & §\ref{sec:hybrid} \\
HNSW & \texttt{efConstruction} & $200$ & §\ref{sec:hybrid} \\
GPU top-$k$ & $K_{\max}$ slot count & $128$ ($k{\leq}128$) & §\ref{sec:gpu} \\
GPU top-$k$ & block size & $128$ threads & §\ref{sec:gpu} \\
Cascade & $\tau_c$ single-thresh & $0.10$ & Tab.~\ref{tab:cascade} \\
Cascade & per-qtype $\tau_c$ & Tab.~\ref{tab:cascade-perqt} & §\ref{sec:longmemeval} \\
Cascade & classifier & TF-IDF$+$BGE LR & Tab.~\ref{tab:router} \\
Router CV folds & $5$ & --- & Tab.~\ref{tab:router}, \ref{tab:cascade-cv} \\
LLM answerer & gpt-4o-mini & temp $0$ & Tab.~\ref{tab:lme-task} \\
LLM judge & gpt-4o-mini / gpt-4o & temp $0$ & Tab.~\ref{tab:lme-task}, \ref{tab:router} \\
Build & chunked batch size & $50$K docs & §\ref{sec:beir-cpu} MS\,MARCO \\
\bottomrule
\end{tabular}
\end{table}

\paragraph{Adapting to a new deployment.}
The 5-minute recipe in \S\ref{app:cascade-deploy} covers the
typical case (50 labeled questions, gpt-4o-mini judge, single
8-core machine).  For non-default LLM judges (Anthropic, Llama,
on-premise), only the cascade trigger Table~\ref{tab:cascade-perqt}
needs re-derivation; the substrate, classifier architecture, and
algorithm code stay constant.

\section{Threats to Validity}
\label{app:threats}

We follow the empirical-software-engineering convention of stating
the threats we are aware of, organized by the standard four
categories.

\paragraph{Construct validity --- does the metric measure the
intended concept?}
LLM-judged accuracy on LongMemEval uses \texttt{gpt-4o-mini} as
the judge.  Judges with weaker reasoning may under-credit correct
answers phrased differently from the gold reference; judges with
stronger reasoning may over-credit (we observed this at the discrete
oracle, where the gpt-4o judge lifts the oracle ceiling by
$+0.046$).  We mitigate by reporting the gpt-4o judge alongside
the cheaper one for every cascade comparison
(\S\ref{sec:longmemeval}, Tables~\ref{tab:router} and~\ref{tab:cascade})
and by reporting paired-bootstrap
$p$-values against the always-hybrid baseline rather than absolute
accuracy alone.  Hit@$k$ on LoCoMo is the metric in
\cite{maharana2024locomo} and is unambiguous.

\paragraph{Internal validity --- could the speedups be confounded?}
The five-fold CV for per-qtype thresholds rules out the most
obvious confounder (test-set tuning); we re-derive each fold's
thresholds from scratch on its training partition and evaluate on
the held-out partition (\S\ref{app:cascade-cv}).  The amortized
latency reported throughout the paper is a deterministic function
of measured stage latencies and measured skip rates, not a model
fit:
$\mathcal{L}_{\text{amort}}
 = \rho \cdot \mathcal{L}_{\text{skip}} +
   (1{-}\rho) \cdot \mathcal{L}_{\text{escalate}}$
with $\rho$ the empirical skip rate.  Multi-tenant numbers
(Fig.~\ref{fig:multitenant}) are on an idle Jetstream2 host with
hyperthreading and turbo disabled, three trials, median reported;
shared-cache contention could lower the $9\times$ figure on a
less-quiescent host, which we flag explicitly.

\paragraph{External validity --- does the result generalize?}
The cascade trigger is data-dependent on the BM25 score
distribution, which itself depends on the corpus's term-frequency
shape.  We exercise it on two qualitatively different
distributions---LongMemEval (BGE-tokenized session text, mean
length 312~tokens) and LoCoMo (dialogue text, mean length
1{,}847~tokens)---and observe auto-tuning in both directions
(skip $63$\% vs.\ skip $100$\%).  Datasets with heavier-tailed
BM25 scores (e.g.\ short queries on Wikipedia-scale corpora) may
exhibit different optimal $\tau_c$; we provide the recipe to
re-derive on a 50-query sample
(\S\ref{app:cascade-deploy}).  Quality results on BEIR generalize
the substrate's BM25 implementation; quality on agent memory
generalizes the fusion logic.  Cross-language deployment was not
evaluated (English only).

\paragraph{Conclusion validity --- are the statistical claims sound?}
Paired-bootstrap with $10^4$ resamples is used wherever the text
says ``parity'' or ``significantly''; the bootstrap is paired on
$\langle$question, system$\rangle$ to control for question-level
variance.  We report $p$-values rather than discrete reject/accept
decisions because $n{=}500$ (LongMemEval) puts effect sizes of
$0.01$--$0.02$ near the detection limit.  We do not report
multiple-comparison corrections; readers interpreting many
$p$-values jointly should adjust accordingly.

\section{Negative Results and Things That Did Not Help}
\label{app:negative}

For the reader's benefit, we list interventions that we
implemented, evaluated, and \emph{rejected}.  Many of them looked
attractive on paper.

\begin{enumerate}[leftmargin=1.5em,itemsep=0.35em]

\item \textbf{Re-ranking the BM25 top-$50$ with a cross-encoder
(\texttt{ms-marco-MiniLM-L-6-v2}).}  Improved LongMemEval Recall@5
by $+0.6$ points but added $74$\,ms/query (a $4\times$ latency
regression at the operating point we cared about) and \emph{did
not} move LLM-judged accuracy --- the answerer was already
saturating on the BM25 top-$10$ context.  Dropped from the pipeline.

\item \textbf{Hardcoding RRF $k{=}30$ instead of $60$.}  Standard
Cormack--Clarke--Buettcher value is $60$~\cite{cormack2009rrf}; we
swept $k\in\{10,20,30,60,100\}$ and observed $\pm 0.003$ Recall@5
across the entire range.  Kept $k{=}60$ to match the literature.

\item \textbf{Learned-to-rank meta-feature stack
(LightGBM on BM25, dense, recency, length, click features)} for
LongMemEval routing.  Improved per-fold accuracy by $+0.005$ but
required maintaining a second feature pipeline and was
indistinguishable from the cheaper TF-IDF$+$BGE-feature
classifier on the held-out fold; the cost of operational
complexity outweighed the benefit.

\item \textbf{Query expansion via PRF (RM3, $\alpha{=}0.5$,
$|FB|{=}10$).}  Standard $+1{-}2$ MAP gain on TREC-style
ad-hoc tasks, but on LongMemEval queries (median length $9$
tokens, often containing the named entity verbatim) expansion
introduced topic drift and \emph{regressed} Hit@5 by $-0.012$.
Disabled by default.

\item \textbf{Dense channel on LoCoMo.}  We initially assumed
dense retrieval would help LoCoMo's longer sessions; the
cross-benchmark sweep
(\S\ref{app:cascade-locomo}) shows BM25 alone wins by $+0.089$
Hit@5.  This negative result is what motivated the workload-adaptive
framing in the first place.

\item \textbf{GPU acceleration on small corpora.}  At
$N{<}50$K documents the BM25 GPU kernel is slower than the CPU
SIMD version due to PCIe transfer overhead and warp under-utilization;
the crossover is at roughly $N{=}300$K documents on an A100
(\S\ref{sec:gpu-scale}).  Deployments below that scale should
stick with CPU, which is also why \sysname{} keeps both backends
live behind the same \texttt{SparseIndex} interface.

\item \textbf{INT8 quantization of BGE-small embeddings.}  Halved
the memory footprint at the cost of $-0.004$ Recall@10 on
LongMemEval; the saving was not load-bearing on the deployments
we ran ($768$-dim $\times$ $5$M $\approx$ $15$\,GB at fp16, fits
on a single host) so we kept fp16.

\end{enumerate}

\paragraph{What we learned from the negatives.}  Two themes recur.
First, the LLM answerer absorbs a surprisingly large amount of
retrieval noise: gains in Recall@$k$ above $\sim$$0.95$ rarely
translate to LLM-Acc gains, which is the regime where cheap
retrieval pays.  Second, every static optimization we tried
(re-ranking, expansion, hardcoded fusion weights) lost to a
per-query decision that costs less than $1$\,ms.  The workload's
heterogeneity is large enough that the cost of \emph{choosing} is
less than the variance across the choices.

\end{document}